\documentclass[11pt]{amsart}

\usepackage[utf8]{inputenc}
\usepackage[T1]{fontenc}	% accents
\usepackage[french,english]{babel}

\usepackage{lmodern}
\usepackage{newtxtext}

\usepackage{enumerate}
\usepackage{enumitem}						% Enumerate package
\usepackage{marginnote}

\usepackage[top=3cm, bottom=3cm, left=3.6cm, right=3.6cm]{geometry}

\usepackage{graphicx}	% insertion des images / graphiques
\usepackage{tikz}
\usetikzlibrary{patterns}

\usepackage{placeins}
\usepackage{float} % pour mettre H dans les figures

\usepackage{color}
\definecolor{vertFonce}	{rgb}{0,0.5,0}
\definecolor{numLignes}	{rgb}{0.17,0.57,0.7}	%{43,145,175}
\definecolor{gris}		{rgb}{0.5,0.5,0.5}
\definecolor{grisFonce}	{rgb}{0.2,0.2,0.2}
\definecolor{orange}	{rgb}{1,0.65,0.31}		%{255,167,79}
\definecolor{orangeFonce}{rgb}{1,0.4,0}
\definecolor{bleuFonce}	{rgb}{0,0,0.4}
\definecolor{rougeFonce}{rgb}{0.3,0,0}
\definecolor{rougeWord}	{rgb}{0.5,0,0}
\definecolor{vertClair}	{rgb}{0.8,1,0.8}
\definecolor{rougeClair}{rgb}{1,0.5,0.5}
\definecolor{violet}	{rgb}{0.5,0,0.5}

\usepackage{pict2e}
\usepackage{multido}

\usepackage{amsfonts,amssymb,amsthm,amsmath}
\usepackage{amsaddr}		% author address (option [foot])
\usepackage{dsfont}			% \mathds symbols
\usepackage{mathrsfs}
\usepackage{yfonts}

\usepackage{mathtools} % to define paired delimiters

\theoremstyle{plain}
\newtheorem{thm}{Theorem}[section]
\newtheorem{lem}[thm]{Lemma}
\newtheorem{cor}[thm]{Corollary}
\newtheorem{prop}[thm]{Proposition}

\theoremstyle{definition}

\newtheorem{remark}{Remark}[subsection]

\newenvironment{system*}{%
	\equation\nonumber\left\{\ \begin{aligned}
}{%
	\end{aligned} \right. \endequation%
}

\usepackage[plainpages=false, colorlinks, linkcolor=bleuFonce, citecolor=rougeFonce, urlcolor=vertFonce, breaklinks]{hyperref}

\newcommand		{\N}		{\mathbb N}			% naturels {0,1,2,3,...}
\newcommand		{\RR}		{\mathbb R}			% réels
\newcommand		{\R}		{\RR}
\newcommand		{\Rd}		{\R^d}
\newcommand		{\Rdd}		{\R^{2d}}
\newcommand		{\hd}		{h^d}

\newcommand		{\CC}		{\mathbb C}			% complexes
\renewcommand	{\SS}		{\mathds S}			% sphère unité
\newcommand		{\cB}		{\mathcal B}		% Quantum Besov
\newcommand		{\cC}		{\mathcal C}
\newcommand		{\cD}		{\mathcal D}

\newcommand		{\cH}		{\mathcal H}		% Hilbert

\renewcommand	{\L}		{\mathcal L}		% Semiclassical schatten or law random var
\newcommand		{\cM}		{\mathcal M}		% Mesures bornées
\newcommand		{\cW}		{\mathcal W}		% Quantum Sobolev

\newcommand		{\scE}		{\mathscr E}

\newcommand		\sfA		{\mathsf A}
\newcommand		\sfB		{\mathsf B}

\newcommand		\sfL		{\mathsf L}			% Opérateur linéaire dissipatif

\newcommand		\sfT		{\mathsf T}			% Opérateur de transport

\newcommand		{\lt}			{\left}				%
\newcommand		{\rt}			{\right}			%
\renewcommand	{\(}			{\lt(}
\renewcommand	{\)}			{\rt)}

\newcommand		{\bangle}[1]	{\lt\langle #1\rt\rangle}
\newcommand		{\weight}[1]	{\bangle{#1}}	% <x>
\newcommand		{\scalar}[2]	{\bangle{#1,#2}}

\newcommand		{\com}[1]		{\lt[{#1}\rt]}		% commutator

\newcommand		{\n}[1]			{\lt|{#1}\rt|}
\newcommand		{\nrm}[1]		{\lt\|{#1}\rt\|}
\newcommand		{\snrm}[1]		{\lVert #1\rVert}

\newcommand		{\Nrm}[2]		{\nrm{#1}_{#2}}
\newcommand		{\sNrm}[2]		{\snrm{#1}_{#2}}

\newcommand		{\indic}	{\mathds{1}}		% indicatrice
\renewcommand		{\d}		{\mathop{}\!\mathrm{d}}		% différential

\newcommand			{\Dx}		{\nabla_x}
\newcommand			{\Dv}		{\nabla_\xi}
\newcommand			{\id}		{\mathbf{1}}		% Identity operator

\DeclareMathOperator{\tr}		{Tr}				% Trace

\DeclareMathOperator{\ran}		{ran}
\DeclareMathOperator{\dist}		{dist}

\newcommand		{\Tr}[1]		{\tr\!\(#1 \)}		% Trace

\newcommand		{\Dist}[1]		{\dist\!\(#1 \)}

\newcommand		{\intd}			{\int_{\Rd}}
\newcommand		{\intdd}		{\int_{\R^{2d}}}
\newcommand		{\iintd}		{\iint_{\R^{2d}}}

\newcommand		{\ii}			{\mathrm{i}}	% i\in\N (et pas i^2 = -1)
\newcommand		{\loc}			{\mathrm{loc}}

\newcommand		{\eps}			{\varepsilon}
\newcommand		{\Eps}			{\mathcal{E}}

\usepackage{braket}
\newcommand		{\Inprod}[2]	{\Braket{#1 | #2}}

\newcommand		{\op}		{{\boldsymbol{\rho}}}	% density operator

\newcommand		{\opb}		{{\boldsymbol{b}}}

\newcommand		{\opgam}	{{\boldsymbol{\gamma}}}	% quantum coupling
\newcommand		{\opc}		{{\boldsymbol{c}}}

\newcommand		{\tildop}		{\,\tilde{\!\op}}	% Wick quantization

\newcommand		{\opp}		{\boldsymbol{p}}
\newcommand		{\opProj}	{\mathbf{P}}
\newcommand		{\opz}		{\mathbf{z}}%{\boldsymbol{z}}

\newcommand		{\Dh}		{\boldsymbol{\nabla}}	% quantum gradient
\newcommand		{\Dhx}[1]	{\Dh_{\!x} #1}			% quantum d_x
\newcommand		{\Dhv}[1]	{\Dh_{\!\xi} #1}		% quantum d_xi

\newcommand{\om}	{\omega_d} % Size of the unit ball in Rd
\newcommand{\vr}{\varrho}
\newcommand{\D}{\mathsf D}
\newcommand{\sL}{\mathsf L}
\newcommand{\sM}{\mathsf M}
\newcommand{\sC}{\mathsf C}

\newcommand{\opHF}	{\opgam_{\textnormal{H}}}
\newcommand{\rhoHF}	{\varrho_{\textnormal{H}}}

\newcommand{\rhoTF}	{\varrho_{\textnormal{TF}}}
\newcommand{\fTF}	{f_{\textnormal{TF}}}

\newcommand{\CHS}	{\cC_{\textnormal{HS}}}
\newcommand{\CTr}	{\cC_{\textnormal{Tr}}}
\newcommand{\CLT}	{C_{\textnormal{LT}}}

\newcommand{\opq}	{{\boldsymbol{q}}}

\title[\textsc{Commutator Estimates and Weyl's Law with Singular Potentials}]{\Large Commutator Estimates and Quantitative Local Weyl's Law for Schr\"odinger Operators with Non-Smooth Potentials}
\author[\textsc{E.~C\'ardenas}]{\large\textsc{Esteban C\'ardenas}}
\address{\vspace{-10pt}Department of Mathematics, University of Texas at Austin, Austin TX, 78712, USA}
\email{eacardenas@utexas.edu}

\author[\textsc{L.~Lafleche}]{\vspace{-10pt}\large\textsc{Laurent Lafleche}}
\address{\vspace{-10pt}Unit\'e de Math\'ematiques pures et appliqu\'ees, \\\ \'Ecole Normale Supérieure de Lyon, 69364 Lyon, France}
\email{laurent.lafleche@ens-lyon.fr}


\begin{document}

\begin{abstract}
	We analyze semiclassical Schr\"odinger operators with potentials of class $C^{1,1/2}$ and establish commutator estimates for the associated projection operators in Schatten norms. These are then applied to prove quantitative versions of the local and phase space Weyl laws in $L^p$ spaces. We study both non-interacting, and interacting particle systems. In particular, we are able to treat the case of the minimizers of the Hartree energy in the case of repulsive singular pair interactions such as the Coulomb potential.
\end{abstract}

\begingroup
\def\uppercasenonmath#1{} % this disables uppercasing title
\let\MakeUppercase\relax % this disables uppercasing authors
\maketitle
\endgroup

%\textbf{Keywords}: Hartree equation, Hartree-Fock equation, Vlasov equation, Coulomb interaction, gravitational interaction, semiclassical limit.

%----------  Table of Contents  ----------
\renewcommand{\contentsname}{\centerline{Table of Contents}}
\setcounter{tocdepth}{1}	% profondeur du commentaire
\tableofcontents
%\addcontentsline{toc}{chapter}{Table of Contents}

%% ********************  Contenu  ********************

%----------  Introduction  ----------
\section{Introduction}

	In this article, we study several asymptotic results related to the spectral projection
	\begin{equation*}
		\opgam := \indic_{H\leq 0}
	\end{equation*}
	on the eigenspace of the negative eigenvalues of the Schr\"odinger operator on $L^2(\Rd)$
	\begin{equation}\label{eq:H}
		H = -\hbar^2 \Delta + V(x)
	\end{equation}
	where $V :\Rd \rightarrow \R$ is a function playing the role of the potential energy, identified with an operator of multiplication by $V$. Here, $\indic_{H \leq 0} = \indic_{(-\infty, 0]}(H)$ is defined via functional calculus. We analyze the so-called \textit{semiclassical limit} in which the Planck constant 
	\begin{equation*}
		h = 2\pi\hbar
	\end{equation*}
	converges to $0$. In particular, we are interested in the case when the potential $V$ is not smooth but of class $C^{1,\alpha}$. Such potentials naturally arise in effective one-body problems of particle systems interacting via mean-field forces, which are induced by singular pair potentials, such as the Coulomb potential.
	
\subsection{Weyl's laws}
	
	A fundamental result in the semiclassical limit of Schr\"odinger operators is Weyl's law, which states that as $\hbar \to 0$
	\begin{equation}\label{eq:Weyl_law}
		N := \Tr{\indic_{H\leq 0}} = h^{-d} \int_{\n{\xi}^2 + V(x) \leq 0} \d x \d \xi
		+ o(h^{-d}) \, , 
	\end{equation}
	see e.g.~\cite[Theorem~A.2.1]{shubin_pseudodifferential_2001}. In particular, much activity has been devoted to understanding the rate of convergence in terms of $\hbar$, see e.g.~\cite{hormander_asymptotic_1979}, and~\cite{robert_comportement_1981}, where the optimal rate was obtained for smooth potentials. For non-smooth potentials, it was recently proved in~\cite{mikkelsen_sharp_2023} that for potentials in the class $C^{1,\alpha}$ for $\alpha \geq \frac{1}{2}$ and $d \geq 3$, it holds for all $\hbar \in (0,1)$ that
	\begin{equation}\label{eq:Mikkelsen}
		\n{N\hd - \int_{\n{\xi}^2 + V(x) \leq 0} \d x \d p} \leq 
		\cC _0 \,\hbar
	\end{equation}
	for a constant $\cC_0 > 0 $ depending only on $V$. The rate of convergence in~\eqref{eq:Mikkelsen} is optimal, and its implications are crucial in the rest of the article. Indeed, in our setting it implies the validity of the \textit{local eigenvalue estimate}
	\begin{equation}\label{eq:LEE}
		\hd \tr \indic_{[a,b]}(H) \leq \cC_1 \(\n{b - a} + \hbar\)
	\end{equation}
	for all appropriate values of $a < b$ around zero. Here again $\cC_1 > 0$ is a distinguished constant, independent of $\hbar$.
	
	While Weyl's law describes the total number of negative eigenvalues of the Schr\"odinger operator $H = -\hbar^2 \Delta + V$, one may also wonder about the average space distribution of its eigenfunctions. To be more precise, if $\op$ is a density operator, i.e. a non-negative trace-class operator on $L^2(\Rd)$, one can introduce its position density\footnote{The position density can be more generally defined in the weak sense by the formula $\intd \vr_\op\,\varphi = \hd \Tr{\op\,\varphi}$ for any $\varphi\in C^\infty_c(\Rd)$.}
	\begin{equation}\label{eq:density}
		\vr_\op(x) := \hd \op(x,x)\, , 
	\end{equation} 
	where $\op(x,y)$ denotes the integral kernel of the operator $\op$, and study its limit. When $\op = \opgam$, it is known under general assumptions that the following limit is true, and is known in the literature as the \textit{local Weyl law} (see e.g. \cite{frank_weyls_2023})
	\begin{equation}\label{eq:Weyl_local}
		\lim_{\hbar \to 0^+} \vr_\opgam =\om \,V_-^{d/2} ,
	\end{equation}
	where $\omega_d = \frac{\pi^{d/2}}{\Gamma(1+d/2)}$ is the volume of the unit ball in $\Rd$. Again, for smooth potentials, one can obtain rates of convergence in weak norms, see e.g.~\cite{chazarain_spectre_1980, ivrii_semiclassical_1992, sobolev_quasi-classical_1995, deleporte_universality_2024}.
	
	These results can be seen as rigorous examples of the quantum-classical correspondence. More generally, one can consider the localization in phase space of the above limit, which is sometimes called convergence of states~\cite{fournais_semi-classical_2018}. For this purpose, one can introduce the Wigner transform of the density operator $\op$, which is the function of the phase space defined by
	\begin{equation}\label{eq:Wigner}
		f_\op(x,\xi) = \intd e^{-i\,y\cdot\xi/\hbar} \,\op(x+\tfrac{y}{2},x-\tfrac{y}{2})\d y
	\end{equation}
	for $(x, \xi)\in \Rd \times \Rd$, where $\xi$ represents the momentum. This is the inverse operation of the Weyl quantization, which to a real-valued phase space density function $f$ associates a self-adjoint operator $\op_f$ on $L^2(\Rd)$ with integral kernel
	\begin{equation}\label{eq:Weyl}
		\op_f(x ,y) = \intd e^{- 2 i \pi \(x - y\) \cdot \xi}\, f \big(\tfrac{x + y}{2}, h\xi\big) \d \xi \, . 
	\end{equation}
	The \textit{phase space Weyl law} then tells that
	\begin{equation}\label{eq:Weyl_phase_space}
		\lim_{\hbar \to 0^+} f_\opgam = \indic_{\n{\xi}^2 + V(x) \leq 0}\, .
	\end{equation}
	This limit has been studied in the physics literature \cite{balazs_quantum_1973, dean_universal_2015, dean_wigner_2018}.
	
	One may observe that each limit is in a sense stronger than the previous one due to the following relations
	\begin{equation*}
		\intd f_\op \d \xi = \vr_\op \quad \text{ and } \quad \intd \vr_\op = \hd \Tr{\op},
	\end{equation*}
	for any sufficiently nice density operator $\op$, and the fact that $\vr_f = \omega_d\,V_-^{d/2}$, where, in the classical case, the position density of a phase space distribution $f : \Rdd\to\R$ is defined by
	\begin{equation}\label{eq:density_classical}
		\vr_f(x) = \intd f(x,\xi) \d\xi\,.
	\end{equation}
	One of the questions that we are interested in this article is the convergence rate of the limits~\eqref{eq:Weyl_local} and~\eqref{eq:Weyl_phase_space}.
	
\subsection{Commutator estimates and semiclassical regularity}\label{sec:intro_commutator}

	% I replaced "necessary" by "useful", we are not proving that the rate of the limit is equivalent to regularity
	It is not surprising that in order to understand the rate at which the limits~\eqref{eq:Weyl_local} and \eqref{eq:Weyl_phase_space} hold, it is useful to understand the regularity properties of the functions $\vr_\op$ and $f_\op$ that are \textit{uniform in $\hbar$}. In connection with this question is the relatively recent need for commutator estimates of the form
	\begin{equation}\label{eq:comms}
		\hd \Tr{\n{\com{x,\op}}^p} \leq C \,\hbar 
		\quad
		\text{ and }
		\quad 
		\hd \Tr{\n{\com{\hbar\nabla, \op}}^p} \leq C\,\hbar
	\end{equation}
	for $1 \leq p < \infty$, with $C$ independent of $\hbar$. To the authors' knowledge, such estimates appeared first for $p=1$, as requirements on the initial data in the quantitative derivation of the Hartree--Fock equation from dynamics of fermions~\cite{benedikter_mean-field_2014}. Here, $\op$ is a rank $N$ orthogonal projection on $L^2(\Rd)$, corresponding to the reduced one-particle density matrix of a Slater determinant in semiclassical scaling $\hbar = N^{-1/d}$. Since then, these estimates have appeared multiple times in the literature. For instance, in~\cite{benedikter_mean-field_2014-1} the authors extended their results to fermions with relativistic dispersion, whereas in~\cite{benedikter_mean-field_2016-1} they extended them to the dynamics of mixed quasi-free states. Similar variants have been used in~\cite{fresta_effective_2023} to study extended system of fermions, and in~\cite{cardenas_effective_2023} to study the dynamics of Bose--Fermi mixtures. The estimates for $p=2$, on the other hand, were employed in~\cite{marcantoni_dynamics_2024} to study the dynamics of fermions with non-zero pairing, and also in~\cite{cardenas_norm_2024} to understand the strong convergence of ground states of $N$-body problems. In~\cite{lafleche_quantum_2023}, it is proved that they also give a bound on the self-distance in the quantum Wasserstein pseudo-metrics introduced in~\cite{golse_mean_2016, golse_schrodinger_2017}, which allows to understand the rate of convergence in the mean-field and semiclassical limits in these two works, as well as in other works using these pseudo-metrics such as \cite{lafleche_propagation_2019, cardenas_effective_2023, chong_semiclassical_2024}.
	
	As mentioned above, the estimates~\eqref{eq:comms} are usually imposed on the initial datum. In practice, verifying such estimates for Slater determinants is not an easy task. The first example was given in~\cite{benedikter_mean-field_2014} for the free Fermi gas on the torus. Later, the authors of~\cite{fournais_optimal_2020} proved their validity in the case $p = 1$ for Schr\"odinger operators $\opgam = \indic_{H \leq 0}$, with $V$ is smooth, see also~\cite{deleporte_universality_2024} for similar estimates when $p=2$. The techniques used in~\cite{fournais_optimal_2020} rely on the use of pseudo-differential operator techniques and, in particular, knowledge of the optimal Weyl law~\eqref{eq:Mikkelsen}. Additionally, for the harmonic oscillator~\cite{benedikter_effective_2022, lafleche_optimal_2024} the estimates were verified directly using creation and annihilation operator calculus. While we also rely on the validity of~\eqref{eq:Mikkelsen}, our proof does not use additional pseudo-differential operator techniques. In particular, with our methods we are able to study the validity of~\eqref{eq:comms} for the case $V \in C^{1,\alpha}_{\rm loc}$ with $\alpha \geq \frac{1}{2}$, and for the ground state of an interacting Fermi system in the Hartree approximation with singular interactions.
	
	The link between commutators estimates and uniform-in-$\hbar$ regularity follows from the correspondence principle, which tells us that the quantum analogue of the gradients in the phase space should be
	\begin{equation}\label{eq:quantum_gradients}
		\Dhx\op := \com{\nabla,\op} \quad \text{ and } \quad \Dhv\op := \frac{1}{i\hbar} \com{x,\op}
	\end{equation}
	and indeed, these correspond to taking the gradient of the Wigner transform in the sense that one verifies that
	\begin{equation*}
		\Dx f_\op = f_{\Dhx \op} \quad \text{ and } \quad \Dv f_\op = f_{\Dhv \op}\,.
	\end{equation*}
	In particular, $\nabla \vr_\op = \vr_{\, \Dhx\op}$, and the commutator estimates~\eqref{eq:comms} with $p=1$ imply uniform-in-$\hbar$ $W^{1,1}$ regularity for the position density $\vr_\op$, that is $\Nrm{\nabla \vr_\op}{L^1}$ is bounded independently of $\hbar$. In terms of phase space regularity, one can first notice that the analogue of the phase space integral of a density is given by $\hd \,\tr$, which justifies the use of scaled Schatten norms for $p\geq 1$,
	\begin{equation}\label{eq:Schatten}
		\Nrm{\op}{\L^p} := \(\hd\Tr{\n{\op}^p}\)^{1/p},
	\end{equation}
	with $\Nrm{\op}{\L^\infty}$ denoting the operator norm. The Wigner transform is then an isomorphism from $\L^2$ to $L^2(\Rdd)$, that is $\Nrm{f_\op}{L^2} = \Nrm{\op}{\L^2}$. We infer from this that the commutator estimates~\eqref{eq:comms} with $p=2$ imply bounds of the form $\Nrm{\nabla f_\op}{L^2} \leq C / \sqrt{\hbar}$. The fact that $f_\op$ is not in $H^1$ uniformly in $\hbar$ should not come as a surprise when $\op = \indic_{H \leq 0}$, since the phase space Weyl law \eqref{eq:Weyl_phase_space} states that $f_\op$ converges to a characteristic function of the phase space. It is however proved in \cite{lafleche_optimal_2024} that the commutator estimates~\eqref{eq:comms} with $p=1$ implies that $f_\op \in H^s(\Rdd)$ uniformly in $\hbar$ for any $s< 1/2$, and that this regularity order is optimal.
	
\subsection{Interacting particles}

	In this article, we consider the \textit{Hartree} energy functional (sometimes also called reduced Hartree--Fock) defined for each density operator $\op$ on $L^2(\Rd)$ by
	\begin{equation}\label{eq:Hartree_enregy}
		\scE_\op := \hd \Tr{\(-\hbar^2\Delta + U(x)\)\op} + \frac{1}{2} \iintd K(x-y)\, \vr_\op(x) \,\vr_\op(y) \d x \d y \, , 
	\end{equation}
	where $U$ is a trapping potential and $K$ a repulsive singular pair interaction potential. Here, we look at a grand canonical ensemble of fermions. That is, the functional is minimized over all density operators verifying $ 0 \leq \op \leq \id$ without fixing the trace, where $\id$ denotes the identity operator. We follow the convention that the chemical potential $\mu>0$ is included in $U(x)$. The operator bound is a manifestation of the \textit{Pauli Exclusion Principle}: no more than two fermions can occupy the same quantum state.

	We are interested in both the regularity properties, as well as the \textit{quantitative} semiclassical limit of the minimizers of the Hartree functional, here and in the sequel denoted by $\opHF$. Recently, the Weyl law~\eqref{eq:Weyl_law} and its local version~\eqref{eq:Weyl_local} for the minimizers $\opHF$ have been studied in \cite{nguyen_weyl_2024}. In this context, the limit of the densities $\vr_{\opHF}$ is determined via Thomas--Fermi theory, and solves the fixed point equation 
	\begin{equation}\label{eq:rhoTF}
		\rhoTF =\om \(U + K*\rhoTF\)_-^{d/2} . 
	\end{equation}
	In particular, it has been proven that the Weyl law and its local version hold
	\begin{equation}\label{eq:local}
		\lim_{\hbar\to 0^+} \intd \vr_{\opHF}(x) \d x = \intd \rhoTF(x) \d x 
		\quad
		\text{ and }
		\quad 
		\lim_{\hbar \to 0^+} \vr_{\opHF} = \rhoTF
	\end{equation}
	where the last limit holds weakly in $L^1(\Rd)\cap L^{1 + d/2}(\Rd)$. In one of our main results, we give a quantitative $L^p$-version of the local Weyl law, based on the regularity that is inherited from the commutator estimates satisfied by $\opHF$.

	For $N$-particle systems, establishing the limit~\eqref{eq:local} as $N \rightarrow \infty$ of the position density of the ground state $\Psi_N$ goes back to~\cite{lieb_thomasfermi_1977}, who analyzed Coulomb systems. The problem has been recently revisited in a modern perspective by~\cite{fournais_semi-classical_2018}, where both the local Weyl law and the phase space Weyl law are established for a wide range of potentials. These convergence results were extended in~\cite{lewin_semi-classical_2019} to states at positive temperature, in~\cite{fournais_semi-classical_2020} to systems with magnetic fields, and in~\cite{girardot_semiclassical_2021} to anyons. At zero temperature, the convergence of states was improved from weak to strong convergence in~\cite{cardenas_norm_2024}. See also the approach via Gamma-convergence in~\cite{gottschling_convergence_2018}.

\subsection{Main results}
 
\subsubsection{Commutator estimates and Weyl's law}

	We consider Schr\"odinger operators with $C^{1,\alpha}$ potentials which grow at infinity. % although with some exponential growth restrictions.
	Our first main result on commutator estimates is the following theorem. Here and in the sequel $\nabla^2$ denotes the Hessian.
	\begin{thm}\label{thm:commutator:linear}
		For $\hbar\in(0,1)$, let $\opgam = \indic_{H \leq 0}$ where $H = -\hbar^2 \Delta + V$ in $d = 3$. Assume that for some $\beta \geq 0$
		\begin{equation*}\tag{H1}\label{hyp:V}
			V\in C^{1,1/2}_\loc(\R^3) \, , \qquad \lim_{\n{x} \to \infty}V(x) = \infty \, , \qquad e^{-\beta \n{x}} \,\nabla V \in L^\infty(\R^3)
		\end{equation*}
		uniformly in $\hbar$. Then, for any $p\in [1,\infty)$, there exists a constant $C>0$ independent of $\hbar$ such that
		\begin{align*}
			\hd \Tr{\n{\com{x,\opgam}}^p} \leq C\, \hbar
			\quad
			\text{ and }
			\quad 
			\hd \Tr{\n{\com{\hbar\nabla,\opgam}}^p} \leq
			\begin{cases}
				C\, \hbar &\text{ if } p \geq 2
				\\
				C\,\hbar \n{\ln\hbar}^{2 - p}&\text{ if } p\leq 2\,.
			\end{cases}
		\end{align*}
		Additionally, if $e^{-\beta\n{x}}\,\nabla^2 V\in L^\infty(\R^3)$ uniformly in $\hbar$, then the $\n{\ln \hbar}$ factor can be removed. That is, for any $1 \leq p < \infty$ there exists a constant $C>0$ independent of $\hbar$ such that
		\begin{align*}
			\hd \Tr{\n{\com{\hbar\nabla,\opgam}}^p} &\leq C \,\hbar \, . 
		\end{align*}
	\end{thm}

	\begin{remark}
		It follows from the proof of Theorem~\ref{thm:commutator:linear} that we can give explicit values of the constants $C>0$ in terms of $V$ and the constant $\cC_1>0$ from the local eigenvalue estimate~\eqref{eq:LEE}. 
	\end{remark}

	\begin{remark}
		All the bounds that contain no $\n{\ln\hbar}$ in our theorem, such as the Hilbert--Schmidt bounds (the case $p =2$) are optimal in terms of $\hbar$. Indeed, this can be checked in the case of the harmonic oscillator $V(x) = \n{x}^2 - \mu $ (see e.g.~\cite{benedikter_effective_2022, lafleche_optimal_2024}).
	\end{remark}
	
	\begin{remark}
		Bounds on commutators with $x$ allow more generally to obtain bounds on commutators with operators of multiplication by Lipschitz continuous functions. Indeed, by~\cite[Inequality~(14)]{potapov_operator-lipschitz_2011}, for any $p\in(1,\infty)$, there exists $c_p > 0$ such that for any $u \in W^{1,\infty}(\Rd)$
		\begin{equation*}
			\Nrm{\com{u(x),\opgam}}{\L^p} \leq c_p \Nrm{\nabla u}{L^\infty} \Nrm{\com{x,\opgam}}{\L^p} .
		\end{equation*}
		Moreover, if $u$ is locally Lipschitz but growing at infinity, then one still obtains a similar bound using the Agmon inequality (see Proposition~\ref{prop:Agmon}).
	\end{remark}
 
	In terms of the quantum gradients defined in Equation~\eqref{eq:quantum_gradients}, when $e^{-\beta \n{x}}\,\nabla^2 V \in L^\infty$, the estimates in Theorem~\ref{thm:commutator:linear} can be written for any $p\in[1,\infty]$
	\begin{equation*}
		\Nrm{\Dhv \opgam}{\L^p} \leq \frac{C}{\hbar^{1/p'}} \quad \text{ and } \quad \Nrm{\Dhx \opgam}{\L^p} \leq \frac{C}{\hbar^{1/p'}} \, ,
	\end{equation*}
	where $p' = \frac{p}{p-1}$, or even more compactly, defining $\Dh\op = (\Dhx\op,\Dhv\op)$ so that $\n{\Dh\op}^2 = \n{\Dhx\op}^2 + \n{\Dhv\op}^2$,
	\begin{equation}\label{eq:bound_Dh}
		\Nrm{\Dh\opgam}{\L^p} \leq \frac{C}{\hbar^{1/p'}} \, .
	\end{equation}
	The constant $C$ is independent of $\hbar$ and can be taken independent of $p$. Recall Definition~\eqref{eq:density}. Since $\nabla\vr_\opgam = \vr_{\Dhx\opgam}$, the case $p=1$ implies the following $W^{1,1}$ estimates
	\begin{cor}[Regularity of the density]
		\label{cor:reg:density}
		Assuming $V$ satisfies ~\eqref{hyp:V}, there exists a constant $C>0$ independent of $\hbar$ so that
		\begin{equation*}
			\Nrm{\nabla \vr_{\opgam}}{L^1} \leq C \n{\ln \hbar}. 
		\end{equation*}
		Additionally, if $e^{- \beta \n{x}}\, \nabla^2 V \in L^\infty (\R^3)$, the logarithm can be removed. That is 
		\begin{equation*}
			\Nrm{\nabla \vr_{\opgam}}{L^1} \leq C\,. 
		\end{equation*}
	\end{cor}
	As written in Section~\ref{sec:intro_commutator}, the divergence as $\hbar\to 0$ of the right-hand side of Inequality~\eqref{eq:bound_Dh} is coherent with the fact that nonzero characteristic functions are not smooth. As reviewed e.g. in~\cite{sickel_regularity_2021}, the correct spaces to measure the regularity of characteristic functions in $L^p$ with $p\in(1,\infty)$ are the Besov spaces $B^{1/p}_{p,\infty}$, which can be seen as a $L^p$ generalization of H\"older spaces of order $1/p$. More precisely, if the regularity order $s\in(0,1)$, the spaces $B^s_{p,\infty}(\Rdd)$ are Banach spaces with respect to the norm $\Nrm{f}{B^s_{p,\infty}} = \Nrm{f}{L^p} + \Nrm{f}{\dot B^s_{p,\infty}}$, where one defines for $f : \Rdd\to \CC$ the homogeneous seminorm
	\begin{equation*}
		\Nrm{f}{\dot B^s_{p,\infty}} := \sup_{z\in\Rdd\setminus\{0\}} \frac{\Nrm{f(\cdot-z)-f}{L^p}}{\n{z}^s}\, .
	\end{equation*}
	They verify for any $\epsilon >0$, $W^{s,p} \subset B^s_{p,\infty} \subset W^{s-\epsilon ,p}$, where $W^{s,p}$ are Sobolev spaces. Analogous semi-norms can be defined in the quantum setting. The phase space shift of an operator being given for $z_0 = (x_0,\xi_0)\in\Rdd$ by
	\begin{equation*}
		\sfT_{z_0} \op = \tau_{z_0}\, \op\, \tau_{-z_0} \quad \text{ where } \quad \tau_{z_0}\varphi(x) = e^{i\,\xi_0\cdot x/\hbar} \,\varphi(x-x_0)\, ,
	\end{equation*}
	one can define as in~\cite{lafleche_quantum_2024} the quantum homogeneous Besov seminorms
	\begin{equation*}
		\Nrm{\op}{\dot{\cB}^s_{p,\infty}} = \sup_{z\in\Rdd} \frac{\Nrm{\sfT_z\op - \op}{\L^p}}{\n{z}^s}.
	\end{equation*}
	They satisfy in particular for $p=2$, $\Nrm{f_\op}{\dot{B}^s_{2,\infty}(\Rdd)} = \Nrm{\op}{\dot{\cB}^s_{2,\infty}}$. As stated in the next theorem, one then recovers that these norms are bounded uniformly in $\hbar$, and the situation is even slightly better in the quantum case, where a regularization at scale $\hbar$ can be observed.

	\begin{thm}\label{thm:comm:Besov}
		Let $\hbar \in (0,1)$ and $V$ satisfy~\eqref{hyp:V} and assume $e^{-\beta\n{x}}\,\nabla^2 V\in L^\infty(\R^3)$. Then, there exists a constant $C$ independent of $\hbar$ such that for all $p\in[1,\infty]$ and $z\in\R^6$,
		\begin{equation*}
			\Nrm{\sfT_z\opgam-\opgam}{\L^p} \leq C \min\!\(\frac{\n{z}}{\hbar^{1/p'}}, \n{z}^{1/p},1\).
		\end{equation*}
		In particular, $\Nrm{\opgam}{\dot{\cB}^{1/p}_{p,\infty}}$ is bounded uniformly in $\hbar$ for any $p\in[1,\infty]$, and for $p=2$, one obtains that $f_\opgam \in B^{1/2}_{2,\infty}(\R^6)$ and so in $H^s(\R^6)$ with $s<1/2$ uniformly in $\hbar$. As in the previous theorems, if $V$ is only in $C^{1,1/2}_\loc(\R^3)$, then the estimates remain true up to the multiplication by $\n{\ln \hbar}^{2/p-1}$ when $p\leq 2$.
	\end{thm}

	\begin{remark}
		Taking $z = (0,\hbar\,\xi)$ in the above theorem and using the fact that $\tau_z$ is unitary gives the following commutator estimate. For any $\xi\in\R^3$,
		\begin{equation*}
			\hd \Tr{\n{\com{e^{i\,\xi\cdot x},\opgam}}^p} \leq C \min(\n{\xi}^{p-1}\!, 1) \n{\xi} \hbar \, .
		\end{equation*}
	\end{remark}

	Now, let us state our results concerning the quantitative convergence of the densities in the semiclassical limit. Recall the definitions of the position densities $\vr_\opgam$ in \eqref{eq:density} and $\vr_f$ in~\eqref{eq:density_classical}, the Wigner function $f_\opgam$ in \eqref{eq:Wigner} and the Weyl quantization $\op_f$ in~\eqref{eq:Weyl}. We have the following theorem.

	\begin{thm}[Linear quantitative convergence]\label{thm:weyl:linear}
		Let $\hbar\in(0,1)$, $d=3$ and $\opgam = \indic_{H \leq 0}$ with $H= -\hbar^2\Delta + V$ and $f = \indic_{\n{\xi}^2+V\leq 0}$. Assuming $V$ satisfies~\eqref{hyp:V}, there is a constant $C>0$ independent of $\hbar$ such that
		\begin{align*}
			\Nrm{\vr_\opgam - \vr_f}{L^2} &\leq C\, \Big(1 + \Nrm{\vr_\opgam}{\dot{B}^{1/2}_{2,\infty}}^{2/3}\Big)\, \hbar^{1/3}
			\\
			\Nrm{\vr_\opgam - \vr_f}{L^1} &\leq C \(1 + \Nrm{\nabla \vr_\opgam}{L^1}^{1/2}\) \hbar^{1/2}
			\\
			\Nrm{f_\opgam - f}{L^2} &\leq C
			\(1 + \Nrm{\nabla \vr_\opgam}{L^1}^{1/2} \)
			\hbar^{1/4}
			\\
			\sNrm{\opgam - \op_f}{\L^1} &\leq C \(1 + \Nrm{\Dh \opgam}{\L^1}\) \hbar^{1/2} \, .
		\end{align*}
		Additionally, if $V \in W^{2,1} (\Omega)$ for an open set $\Omega$ containing $\{V < 0 \}$, then there exists $C>0$ such that 
		\begin{equation*}
			\Nrm{f_\opgam - f}{L^2} \leq C \,\hbar^{1/4}.
		\end{equation*} 
		Here, $\L^1$ denotes the scaled trace norm defined in Equation~\eqref{eq:Schatten} and the values of the constants $C$ can be made explicit (see Proposition~\ref{prop:energy2:linear}).
	\end{thm}
	
	\begin{remark}
		The gradients $\Nrm{\Dh \opgam}{\L^1}$ and $\Nrm{\nabla \vr_{\opgam}}{L^1}$ grow at most logarithmically in $\hbar$ thanks to Theorem~\ref{thm:commutator:linear}. Note also that the Besov norm $\dot B_{2,\infty}^{1/2}$ can be controlled via interpolation with $\Nrm{\vr_{\opgam}}{L^\infty}$ and $\Nrm{\nabla \vr_{\opgam}}{L^1}$, see e.g. Remark~\ref{remark:besov}. In particular, if $e^{-\beta\n{x}}\,\nabla^2 V\in L^\infty(\R^3)$ uniformly in $\hbar$, then under the hypotheses of Theorem~\ref{thm:weyl:linear}, one obtains the existence of a constant $C$ independent of $\hbar$ such that
		\begin{equation*}
			\Nrm{\vr_\opgam - \vr_f}{L^2} \leq C\, \hbar^{1/3} \qquad \Nrm{\vr_\opgam - \vr_f}{L^1} \leq C \, \hbar^{1/2}
		\end{equation*}
		and
		\begin{equation*}
			\Nrm{f_\opgam - f}{L^2} \leq C \, 
			\hbar^{1/4} \qquad \sNrm{\opgam - \op_f}{\L^1} \leq C \, \hbar^{1/2} \, .
		\end{equation*}
	\end{remark}
	
	\begin{remark}
		In general, it is not expected that the trace-class norm of an operator $\op$ controls the $L^1$-norm of its Wigner transform $f_\op$. However, denoting $\widehat f$ the Fourier transform of a function $f : \Rdd \to \CC$, we have Groenewold's formula~\cite{groenewold_principles_1946}
		\begin{equation*}
			\widehat f_\op(\xi,\eta) = \hd \Tr{e^{-2i\pi\(x\cdot\xi + \eta\cdot\opp\)}\op}
		\end{equation*}
		with $\opp = -i\hbar\nabla$. It implies the bound in Fourier space
		\begin{equation}
			\Nrm{\widehat f_\opgam - \widehat f}{L^\infty} \leq C \(1 + \Nrm{\Dh \opgam}{\L^1}\) \hbar^{1/2} \, . 
		\end{equation}
	\end{remark}
	
	\begin{remark}
		Our bounds imply convergence in Wasserstein distance. Indeed, if $W_p(\vr_\opgam, \vr_f)$ denotes the Wasserstein distance of order $p$ between $\vr_\opgam$ and $\vr_f$, then for any $p\in [1,\infty)$ it holds (see e.g.~\cite[Proposition~7.10]{villani_topics_2003})
		\begin{equation}\label{eq:Wp_vs_L1}
			W_p(\vr_\opgam, \vr_f) \leq 2 \Nrm{\n{x}^p \(\vr_\opgam-\vr_f\)}{L^1}^{1/p}.
		\end{equation}
		Taking $R>0$ sufficiently large so that the ball $B_R$ of radius $R$ contains the support of $\vr_f$ and such that $\Nrm{\n{x}^p \vr_\opgam}{L^1(B_R^c)} \leq \sqrt{\hbar}$ by Agmon's estimates (Proposition~\ref{prop:Agmon}), it follows that the right-hand side of Equation~\eqref{eq:Wp_vs_L1} is bounded by $2 R \Nrm{\vr_\opgam-\vr_f}{L^1(B_R)}^{1/p} + 2\,\hbar^\frac{1}{2p}$ which gives
		\begin{equation*}
			W_p(\vr_\opgam, \vr_f) \leq C\,\hbar^\frac{1}{2p}.
		\end{equation*}
	\end{remark}

\subsubsection{The interacting case}

	In the Hartree case, we consider the minimizers of the functional $\scE_\op$ with an external trap $U(x)$ and a singular pair potential $K(x-y)$ which we assume to be repulsive. Our results can cover Coulomb singularities, at the expense of adding a logarithmic factor $\sqrt{\n{\ln \hbar}}$ in the estimates. Currently, we do not know if this factor is optimal and a feature of the theory, or a consequence of the proof. For potentials that are less singular than Coulomb, we obtain the desired optimal commutator estimates.

	Our main results covering commutator estimates for interacting systems reads as follows. 
	\begin{thm}[Non-linear commutator estimates]\label{thm:commutator_nonlinear}
		For $\hbar\in(0,1)$ and in $d = 3$, let $\opHF$ be a minimizer of the Hartree energy $\scE$ given by~\eqref{eq:Hartree_enregy}. Assume that the interaction potential $K: \R^3 \rightarrow \R$ is of the form
		\begin{equation*}\label{hyp:K}\tag{H2}
			K(x) = \kappa \n{x}^{-a} \quad \text{ with } \quad \kappa \geq 0
		\end{equation*}
		where $a\in[0,1]$, and that the external potential $U : \R^3 \rightarrow \R$ satisfies, for some $\beta\geq 0$
		\begin{equation*}\label{hyp:U}\tag{H3}
			\lim_{\n{x} \to \infty} U (x) = \infty \, , 
			\qquad e^{-\beta \n{x}} \, \nabla^2 U\, \in L^{\infty}(\R^3) \, . 
		\end{equation*}
		Let $p\in[1,\infty)$. Then, if $a \in (0,1)$ there exists $C>0$ independent of $\hbar$ such that
		\begin{equation*}
			\hd \Tr{\n{\com{x,\opHF}}^p} \leq C\, \hbar
			\quad \text{ and } \quad \hd \Tr{\n{\com{\hbar \nabla, \opHF}}^p} \leq C\, \hbar 
		\end{equation*}
		while if $a = 1$ there exists $C>0$ independent of $\hbar$ such that
		\begin{align*}
			\hd \Tr{\n{\com{x,\opHF}}^p} \leq C\, \hbar
			\quad
			\text{ and }
			\quad 
			\hd \Tr{\n{\com{\hbar\nabla,\opHF}}^p} \leq
			\begin{cases}
				C\, \hbar &\text{ if } p \geq 2
				\\
				C \,\hbar \n{\ln\hbar}^\frac{2-p}{2} &\text{ if } p\leq 2\, .
			\end{cases}
		\end{align*}
		The values of the constants $C$ can be made explicit in terms of $V$ and the analogous constant $\cC_1>0$.
	\end{thm}
	
	\begin{remark}
		In the Coulomb case $a=1$, this gives for $p=1$ the commutator estimates
		\begin{equation*}
			\Nrm{\Dhv\opHF}{\L^1} \leq C \quad \text{ and } \quad \Nrm{\Dhx\opHF}{\L^1} \leq C \n{\ln\hbar}^{1/2}.
		\end{equation*}
		The fact that we obtain a $\n{\ln\hbar}^{1/2}$ rather than a $\n{\ln\hbar}$ correction is useful for applications.
		Indeed, in works such as \cite{benedikter_mean-field_2014, benedikter_hartree_2016} one typically obtain errors of the form
		\begin{equation*}
			\Nrm{\op(t) - \op_2(t)}{\L^p} \leq \frac{C}{N^\alpha}\, e^{C_t \Nrm{\Dh\op(0)}{\L^1}}
		\end{equation*}
		for some constant $\alpha>0$, with $\op$ and $\op_2$ solutions of two different evolution equations and $N = \hbar^{-d}$, and $C_t$ some constant growing with time. If $\Nrm{\Dh\op(0)}{\L^1}$ was of size $\n{\ln\hbar}$, then we get an error $\frac{C}{N^\alpha}\, \hbar^{-C_t} = C\,\hbar^{\alpha d -C_t}$ which would become large in finite time. On the other hand, if $\Nrm{\Dh\op(0)}{\L^1}$ is of size $\n{\ln\hbar}^{1/2}$, the error
		remains small globally in time. This is because the error is of the form $\frac{C}{N^\alpha}\, e^{C_t \sqrt{\n{\ln\hbar}}} = C\,e^{C_t \sqrt{\n{\ln\hbar}} - \alpha d \n{\ln\hbar}}$.
	\end{remark}
 
	\begin{remark}\label{rem:rho}
		Similarly as in the linear theory, we obtain the estimates for the densities 
		\begin{equation}
			\Nrm{\nabla \vr_{\opHF}}{L^1} \leq C \, , 
			\quad \text{ if } a < 1
			\quad \text{ and }
			\quad 
			\Nrm{\nabla \vr_{\opHF}}{L^1} \leq C \n{\ln \hbar}^{1/2} , 
			\quad \text{ if } a = 1 \,.
		\end{equation}
	\end{remark}
	
	Finally, we turn to our last main result, regarding the quantitative convergence of position densities, and of states. Recalling that $\rhoTF(x)$ is the associated minimizer of the Thomas--Fermi functional, solving the fixed point equation~\eqref{eq:rhoTF}, we let the associated classical phase space distribution be 
	\begin{equation}\label{eq:fTF}
		\fTF(x,\xi): = \indic(\n{\xi} \leq \om^{-1/d} \rhoTF^{{1/d}}) \, , 
	\end{equation}
	for $(x , \xi) \in \Rd \times \Rd$.	Then, we obtain the following quantitative result. Here again, $\L^1$ denotes the scaled trace norm defined in~\eqref{eq:Schatten}. 
	
	\begin{thm}[Non-linear quantitative convergence]\label{thm:weyl:nonlinear}
		For $\hbar\in(0,1)$ and $d =3$, let $\opHF$ be a minimizer of the Hartree energy $\scE$ given by Equation~\eqref{eq:Hartree_enregy}.
		% Do we need to repeat that in the main theorem? Then reference to the eq defining it ? ...
		Let $\vr_{\opHF}$ be its density, and $f_{\opHF}$ its Wigner transform. 
		Denote by $ \rhoTF$ and $\fTF$ the respective limits, given by \eqref{eq:rhoTF} and \eqref{eq:fTF}. Assuming the potentials verify hypotheses~\eqref{hyp:K} with $a\in (0,1]$ and \eqref{hyp:U}, there is a constant $C> 0$ independent of $\hbar$ such that
		\begin{align*}
			\Nrm{f_{\opHF} - \fTF}{L^2} &\leq C\, \hbar^{1/4}
			\\
			\Nrm{\vr_{\opHF} - \rhoTF}{L^2} &\leq 
			C \, \big(1 + \Nrm{\vr_{\opHF}}{\dot{B}^{1/2}_{2,\infty}}^{2/3} \big) \,
			\hbar^{1/3} \, . 
		\end{align*}
		Moreover, in the case when $ 0< a<1$, it holds
		\begin{equation*}
			\Nrm{\opHF - \op_{\fTF}}{\L^1} \leq C\, \hbar^{1/2}
			\quad
			\text{ and }
			\quad 
			\Nrm{\vr_{\opHF} - \rhoTF}{L^1} \leq C\, \hbar^{1/2}
		\end{equation*}
		while in the Coulomb case $a=1$, one obtains
		\begin{equation*}
			\Nrm{\opHF - \op_{\fTF}}{\L^1} \leq C\, \hbar^{1/2} \n{\ln \hbar}^{1/2} \quad \text{ and }\quad \Nrm{\vr_{\opHF} - \rhoTF}{L^1} \leq C\, \hbar^{1/2} \n{\ln \hbar}^{1/4}.
		\end{equation*}
		Additionally, if $\hbar^4\, e^{-\beta \n{x}}\, \nabla^5 U \in L^\infty(\R^3)$ uniformly in $\hbar$, then there exists a constant $C>0$ independent of $\hbar$ such that for any $0 < a \leq 1$
		\begin{equation*}%	\label{eq:L^2}
			\Nrm{\vr_{\opHF} - \rhoTF}{L^2} \leq C \, \hbar^{1/3} \, . 
		\end{equation*}
	\end{thm}

	\begin{remark}
		A similar proof also yields estimates in $L^q$ for $\rhoHF$ for any $q\in [1,2]$, with rate $\hbar^\frac{1}{q+1}$, and an additional logarithmic factor if $a=1$. Of course, since $\Nrm{\vr_{\opHF}}{L^\infty}$ is bounded uniformly in $\hbar$ in our analysis, one also gets convergence for $q\geq 2$ by interpolation (i.e. H\"older inequality), but in this case with a rate $\hbar^{1/(3q)}$. Similarly, one obtains estimates in $\L^q$ for $\opHF$. Notice however that these do not imply convergence of the Wigner transform in $L^q$ with $q < 2$. Such results could still be obtained for example using weighted $L^2$ estimates as in~\cite{lafleche_optimal_2024}.
	\end{remark}
 
	\begin{remark}\label{remark:besov}
		In our analysis, $\Nrm{\vr_{\opHF}}{L^\infty}$ is bounded uniformly in $\hbar$. Thus, thanks to H\"older's inequality, we obtain that under hypotheses~\eqref{hyp:K} and~\eqref{hyp:U}
		\begin{equation*}
			\Nrm{\vr_{\opHF}}{\dot{B}^{1/2}_{2,\infty}} \leq C \Nrm{\nabla \vr_{\opHF}}{L^1}^{1/2} .
		\end{equation*}
		Note that the right-hand side introduces a logarithmic factor $\n{\ln \hbar}$ for the Coulomb case $a =1$. On the other hand, under the additional condition $\hbar^4\, e^{-\beta \n{x}}\, \nabla^5 U \in L^\infty(\R^3)$ uniformly in $\hbar$, we prove that the following are bounded uniformly in $\hbar$
		\begin{equation}\label{eq:reg:rho}
			\Nrm{ \vr_{\opHF}}{\dot{B}^{1/2}_{2,\infty}} , \quad \sqrt \hbar \Nrm{\nabla \vr_{\opHF}}{L^2} , \quad \text{ and } \quad \Nrm{\vr_{\opHF}}{H^s} , \quad \text{ for } s < 1/2 \, . 
		\end{equation}
		See Proposition~\ref{prop:regu_rho_L2}. The reader should not confuse the estimates~\eqref{eq:reg:rho} with those for $f_{\opHF}$. The latter follow essentially from Theorem~\ref{thm:commutator_nonlinear}, whereas the validity of those of $\vr_{\opHF}$ must be established via \textit{weighted} commutator estimates.
	\end{remark}

\subsection{Strategy of proof}

	Let us summarize here some of the main ideas involved in our proofs. First, for the linear commutator estimates in Theorem~\ref{thm:commutator:linear}, one of the main ingredients is the following local eigenvalue estimate
	\begin{equation}\label{local}
		\hd \tr \indic_{[a,b]}(H) \leq \cC_1 \(\n{b -a} + \hbar\)
	\end{equation}
	where $\cC_1 >0$ is a distinguished constant (depending only on $V$) and $a < b$ are appropriate real numbers sufficiently close to $0$. In our analysis, we derive~\eqref{local} from Weyl's law for $C^{1, \alpha}$ potentials. Equipped now with local eigenvalue estimate~\eqref{local}, we can prove the following singular resolvent estimate 
	\begin{equation*} 
		\hd \Tr{\indic_{H\leq 0}\(\hbar - H\)^{-2}} \leq \cC_2 \,\hbar
	\end{equation*}
	see e.g. Section~\ref{sec:mot} for its motivation and a simple example. Its role in the proof of the commutator estimates is crucial and is developed in Section~\ref{sec:com}. Let us mention that the connection between~\eqref{local} and commutator estimates was first established in~\cite{fournais_optimal_2020}. Here, while we start from the same local eigenvalue estimate, we do not employ additional pseudo-differential operators methods. Instead, our methods are based on the understanding of correlations between eigenfunctions of $H$, created by $x$ or $\opp$. In this regard, our estimates borrow inspiration from~\cite{benedikter_effective_2022}.
	
	Our technique seems to be also related to the double operator integral techniques which were introduced in~\cite{daletskii_integration_1965}, studied more systematically by Birman and Solomyak~\cite{birman_stieltjes_1967} and allowed to obtain the characterization of operator-lipschitz functions in Schatten spaces \cite{potapov_operator-lipschitz_2011}, which is a problem closely related to commutator estimates. We refer the reader to the survey~\cite{aleksandrov_operator_2016} for details and additional references in the huge literature on the subject.
	
	As for the proof of the linear convergence in Theorem~\ref{thm:weyl:linear}, it is based on two steps. First, we use the well-known methods of coherent states to show via variational methods that the Husimi measure of $\opgam = \indic_{H \leq 0}$ is close to the limit $f = \indic_{\n{\xi}^2 + V \leq 0}$. Secondly, we compare the Wigner function to the Husimi measure via convolution estimates --- here, the regularity of $\opgam$ and $\vr_\opgam$ enter the estimate in Theorem~\ref{thm:weyl:linear}.
	
	As for the interacting setting, one of the main ingredients is a representation formula for the minimizers $\opHF$ of $\scE$. Namely, let us denote the effective one-body Hamiltonian by 
	\begin{equation*}
		H_\op := -\hbar^2 \Delta + V_\op(x) \quad \text{ with } \quad V_\op := K*\vr_\op + U\, . 
	\end{equation*}
	Then, the minimizers satisfy the fixed-point equation 
	\begin{equation*}
		\opHF = \indic_{(-\infty, 0)}(H_{\opHF}) + \opq
		\quad \text{ where }\quad
		\ran \opq \subset \ker H_{\opHF} 
	\end{equation*}
	for some self-adjoint $0 \leq \opq \leq \id $, see e.g.~\cite{nguyen_weyl_2024}. This fixed point equation will then establish the connection to the linear theory, after one proves several \textit{apriori estimates} for the ($\hbar$-dependent) non-linear interaction potential
	\begin{equation*}
		V_{\opHF} = U + K* \vr_{\opHF} \, .
	\end{equation*}
	In particular, we prove that its $C^{1, 1/2}_{\loc}$ semi-norm is uniformly bounded in $\hbar$, which yields the optimal Weyl laws for $H_\opgam$ and as a corollary the local eigenvalue estimate~\eqref{local}.
	
	Finally, let us discuss the differences occurring in the choice of $a$ for $K(x) = \kappa\n{x}^{-a}$. The case $a \in [0,1)$ is \textit{sub-critical} and we obtain the optimal commutator estimates~\eqref{eq:comms} for the minimizers of $\scE$. The case $a =1$ corresponds to the Coulomb potential and becomes the \textit{critical case}, where we can only show the commutator estimates~\eqref{eq:comms} with an additional logarithmic factor $\n{\ln\hbar}$. The key difference between these two regimes comes from the estimates that are satisfied by the mean-field potential
	\begin{equation}\label{map}
		\vr\in L^1(\R^3) \cap L^\infty(\R^3) \implies 
		K * \vr \in C^2_b(\R^3) \, . 
	\end{equation}
	Let us further explain. First, we observe that the CLR bound and standard Agmon estimates imply $\sNrm{\vr_{\opHF}}{L^1} \leq C$ and $\sNrm{\vr_{\opHF}}{L^\infty} \leq C$, respectively. In our proof, in order to obtain the optimal commutator estimates in the trace-class norm, we require boundedness of the $C^2$ norm of the mean-field potential $K *\vr_{\opHF}$, uniformly in $\hbar$. In particular, since the map $\vr\mapsto K*\vr$ is continuous for $ 0 < a < 1$ from $L^1\cap L^\infty \to C^2_b$, this is easily obtained via the $L^1$ and $L^\infty$ estimates. On the other hand, for the Coulomb case, we lose the continuity of the map, and are unable to make this strategy work.

\section{A priori estimates}

	Throughout this section, we denote, in dimension $d \geq 1$
	\begin{equation*}
		\opgam = \indic_{H\leq 0} \quad \text{ with } \quad H=-\hbar^2\Delta+V
	\end{equation*}
	for some real-valued $V \in L^1_{\loc}(\Rd)$ such that $V_-\in L^\infty(\Rd)$. We state several estimates on $\opgam$, quantitatively with respect to the potential $V$. Let us introduce here the notation for the momentum operator on $L^2(\Rd)$, to be used in the rest of the paper
	\begin{equation*}
		\opp := - i \hbar \nabla \, . 
	\end{equation*}
	
\subsection{\texorpdfstring{$L^p$}{Lp} estimates}

	Recall that the position density $\vr_\opgam$ associated to $\opgam$ was defined in~\eqref{eq:density}. Let us first notice that the $L^1$ norm of $\vr_\opgam$ is bounded uniformly in $\hbar$ by the Cwikel--Lieb--Rozenblum inequality \cite{rozenblum_distribution_1972, cwikel_weak_1977, lieb_number_1980} in dimension $d\geq 3$. That is, there exists a constant $\sfL_{0,d}>0$ such that
	\begin{equation}\label{eq:CLR}
		\Nrm{\vr_\opgam}{L^1} \leq \sfL_{0,d}
		\intd V_-^{d/2} \d x \,. 
	\end{equation}
	In lower dimensions, one can use the fact that $V_-$ is compactly supported and bounded and $V$ is growing at infinity as in~\cite[Theorem~XIII.81]{reed_analysis_1978}. 
	
	Combined with the $L^1$ bound~\eqref{eq:CLR} and Agmon's estimates (established below), the following lemma implies 
	uniform bounds for all $L^p$ norms of $\vr_\opgam$.

	\begin{lem}\label{lem:Linfty}
		Let $1 \leq d \leq 3$. Then, there exists $C_d>0$ such that
		\begin{equation*} 
			\Nrm{\vr_\opgam}{L^\infty} \leq C_d \(1 + \Nrm{V_-}{L^\infty} + \Nrm{V \opgam}{\L^\infty}\).
		\end{equation*}
	\end{lem}

	\begin{proof}
		Let $u \in L^1(\Rd)$ be a non-negative test function. Then, we have for all $s>0$
		\begin{equation*}
			\intd u \,\vr_\opgam = \hd \Tr{u\,\opgam} = \Nrm{\sqrt{u}\,\opgam}{\L^2}^2 \leq
			\Nrm{\sqrt u \weight{\opp}^{-s}}{\L^2}^2 \Nrm{\weight{\opp}^s \opgam}{\L^\infty}^2
		\end{equation*}
		where $\weight{\opp}^2 = \id + \opp^*\opp$. For $s > \frac{d}{2}$ we compute $\sNrm{\sqrt u \weight{\opp}^{-s}}{\L^2}^2 = \Nrm{u}{L^1} \sNrm{\weight{\xi}^{-2s}}{L^1_\xi} = C(s,d) \Nrm{u}{L^1}$. We therefore find for $d\leq 3$ that with $s=2$
		\begin{equation}
			\Nrm{\vr_\opgam}{L^\infty} \leq C_d \Nrm{\weight{\opp}^2 \opgam}{\L^\infty} \leq C_d \(1 + \sNrm{\n{\opp}^2 \opgam}{\L^\infty}\). 
		\end{equation}
		The last term can be estimated as follows
		\begin{equation}\label{eq:weight_2}
			\Nrm{\n{\opp}^2 \opgam}{\L^\infty} \leq \Nrm{H \opgam}{\L^\infty} + \Nrm{V \opgam}{\L^\infty}\, . 
		\end{equation}
		It suffices now to use $H \geq - \Nrm{V_-}{L^\infty}$. This finishes the proof.
	\end{proof}
	
\subsection{Energy estimates}

	Another common way to obtain bounds on the $L^p$ norms with $p<\infty$ is through Lieb--Thirring inequalities~\cite{lieb_inequalities_1976}, which give for instance the following quantum analogue of a classical interpolation inequality used in kinetic theory (see e.g.~\cite{lions_propagation_1991})
	\begin{equation}\label{eq:Lieb-Thirring}
		\Nrm{\vr_{\op}}{L^p} \leq \(\CLT \Nrm{\op}{\L^\infty}\)^{1/p'} \(\hd\Tr{\n{\opp}^2\!\op}\)^{1/p}
	\end{equation}
	for $d \geq 1$. Here, $\op$ is a density operator, $p = 1 + \frac{2}{d}$ and $\CLT$ only depends on $d$. It gives in particular the following natural upper bounds for the kinetic and potential energy of the spectral function $\opgam$ in terms of $V_-$.
	\begin{lem}\label{lem:LT}
		There holds
		\begin{equation*}
			\hd \Tr{\n{\opp}^2\opgam} \leq \intd \vr_{\opgam} V_- \leq 
			\CLT \Nrm{V_-}{L^{p'}}^{p'}
		\end{equation*}
		where $p' = 1+\frac{d}{2}$.
	\end{lem}

	\begin{proof}
		The first inequality just follows from the fact that $\intd \vr_{\opgam} V_- - \hd \Tr{\n{\opp}^2\opgam} = \hd \Tr{H\opgam} - \intd \vr_{\opgam} V_+ \leq 0$. Combining this inequality with H\"older's inequality, the fact that $0\leq \opgam \leq \id$ and the Lieb--Thirring inequality~\eqref{eq:Lieb-Thirring} gives
		\begin{align*}
			\intd \vr_{\opgam} V_- &\leq \Nrm{\rho_\opgam}{L^p} \Nrm{V_-}{L^{p'}} \leq \CLT^{1/p'} \(\hd\Tr{\n{\opp}^2\!\op}\)^{1/p} \Nrm{V_-}{L^{p'}}
			\\
			&\leq \CLT^{1/p'} \(\intd \vr_{\opgam}\, V_-\)^{1/p} \Nrm{V_-}{L^{p'}}
		\end{align*}
		from which the second inequality follows.
	\end{proof}

\subsection{Agmon-like estimates}
	
	In order to establish some decay properties of the density $\vr_\opgam$ at infinity, we review here the well-known Agmon estimates \cite{agmon_spectral_1975, agmon_lectures_1982} for the density operator $\opgam$.
	\begin{lem}[Energy estimate]
		Let $d \geq 1 $. Then, for any smooth function $u:\Rd\to\R$,
		\begin{equation}\label{eq:Agmon_0}
			\Tr{(V-\hbar^2 \n{\nabla u}^2)\,e^{2u}\,\opgam} = \Tr{e^{2u}\,\opgam\,H} - \Tr{\n{\opp}^2 e^u\,\opgam\, e^u} \leq 0 \, .
		\end{equation}
	\end{lem}

	\begin{proof}
		By the Leibniz rule for commutators and the fact that $\com{\opp,e^u} = -i\hbar\,e^u\,\nabla u$
		\begin{align*}
			\com{H,e^u} = \com{\n{\opp}^2,e^u} &= -i\hbar \(\opp\cdot e^u\nabla u + e^u\nabla u\cdot \opp\)
			\\
			&= -i\hbar\, e^u \(\opp\cdot \nabla u + \nabla u\cdot \opp\) - \hbar^2 \n{\nabla u}^2 e^u.
		\end{align*}
		Therefore, by cyclicity of the trace
		\begin{multline*}
			\Tr{H\,e^{2u}\,\opgam} = \Tr{\com{H,e^u} e^u\,\opgam} + \Tr{H\, e^u\,\opgam\, e^u}
			\\
			= -i\hbar\Tr{e^u \(\opp\cdot \nabla u + \nabla u\cdot \opp\) e^u\,\opgam} - \hbar^2 \Tr{\n{\nabla u}^2 e^{2u}\,\opgam} + \Tr{H\, e^u\,\opgam\, e^u}.
		\end{multline*}
		Since $\opgam\,H$ is self-adjoint, $\tr(H\,e^{2u}\,\opgam) = \tr(\opgam\,H\,e^{2u}) \in\R$. Moreover, $\tr(\n{\nabla u}^2 e^{2u}\,\opgam)$ and $\tr(H\, e^u\,\opgam\, e^u)$ are also real because $\opgam$ is self-adjoint. On the other side, the term $i\hbar\Tr{e^u \(\opp\cdot \nabla u + \nabla u\cdot \opp\) e^u\,\opgam}$ is imaginary. Hence, taking the real part of the previous equation yields
		\begin{equation*}
			\Tr{H\,e^{2u}\,\opgam} = - \hbar^2 \Tr{\n{\nabla u}^2 e^{2u}\,\opgam} + \Tr{H\, e^u\,\opgam\, e^u}
		\end{equation*}
		which can also be written
		\begin{equation*}
			\Tr{(V-\hbar^2 \n{\nabla u}^2)\,e^{2u}\,\opgam} = \Tr{e^{2u}\,\opgam\,H} - \Tr{\n{\opp}^2 e^u\,\opgam\, e^u}.
		\end{equation*}
		Since $\opgam\,H\leq 0$, it implies the result.
	\end{proof}

	The above lemma implies both exponential decay in $x$ and smallness in $\hbar$ in the region where the potential is positive. In the following proposition, we write operator versions of Agmon's estimates on exponential decay away from the classically allowed regions. For $ 0 \leq a < b $ we use the notation
	\begin{equation}\label{eq:omega}
		\Omega_a := \{V \leq a \}\,,
		\quad
		\Omega_{a,b} := \{x \in \Rd: \Dist{x, \Omega_a} < b \} = \Omega_a + B_b 
	\end{equation}
	where $B_r : = \{ x \in \Rd : \n{x} < r \}$. We also use the distance function 
	\begin{equation}\label{eq:dist}
		d_a(x) := \Dist{x, \Omega_a} .
	\end{equation}

	\begin{prop}[Smallness outside of the bulk]\label{prop:Agmon}
		Let $d \geq1$ and $a>0$. Then, the following statements are true. For any $0 < \alpha < a$
		\begin{equation}\label{eq:smallness_out_of_bulk}
			\hd\Tr{\indic_{V>a}\,e^{2\sqrt{\alpha}\,d_a/\hbar}\,\opgam} \leq \frac{1}{a-\alpha}\intd \vr_\opgam\,V_- \, . 
		\end{equation}
		For any $R>0$,
		\begin{align}\label{eq:smallness_out_of_bulk_2}
			\int_{\Omega_{1,R}^c} e^{\frac{d_1}{2\hbar}}\vr_\opgam &\leq \frac{4}{3}\, e^{-\frac{R}{2\hbar}} \intd \vr_\opgam\, V_-\, ,
			\\\nonumber
			\Nrm{\indic_{\Omega_{1,R}^c} e^\frac{d_1}{4\hbar}\, \opgam}{\L^\infty} &\leq \frac{C_{d}}{R^{d/2}} \,e^{-\frac{R}{8\hbar}} \(\intd \vr_\opgam\, V_-\)^{1/2}
		\end{align}
		where $C_d = 2/\sqrt{3} \(2d/(\pi e)\)^{d/2}$. For all $\beta>0$ such that $\beta\,\hbar \leq 1/4$, 
		\begin{equation}\label{eq:smallness_out_of_bulk_3}
			\Nrm{e^{\beta\,d_1}\,\opgam}{\L^\infty} \leq e^{\beta} + C_d \,e^{-\,\frac{1}{8\hbar}} \(\intd \vr_\opgam\, V_- \)^{1/2}.
		\end{equation}
		In particular, let $r \geq 1$ be such that $\Omega_{1,1} \subset B_r$. Then, for $0 \leq \beta \leq 1/ 4 \hbar$
		\begin{equation}\label{exp:decay}
			\intd e^{ \beta \n{x}} \vr_\opgam \leq C \intd \vr_\opgam \(1 + V_-\) ,
		\end{equation}
		with $C = e^{ 2 \beta r} + 4 / (3 \sqrt e)$. 
	\end{prop}

	\begin{proof}[Proof of Proposition~\ref{prop:Agmon}]
		Take $u = \sqrt{\alpha}\, d_a/\hbar$ in Inequality~\eqref{eq:Agmon_0}. Then $\n{\nabla u} = \frac{\sqrt{\alpha}}{\hbar} \,\indic_{V>a}$ and Inequality~\eqref{eq:Agmon_0} gives
		\begin{equation*}
			\hd\Tr{(V-\alpha \indic_{V>a})\,e^{2u}\,\opgam} \leq 0 \, .
		\end{equation*}
		This implies that
		\begin{equation*}
			\hd\Tr{\(V-\alpha\)\indic_{V>a}\,e^{2u}\,\opgam} \leq -\hd\Tr{V\,\indic_{V\leq a}\,e^{2u}\,\opgam} \leq \intd \vr_\opgam \, V_-\,,
		\end{equation*}
		which proves Inequality~\eqref{eq:smallness_out_of_bulk}. Let us now take $a = 1$ and $\alpha = 1 /4$ to find that 
		\begin{equation}\label{eq:smallness_out_of_bulk_1}
			\hd\Tr{\indic_{V>1}\,e^{\,d_1 /\hbar}\,\opgam} \leq \frac{4}{3}\intd \vr_\opgam\,V_- \, . 
		\end{equation}
		Inequality~\eqref{eq:smallness_out_of_bulk_2} now follows by noticing that $x\in\Omega_{1,R}^c \implies d_1(x) \geq (d_1(x)+R)/2$. To get the next inequality, observe that
		\begin{equation*} %\label{eq:2}
			\Nrm{\indic_{\Omega_{1,R}^c} e^\frac{d_1}{4\hbar}\, \opgam}{\L^\infty}^2 = \Nrm{\opgam\,\indic_{\Omega_{1,R}^c} e^\frac{d_1}{2\hbar}\, \opgam}{\L^\infty} \leq \Tr{\indic_{\Omega_{1,R}^c} e^\frac{d_1}{2\hbar}\, \opgam} = \frac{1}{\hd} \int_{\Omega_{1,R}^c} e^\frac{d_1}{2\hbar} \vr_\opgam \, . 
		\end{equation*} 
		Hence it follows from Inequality~\eqref{eq:smallness_out_of_bulk_2} and the fact that for any $t>0$, $t^d \,e^{-R\pi t/2} \leq \(2d/(R\pi e)\)^d$
		\begin{equation*}
			\Nrm{\indic_{\Omega_{1,R}^c} e^\frac{d_1}{4\hbar}\, \opgam}{\L^\infty} \leq \frac{C_d}{R^{d/2}} \,e^{-\frac{R}{8\hbar}} \(\intd \vr_\opgam V_-\)^{1/2}.
		\end{equation*}
		Finally, we use the triangle inequality and use $0\leq \opgam\leq 1$ and $\beta\, \hbar \leq 1/4$ to get
		\begin{equation*}
			\Nrm{e^{\beta d_1}\, \opgam}{\L^\infty} \leq e^{\beta R} + \Nrm{\indic_{\Omega_{1,R}^c} e^\frac{d_1}{4\hbar}\, \opgam}{\L^\infty} ,
		\end{equation*}
		which implies Inequality~\eqref{eq:smallness_out_of_bulk_3} by taking $R=1$.

		For the last inequality, we split the integration region in $B_{2 \ell}$ and $B_{2\ell}^c$. In the first one, $e^{\beta \n{x}} \leq e^{2 \beta \ell} $. In the second one, we use $\beta \n{x} \leq \frac{\hbar \n{x}}{4}$ and $d_1(x) \geq \n{x} - 1 \geq \frac{\n{x}}{2}$, and then~\eqref{eq:smallness_out_of_bulk_2}. 
	\end{proof}
 
	We will need to consider some higher order Agmon-like estimates, which can be thought of as Agmon's estimates for gradients of eigenfunctions, or from a semiclassical point of view, as decay in phase space in both $x$ and $\xi$. Recalling that $\vr_\op = \hd \op(x,x)$, we look first at $\vr_{\opp\cdot\op\,\opp}$, which is a quantum analogue of $\intd f \n{\xi}^2\d\xi$. In the case when $\op$ is diagonalized in the form $\op = \sum_{j\geq 0} \lambda_j \ket{\psi_j}\bra{\psi_j}$ for some orthonormal family of eigenvectors $(\psi_j)_{j\geq 0}$, then
	\begin{equation}\label{eq:rho_2_diagonalisation}
		\vr_{\opp\cdot\op\,\opp} = \hd \sum_{j\geq 0} \lambda_j \n{\hbar\nabla\psi_j}^2.
	\end{equation}
	One can also notice that by the definition of $\vr_\op$ in the weak sense, for any nonnegative function $\varphi\in L^\infty(\Rd)$,
	\begin{equation*}
		\intd \vr_{\opp\cdot\op\,\opp} \,\varphi = \hd \Tr{\opp\cdot\op\,\opp\,\varphi} = \Nrm{\sqrt{\varphi}\,\opp\,\op}{\L^2}^2.
	\end{equation*}
	We will also need similar estimates for higher order gradients. When multiplying two vectors, we will mean their tensor product. In particular, we denote by $\opp^k := \opp\otimes \opp\dots\otimes \opp$ the operator-valued tensor $(\opp_{j_1}\dots\opp_{j_k})_{1\leq j_1\dots,j_k\leq d}$, and, for tensors of order $2$, denoting their double contraction by $A : B := \sum_k A_{jk}B_{kj}$,
	\begin{equation}\label{eq:rho_4_diagonalisation}
		\vr_{\opp^2 :\op\,\opp^2} = \hd \sum_{j\geq 0} \lambda_j \n{\hbar^2\nabla^2\psi_j}^2.
	\end{equation}
	with $\nabla^2\psi_j$ denoting the Hessian of $\psi_j$. Then, for any nonnegative $\varphi\in L^\infty(\Rd)$,
	\begin{equation*}
		\intd \vr_{\opp^2:\op\,\opp^2} \,\varphi = \hd \Tr{\opp^2:\op\,\opp^2\,\varphi} = \Nrm{\sqrt{\varphi}\,\opp^2\op}{\L^2}^2.
	\end{equation*}
	
	In the classical case, if $f = \indic_{\n{\xi}^2 \leq V_-}$, then we have the following straightforward inequalities
	\begin{equation*}
		\intd f \n{\xi}^2 \d \xi \leq V_- \,\vr_f \quad \text{ and } \quad \intd f \n{\xi}^4 \d \xi \leq V_- \intd f \n{\xi}^2 \d \xi \, .
	\end{equation*}
	The following lemmas can be seen as quantum analogues of the above formulas. Recall $d_1$ is given by \eqref{eq:dist}.
	
	\begin{lem}[Agmon-like estimate for gradients]\label{lem:agmon_gradients}
		Let $d \geq 1 $. Then, in the weak sense, there holds 
		\begin{equation}\label{eq:Agmon_rho_2}
			0 \leq \vr_{\opp\cdot\opgam\,\opp} \leq \(\tfrac{\hbar^2}{2}\Delta -V\)\vr_\opgam\,.
		\end{equation}
		This and Agmon estimates give that for any $R>0$,
		\begin{align}\label{eq:Agmon_rho_2_L1}
			\int_{\Omega_{1,R}^c} \vr_{\opp\cdot\opgam\,\opp} \, e^\frac{d_1}{2\hbar} &\leq \,e^{1-\frac{R}{2\hbar}} \intd \vr_\opgam\,V_-
			\\\label{eq:Agmon_rho_2_Linfty}
			\Nrm{\indic_{\Omega_{1,R}^c}\,e^\frac{d_1}{4\hbar}\,\opp\,\opgam}{\L^\infty} &\leq \frac{C_d}{R^{d/2}} \,e^{\frac{1}{2}-\frac{R}{8\hbar}} \(\intd \vr_\opgam\,V_-\)^{1/2}
		\end{align}
		with $C_d = \(2d/(\pi e)\)^{d/2}$. Moreover, for any $\beta >0$ such that $\beta \hbar \leq 1/4$,
		\begin{equation}\label{eq:Agmon_rho_2_Linfty_global}
			\Nrm{e^{\beta d_1} \opp \,\opgam}{\L^\infty} \leq \sqrt{2} \, e^\beta \Nrm{V_-}{L^\infty}^{1/2} + C_d \, e^{\frac{1}{2}-\frac{1}{8\hbar}} \(\intd \vr_\opgam\,V_-\)^{1/2}.
		\end{equation}
	\end{lem}
	
	\begin{proof}
		Let $\varphi\in C^2_b(\Rd)$ be such that $\varphi\geq 0$. Then since $\com{\opp,\varphi} = -i\hbar\nabla \varphi$ and by cyclicity of the trace
		\begin{align*}
			0\leq \intd \vr_{\opp\cdot\opgam\,\opp}\, \varphi &= \hd\Tr{\opp\cdot \opgam\,\opp\, \varphi} = -i\hbar\,\hd\Tr{\opp\cdot \opgam\,\nabla \varphi} + \hd\Tr{\n{\opp}^2 \opgam\, \varphi}
			\\
			&= -i\hbar\,\hd\Tr{\nabla \varphi\cdot\opp\, \opgam} + \hd\Tr{\varphi\, \opgam\, H} - \hd\Tr{V \varphi\, \opgam}
		\end{align*}
		where $\nabla \varphi$ denote the operator of multiplication by $\nabla \varphi$. Taking the real part and observing that $i\hbar\,\opp\cdot\nabla \varphi - i\hbar\,\nabla \varphi\cdot\opp = \hbar^2 \Delta \varphi$ yields
		\begin{equation*}
			\intd \vr_{\opp\cdot\opgam\,\opp}\, \varphi = \hd\Tr{\tfrac{\hbar^2}{2}\,\Delta \varphi\, \opgam} + \hd\Tr{\varphi\, \opgam\, H} - \hd\Tr{V \varphi\, \opgam}
		\end{equation*}
		and so since $\varphi\geq 0$ and $\opgam\, H \leq 0$, it follows that
		\begin{equation*}
			0\leq \intd \vr_{\opp\cdot\opgam\,\opp}\, \varphi \leq \intd \vr_\opgam\(\tfrac{\hbar^2}{2}\Delta \varphi - V\,\varphi\)
		\end{equation*}
		which is the meaning of Inequality~\eqref{eq:Agmon_rho_2}. The previous Agmon estimates together with a standard approximation argument show that we can take $\varphi$ to grow exponentially fast at infinity. In particular, taking $\varphi = e^\frac{v}{\hbar}$ gives
		\begin{equation*}
			\intd \vr_{\opp\cdot\opgam\,\opp}\, e^\frac{v}{\hbar} \leq \tfrac{1}{2} \intd \vr_\opgam\(\n{\nabla v}^2 + \hbar\,\Delta v - 2\,V\) e^\frac{v}{\hbar}.
		\end{equation*}
		Let us take for instance $v = \(\hbar^2+d_1^2\)^{1/2} \geq \hbar$, so that $d_1 \leq v \leq d_1 + \hbar$. Then $\n{\nabla v} = \frac{d_1}{v} \leq \indic_{V\geq 1}$ and $\Delta v = \indic_{V\geq 1} \, \frac{\hbar^2}{v^3} \leq \indic_{V\geq 1} \, \frac{1}{\hbar}$, and so
		\begin{equation}\label{eq:Agmon_rho_2_L1_0}
			\intd \vr_{\opp\cdot\opgam\,\opp}\, e^\frac{d_1}{\hbar} \leq \intd \vr_{\opp\cdot\opgam\,\opp}\, e^\frac{v}{\hbar} \leq e \intd \vr_\opgam\,V_- \, .
		\end{equation}
		Similarly as for the Agmon estimate, restricting to $\Omega_{1,R}^c$ gives Inequality~\eqref{eq:Agmon_rho_2_L1} and then using the inequalities between Schatten norms leads to Inequality~\eqref{eq:Agmon_rho_2_Linfty}. This inequality with $R=1$, combined with the fact that
		\begin{equation}\label{eq:weight_1}
			\n{\opp\,\opgam}^2 = \opgam \n{\opp}^2 \opgam \leq \opgam \(H + V_-\) \opgam \leq 2 \Nrm{V_-}{L^\infty}
		\end{equation}
		leads to Inequality~\eqref{eq:Agmon_rho_2_Linfty_global}.
	\end{proof}

	In the case of a vector valued function $u:\Rd\to\Rd$, the previous lemma gives the following.
	\begin{cor}\label{cor:agmon_gradients}
%		Let $\opgam = \indic_{H\leq 0}$ be given by Equation~\eqref{eq:H} in $d \geq 1$. Then
		Let $d \geq 1$. Then, there holds
		\begin{equation*}
			\Nrm{\indic_{\Omega_{1,R}^c} u\cdot\opp\,\opgam}{\L^\infty} \leq \frac{C_d}{R^{d/2}} \,e^{\frac{1}{2}-\frac{R}{8\hbar}} \Nrm{\n{u}^2 \, e^{-\,\frac{d_1}{2\hbar}}}{L^\infty}^{1/2} \(\intd \vr_\opgam\,V_- \)^{1/2}
		\end{equation*}
		with $C_d = \(2d/(\pi e)\)^{d/2}$. Thus, if $\beta\,\hbar\leq 1/4$,
		\begin{equation*}
			\Nrm{u\cdot\opp\,\opgam}{\L^\infty} \leq C_{\beta,\opgam} \Nrm{V_-}{L^\infty} \Nrm{\n{u} \, e^{-\beta\,d_1}}{L^\infty}.
		\end{equation*}
		with $C_{\beta,\opgam} = \sqrt{2} \,e^\beta + C_d \,e^{\frac{1}{2}-\frac{1}{8\hbar}} \Nrm{\vr_\opgam}{L^1}^{1/2}$.
	\end{cor}

	\begin{proof}
		Diagonalizing $\opgam$ in the form $\opgam = \sum_{j\geq 0} \ket{\psi_j}\bra{\psi_j}$ with $(\psi_j)_{j\geq 0}$ orthonormal, we can write
		\begin{equation*}
			\Nrm{u\cdot\opp\,\opgam}{\L^2}^2 = \hd \Tr{\opgam \,\opp\cdot u\,u\cdot\opp\,\opgam} = \hd \sum_{j\geq 0} \intd \n{u\cdot \hbar\nabla\psi_j}^2 \leq \intd \n{u}^2 \hd \sum_{j\geq 0} \n{\hbar\nabla\psi_j}^2
		\end{equation*}
		hence by Formula~\eqref{eq:rho_2_diagonalisation},
		\begin{equation*}
			\Nrm{\indic_{\Omega_{1,R}^c} u\cdot\opp\,\opgam}{\L^2}^2 \leq \int_{\Omega_{1,R}^c} \n{u}^2 \vr_{\opp\cdot\opgam\,\opp}\, ,
		\end{equation*}
		which yields the result as in the previous proof using the inequality between Schatten norms and Inequality~\eqref{eq:Agmon_rho_2_L1}.
	\end{proof}
	
	Let us give as an application of the Agmon-like estimates for gradients a lemma on the decay in the momentum direction for $\opgam$. These will find an application when we treat Coulomb-like potentials. 
	
	\begin{lem}\label{lem:weight_4}
		Let $d \geq 1$, and assume that $\nabla V$ and $\Delta V$ are locally integrable. Then, 
		\begin{equation*}
			\Nrm{\n{\opp}^4 \opgam}{\L^\infty} \leq C_{\beta,V,\opgam} \(1+\Nrm{V_-}{L^\infty}\)
		\end{equation*}
		with $C_{\beta,V,\opgam} = \Nrm{V_-}{L^\infty}+3\Nrm{V\opgam}{\L^\infty} + 2 \, C_{\beta,\opgam} \Nrm{\n{\hbar\nabla V} e^{-\beta\,d_1}}{L^\infty} + \Nrm{\hbar^2\Delta V\opgam}{\L^\infty}$.
	\end{lem}
	
	\begin{proof}
		Since $\n{\opp}^4 \opgam = \n{\opp}^2 \opgam \, H - \n{\opp}^2 V \opgam$, and $\com{\n{\opp}^2, V} = -i\hbar \(\opp\cdot\nabla V + \nabla V\cdot \opp\)$, it follows that
		\begin{align*}
			\n{\opp}^4 \opgam &= \n{\opp}^2 \opgam \, H - V \n{\opp}^2 \opgam + i\hbar \(\opp\cdot\nabla V + \nabla V\cdot \opp\)\opgam
			\\
			&= \n{\opp}^2 \opgam \, H - V \opgam\, H + V^2 \opgam + 2i\hbar \,\nabla V\cdot \opp\,\opgam + \hbar^2 \Delta V \opgam
		\end{align*}
		and so taking the operator norm, we deduce from the triangle inequality that
		\begin{equation*}
			\Nrm{\n{\opp}^4 \opgam}{\L^\infty} \leq I_1 + I_2 + I_3
		\end{equation*}
		with
		\begin{align*}
			I_1 &= \Nrm{\(\n{\opp}^2 -V\) \opgam\, H}{\L^\infty} \leq \Nrm{\(H - 2V\) \opgam}{\L^\infty} \Nrm{\opgam \, H}{\L^\infty}
			\\
			I_2 &= \Nrm{\(V^2 +\hbar^2 \Delta V\)\opgam}{\L^\infty} \leq \Nrm{V_-}{L^\infty}\Nrm{V\opgam}{\L^\infty} + \hbar^2 \Nrm{\Delta V\opgam}{\L^\infty}
			\\
			I_3 &= 2\hbar \Nrm{\nabla V\cdot\opp\,\opgam}{\L^\infty}.
		\end{align*}
		The term $I_1$ is bounded using $H_- \leq \Nrm{V_-}{L^\infty}$ by
		\begin{equation*}
			I_1 \leq \(\Nrm{V_-}{L^\infty}+2\Nrm{V\opgam}{\L^\infty}\) \Nrm{V_-}{L^\infty}
		\end{equation*}
		while the terms $I_2$ and $I_3$ are bounded using the Agmon-like inequalities in Proposition~\ref{prop:Agmon} and Corollary~\ref{cor:agmon_gradients}.
	\end{proof}
	
	\begin{lem}[Agmon-like estimate for the Hessian]\label{lem:agmon2}
		In the weak sense, it holds
		\begin{equation}\label{eq:Agmon_rho_4}
			0\leq \vr_{\opp^2:\opgam\,\opp^2} \leq \(\tfrac{\hbar^2}{2}\Delta - V\)\vr_{\opp\cdot\opgam\,\opp} -\tfrac{\hbar^2}{2}\, \nabla\vr_\opgam \cdot \nabla V\,.
		\end{equation}
		Combined with the previous Agmon-like estimates, it gives for $\beta\,\hbar\leq 1/8$
		\begin{equation*}
			\Nrm{e^{\beta d_1}\opp^2\,\opgam}{\L^\infty} \leq e^{\beta} \(\Nrm{V_-}{L^\infty} + \Nrm{V \opgam}{\L^\infty}\) + C_d\, C_{\opgam,V}^{1/2} \Nrm{\vr_\opgam}{L^1}^{1/2}
		\end{equation*}
		with $C_d = \(2d/(\pi e)\)^{d/2}$ and
		\begin{equation}\label{eq:cst_Agmon_rho_4}
			C_{\opgam,V} = \Nrm{\(\tfrac{\hbar}{2}\n{\nabla V} + \hbar^2\Delta V\) e^{-\frac{d_1}{2\hbar}}}{L^\infty} \(1+\Nrm{V_-}{L^\infty}\) + \sqrt{e} \Nrm{V_-}{L^\infty}^2.
		\end{equation}
	\end{lem}
	
	\begin{remark}
		As in Corollary~\ref{cor:agmon_gradients}, this allows to control terms of the form $\Nrm{u\cdot \opp\,\opp\, \opgam}{\L^\infty}$ if $u:\Rd\to\Rd$ or $\Nrm{u: \opp^2\, \opgam}{\L^\infty}$ if $u$ is matrix valued.
	\end{remark}
	
	\begin{proof}
		Let $\varphi\geq 0$. Then since $\com{\opp,\varphi} = -i\hbar\nabla \varphi$, with $\nabla \varphi$ the operator of multiplication by $\nabla \varphi$, and by cyclicity of the trace
		\begin{align*}
			0\leq \intd \vr_{\opp^2:\opgam\,\opp^2}\, \varphi &= \hd\Tr{-i\hbar\,\opp\cdot \opgam\,\opp \(\nabla \varphi\cdot\opp\) + \n{\opp}^2 \opp \cdot\opgam\,\opp\, \varphi}
			\\
			&= \hd\Tr{-i\hbar\,\opp\cdot \opgam\,\opp \(\nabla \varphi\cdot\opp\) + \opp \cdot H \opgam\,\opp\, \varphi - \opp \cdot V \opgam\,\opp\, \varphi}
		\end{align*}
		which can be written
		\begin{multline*}
			\intd \vr_{\opp^2:\opgam\,\opp^2}\, \varphi = \hd\Tr{-i\hbar\,\opp\cdot \opgam\,\opp \(\nabla \varphi\cdot\opp\) + \opp \cdot H \opgam\,\opp\, \varphi - V\opp \cdot \opgam\,\opp\, \varphi + i\hbar\,\opgam\,\opp\cdot \nabla V \varphi}
			\\
			= \hd\Tr{\opp \cdot H \opgam\,\opp\, \varphi + i\hbar\(\opgam\,\opp\cdot \(\nabla V\) \varphi- \opp\cdot \opgam\,\opp \(\nabla \varphi\cdot\opp\)\)} - \intd \vr_{\opp\cdot\opgam\,\opp} \,V\varphi \, .
		\end{multline*}
		Taking the real part and observing that $i\hbar\,\opp\cdot\nabla \varphi - i\hbar\,\nabla \varphi\cdot\opp = \hbar^2 \Delta \varphi$ yields
		\begin{equation*}
			\intd \vr_{\opp^2:\opgam\,\opp^2}\, \varphi = \hd\Tr{\opp \cdot H \opgam\,\opp\, \varphi} + \tfrac{\hbar^2}{2}\intd \vr_\opgam\, \nabla\cdot\(\varphi\nabla V\) + \vr_{\opp\cdot\opgam\,\opp} \,\Delta \varphi - \intd \vr_{\opp\cdot\opgam\,\opp} \,V\varphi
		\end{equation*}
		and we conclude using the fact that $\varphi\geq 0$ and $\opgam\,H \leq 0$ that
		\begin{equation*}
			0\leq \intd \vr_{\opp^2:\opgam\,\opp^2}\, \varphi \leq \tfrac{\hbar^2}{2} \intd \vr_\opgam\, \nabla\cdot\(\varphi\nabla V\) + \vr_{\opp\cdot\opgam\,\opp} \(\tfrac{\hbar^2}{2}\Delta \varphi - V\varphi\)
		\end{equation*}
		which gives Inequality~\eqref{eq:Agmon_rho_4}. Let $v = (\hbar^2+d_1^2)^{1/2}$ and take $\varphi = e^\frac{v}{2\hbar}$. Then we get similarly as for the Agmon estimate for gradients
		\begin{align*}
			\intd \vr_{\opp^2:\opgam\,\opp^2}\, e^\frac{v}{2\hbar}
			&= \intd \(\tfrac{\hbar}{4}\,\nabla v\cdot\nabla V+ \tfrac{\hbar^2}{2}\Delta V\) \vr_\opgam \,e^\frac{v}{2\hbar} + \(\tfrac{\hbar}{4}\, \Delta v + \tfrac{1}{8}\n{\nabla v}^2 - V\) \vr_{\opp\cdot\opgam\,\opp}\, e^\frac{v}{2\hbar}
%			\\
%			&\leq \intd \(\tfrac{\hbar}{4}\n{\nabla V} + \tfrac{\hbar^2}{2}\Delta V\) \vr_\opgam \,e^\frac{v}{2\hbar} + \(\tfrac{3}{8}\, \indic_{V\geq 1} - V\) \vr_{\opp\cdot\opgam\,\opp}\, e^\frac{v}{2\hbar}
			\\
			&\leq \sqrt{e} \intd \(\tfrac{\hbar}{4}\n{\nabla V} + \tfrac{\hbar^2}{2}\Delta V\) \vr_\opgam \,e^\frac{d_1}{2\hbar} + V_-\, \vr_{\opp\cdot\opgam\,\opp}
		\end{align*}
		where we used the fact that $v \leq \hbar+d_1$ to get the last inequality. By Inequality~\eqref{eq:smallness_out_of_bulk_1} and using $\Tr{\opgam \n{\opp}^2} \leq \Tr{\opgam\, V_-}$, it follows that
		\begin{equation}\label{eq:Agmon_rho_4_L1}
			\intd \vr_{\opp^2:\opgam\,\opp^2}\, e^\frac{d_1}{2\hbar} \leq C_{\opgam,V} \Nrm{\vr_\opgam}{L^1} . 
		\end{equation}
		Thus, arguing as before, and using the embedding between Schatten norms, the fact that $\sNrm{\indic_{\n{\opp}\neq 0} \n{\opp}^{-2}\opp^2}{\L^\infty} \leq 1$, it gives
		\begin{equation*}
			\Nrm{e^{\beta d_1}\opp^2\,\opgam}{\L^\infty} \leq C_d\, C_{\opgam,V}^{1/2} + e^{\beta} \Nrm{\n{\opp}^2\opgam}{\L^\infty}
		\end{equation*}
		and the result follows as in the proof of Lemma~\ref{lem:Linfty}.
	\end{proof}

	\begin{lem}\label{lem:weight_6}
		In $d \geq 1$, assume that $V$, $\nabla V$, $\nabla^2 V$, $\nabla \Delta V$, $ \Delta^2 V$ are locally bounded. Then, 
		\begin{equation*}
			\Nrm{\n{\opp}^6\opgam}{\L^\infty} \leq \widetilde C_{\opgam,V} \(1+\Nrm{\n{\opp}^4\opgam}{\L^\infty}\) \(1+ \Nrm{V_-}{L^\infty}\)
		\end{equation*}
		with
		\begin{multline*}
			\widetilde C_{\opgam,V} = 2 + 2 \Nrm{\n{V}^3 \opgam}{\L^\infty}+ \hbar\(4 \Nrm{\nabla V\cdot\opp \,\opgam}{\L^\infty} + 3 \Nrm{\nabla (V^2)\cdot\opp\, \opgam}{\L^\infty}\)
			\\
			+ \hbar^2\(2 \Nrm{\Delta V\opgam}{\L^\infty} + 3 \Nrm{V \Delta V\,\opgam}{\L^\infty} + 4 \Nrm{\nabla^2 V:\opp^2\opgam}{\L^\infty} + 4 \Nrm{\n{\nabla V}^2 \opgam}{\L^\infty}\)
			\\
			+ 4\,\hbar^3\Nrm{\Delta\nabla V\cdot\opp\,\opgam}{\L^\infty} + \hbar^4\Nrm{\Delta^2 V\,\opgam}{\L^\infty}.
		\end{multline*}
	\end{lem}

	\begin{proof}
		We write $\n{\opp}^6 \opgam = \n{\opp}^4 H\opgam - \n{\opp}^4 V\opgam$ and then since $\opgam^2=\opgam$
		\begin{equation*}
			\Nrm{\n{\opp}^4 H\opgam}{\L^\infty} = \Nrm{\n{\opp}^4 \opgam\, H}{\L^\infty} \leq \Nrm{\n{\opp}^4 \opgam}{\L^\infty} \Nrm{\opgam\, H}{\L^\infty}
		\end{equation*}
		while by the chain rule
		\begin{multline*}
			\n{\opp}^4 V\opgam = \hbar^4\Delta^2 V\,\opgam + 4i\hbar^3\Delta\nabla V\cdot\opp\,\opgam + 2\hbar^2\Delta V\n{\opp}^2\opgam + 4 \hbar^2\,\nabla^2 V:\opp^2\opgam
			\\
			+ 4i\hbar \,\nabla V\cdot\opp \n{\opp}^2\opgam + V\n{\opp}^4 \opgam \, .
		\end{multline*}
		Observe that we can again use the fact that $\n{\opp}^2 = H-V$ for some terms. This gives for instance
		\begin{equation*}
			\Nrm{\Delta V\n{\opp}^2\opgam}{\L^\infty} \leq \Nrm{\Delta V\opgam}{\L^\infty} \Nrm{\opgam \,H}{\L^\infty} + \Nrm{\Delta V\, V\,\opgam}{\L^\infty}
		\end{equation*}
		and
		\begin{align*}
			\nabla V\cdot\opp \n{\opp}^2\opgam &= \nabla V\cdot\opp \,\opgam\, H - \nabla V\cdot\opp\, V \opgam
			\\
			&= \nabla V\cdot\opp \,\opgam\, H - V\nabla V\cdot\opp\, \opgam + i\hbar \n{\nabla V}^2 \opgam
		\end{align*}
		and, using also the identity $\n{\opp}^2 V = V \n{\opp}^2 - 2i\hbar \,\nabla V\cdot\opp - \hbar^2 \Delta V$,
		\begin{align*}
			V \n{\opp}^4 \opgam &= V \n{\opp}^2 \opgam\, H - V \n{\opp}^2 V \opgam
			\\
			&= V \opgam\, H^2 - 2\,V^2 \opgam\, H - V^3 \opgam + i\hbar \,\nabla V^2\cdot\opp\, \opgam + \hbar^2\, V\Delta V \opgam \, .
		\end{align*}
		The result follows by using the triangle inequality for the operator norm in all these inequalities, combining them, and using each time that $\Nrm{\opgam\,H}{\L^\infty} \leq \Nrm{V_-}{L^\infty}$.
	\end{proof}

\section{Commutator estimates}\label{sec:com}

	The main purpose of this section is the proof of the commutators estimates in Theorem~\ref{thm:commutator:linear}. We assume the interaction potential satisfies the assumptions contained in its statement.

\subsection{Motivation}\label{sec:mot}
	
	In order to motivate the upcoming methods, let us consider the simplest case of the one-dimensional harmonic oscillator $H = - \hbar^2 \frac{\d^2}{\d x^2} + x^2$. As is well-known its spectrum consists of the sequence of eigenvalues $E_n = \(2n+1\) \hbar$ for numbers $n \geq 0 $. Given a chemical potential $\mu \in \hbar \,\N$, let us denote
	\begin{equation*}
		\opgam := \sum_{E_n \leq \mu} \ket{\varphi_n}\bra{\varphi_n}
	\end{equation*}
	where $\(\varphi_n\)_{n\in\N}$ are the Hermite functions at scale $\hbar$. 

	We now claim that the commutator estimates follow as soon as we know that the eigenvalues satisfy the \textit{gap condition}
	\begin{equation}\label{eq:eigenvalue:gap}
		\forall n, k \geq 1,\,\n{E_n - E_k} \geq C \,\hbar \n{n - k}. 
	\end{equation}
	Indeed, let us consider now a self-adjoint operator $A$ on $L^2(\R)$, which we think of either being $A = x$ or $A = - i \hbar\, \frac{\d}{\d x}$. For the purpose of the exposition, let us work with the Hilbert--Schmidt norms which are easier to compute. Indeed, using the fact that $\opgam^2 = \opgam$, we find 
	\begin{equation*}
		\frac{1}{4} \Tr{\n{\com{A, \opgam}}^2} = \Tr{\n{\(\id - \opgam\)A \, \opgam}^2} = \sum_{E_k > \mu} \sum_{E_n \leq \mu} \n{\Inprod{\varphi_k}{A\varphi_n}}^2 .
	\end{equation*}
	Of course, if $A$ was not present we readily obtain $\Inprod{\varphi_n}{\varphi_k} = 0$ due to summation constrains. However, the operator $A$ induces correlations between the different eigenfunctions. These correlations, in fact, decay with $\n{n - k}$ thanks to Inequality~\eqref{eq:eigenvalue:gap}. In order to estimate these correlations we may easily compute using the fact that $\varphi_n$ are eigenfunctions of $H$
	\begin{equation*}
		\n{\Inprod{\varphi_k}{A \varphi_n}}^2 = \frac{\n{\Inprod{\varphi_k}{\com{H,A}\varphi_n}}^2}{\n{E_k - E_n}^2} \leq \frac{C \,\hbar^{-2}}{\n{k - n}^2} \n{\Inprod{\varphi_k}{\com{H,A} \varphi_n}}^2 . 
	\end{equation*}
	Here, we have used the gap condition to obtain some decay with respect to $n$ and $k$. Additionally, thanks to the separation $E_n \leq \mu < E_k$ we have $\n{k - n} \geq \n{k_0 - n}$ for some fixed $k_0 \leq k$. The Cauchy--Schwarz inequality now implies
	\begin{equation*}
		\sum_{E_k > \mu} \sum_{E_n \leq \mu} \n{\Inprod{\varphi_k}{A \varphi_n}}^2 
		\leq C\, \hbar^{-2} \sum_{E_n \leq\mu}
		\frac{\Nrm{\com{H, A} \varphi_n}{L^2}^2}{\n{k_0 - n}^2} \, . 
	\end{equation*}
	Finally, one may verify that for each eigenfunction $\varphi_n$ with $E_n \leq \mu $ we have $\Nrm{\com{H, A} \varphi_n}{L^2} \leq C\, \hbar \, \mu$. All in all, we obtain 
	\begin{equation}\label{eq:comm2}
		\Tr{\n{\com{A, \opgam}}^2} \leq C \sum_{E_n \leq \mu} \frac{\mu^2}{\n{k_0 - n}^2} \, . 
	\end{equation}
	Observe now that the sum on the right-hand side is convergent and uniformly bounded in $\hbar$. On the other hand, in one dimension $\Tr{\opgam^2} = C\, \hbar^{-1}$ and we have, therefore, reduced by a factor $\hbar$ the growth of the commutator. In other words, \eqref{eq:comm2} coincides with \eqref{eq:comms} for $p=2$ and $d=1$. 
	
	The argument presented above can be readily generalized to the harmonic oscillator in any spatial dimension, as we know that the gap condition \eqref{eq:eigenvalue:gap} is verified for such models. Additionally, it can be suitably modified to prove its stronger trace-class variant. 
	
	For more complicated models, we do not expect such strong gap condition to hold pointwise (i.e. for all $n$ and $k$) but at least in an \textit{averaged} sense. More precisely, we replace it with the resolvent inequality (see Lemma~\ref{lem:trace})
	\begin{equation}\label{eq:sum}
		\hd \Tr{\indic_{H\leq 0}\(\hbar - H\)^{-2}} \leq C \,\hbar \, . 
	\end{equation}
	Similar to \eqref{eq:eigenvalue:gap}, the new bound \eqref{eq:sum} states that there is an averaged separation of eigenvalues by a gap of order $\hbar$. As we shall see, it will become sufficient to prove our commutator estimates, which is a suitable modification of the argument just presented.

\subsection{Eigenvalue estimates}

	As explained in the introductory section our starting point is the optimal Weyl law for $C^{1, \alpha}$ potentials. More precisely, it follows from \cite[Theorem~1.5]{mikkelsen_sharp_2023} that for potentials $V$ verifying \eqref{hyp:V} in $d \geq 3$, there exists $\cC_0 >0$ such that for all $ \hbar \in (0,1)$ 
	\begin{equation}\label{eq:weyl0}
		\n{\hd\Tr{\indic_{H \leq 0}} - \int_{\n{\xi}^2 + V (x)\leq 0} \d x\d \xi} \leq \cC_0 \,\hbar \, .
	\end{equation}
	For the next observation, fix some $\nu>0$. Although not explicitly mentioned in~\cite{mikkelsen_sharp_2023}, 
	the constant $\cC_0$ is uniform over potentials which satisfy 
	\begin{equation}\label{eq:C1}
		\Nrm{V}{C^{1,\alpha}(\Omega_{\nu})} < C_{\Omega_\nu} , \quad \text{ with } \quad \Omega_\nu = \{V \leq \nu\} \, , 
	\end{equation}
	where thanks to \eqref{hyp:V} the set $\Omega_\nu$ is compact. Indeed, the proof introduces a regularization $V_\eps(x)$ whose various $C^{k ,\alpha}$ semi-norms depend on $V$ only through~\eqref{eq:C1}. The proof of the Weyl law is then a delicate pseudo-differential operators expansion based on the regularization $V_\eps$. As stated in~\cite{mikkelsen_sharp_2023}, all resulting estimates depend only on the $C^{k, \alpha}$ norms of $V_\eps$, which then depend on $V$ only through $C_{\Omega_\nu}$.

	This observation has the following consequence. Namely, one can find a small enough and fixed $ 0 < \varepsilon_0 \leq 1$ (independent of $\hbar$) such that the bound \eqref{eq:C1} still holds if we replace $V \mapsto V +E $ for all $\n{E}\leq \eps_0$. Thus, the following perturbed version of the result of \cite{mikkelsen_sharp_2023} is valid
	\begin{equation}\label{eq:weyl}
		\n{\hd\Tr{\indic_{H \leq E}} - \int_{\n{\xi}^2 + V(x)\leq E} \d x \d \xi} \leq \cC_0 \,\hbar \, ,
	\end{equation}
	for all $E \in [-\varepsilon_0, \varepsilon_0]$, where $\cC _0 >0$ is the same constant as before. 
	
	In our setting, the main consequence of~\eqref{eq:weyl} is the validity of the following \textit{local eigenvalue estimate}.
	\begin{lem}\label{lem:local_estimate}
		Let $a \in [- \eps_0, 0]$ and $b \in [a, 1]$. Then, in $d \geq 3$
		\begin{equation}\label{eq:local_estimate}
			\hd \tr{\indic_{[a,b]}(H)} \leq \cC_1 \(\n{b- a} + \hbar\) 
		\end{equation}
		where $\cC_1 = 2\, \cC_0 + \(\frac{\omega_d}{2} + \frac{\sfL_{0,d}}{\eps_0}\) \intd \(V-1\)_-^{d/2 -1}$.
	\end{lem}

	\begin{proof}
		Assume first $ b\leq \eps_0$ so that $ - \eps_0 \leq a \leq b \leq \eps_0$. Since $\indic_{[a,b]}(H) = \indic_{H\leq b} - \indic_{H < a}$, Inequality~\eqref{eq:weyl} gives
		\begin{equation*}
			\hd\tr {\indic_{[a,b]}(H)} \leq 2\,\cC_0\, \hbar + \frac{\omega_d}{d}\intd \(b-V\)_+^\frac{d}{2} - \(a-V\)_+^\frac{d}{2}.
		\end{equation*}
		The second term on the right-hand side is now bounded using the fact that for $0\leq \alpha\leq\beta$, $\beta_+^{d/2} - \alpha_+^{d/2} \leq \frac{d}{2} \(\beta - \alpha\) \beta_+^{d/2-1}$, which gives
		\begin{equation*}
			\intd \(b-V\)_+^\frac{d}{2} - \(a-V\)_+^\frac{d}{2} \leq \frac{d}{2} \(b-a\) \intd \(b-V\)_+^{\frac{d}{2}-1}
		\end{equation*}
		which gives the result in this case. Assume now $\eps_0 \leq b \leq 1$ so that $\n{b-a} \geq \eps_0$. Then, it suffices to use the Cwikel--Lieb--Rozenblum inequality~\eqref{eq:CLR} to find 
		\begin{equation*}
			\hd \tr {\indic_{[a,b]}(H)} \leq \hd \Tr{\indic_{H\leq b}} \leq \frac{\sfL_{0,d}}{\eps_0} \n{b-a} \intd \(b-V\)_+^{d/2}.
		\end{equation*}
		This finishes the proof. 
	\end{proof}

	We will use the local eigenvalue estimate both directly in the proof of the commutator estimates, but also to establish the following singular sum estimate. Here and below, $\lambda$ is an auxiliary parameter, which
	we take as $\lambda = \hbar$ in some cases of interest. However, keeping $\lambda$ flexible allows us to improve upon some rates, as well as proving Besov-type estimates. 

	\begin{lem}\label{lem:trace}
		Let $d\geq 3$. Then for any $\lambda \geq \hbar$
		\begin{align*} 
			\hd\Tr{\indic_{H\leq 0} \(\lambda - H\)^{-2}} &\leq \cC_2 \, \lambda^{-1}
			\\ 
			\hd \Tr{\indic_{H\leq 0} \(\lambda - H\)^{-1}} &\leq \cC_2 \(1+\ln\!\(1+\tfrac{\eps_0}{\lambda}\)\)
		\end{align*}
		where $\cC_2 = 6 \,\cC_1 + \frac{\sfL_{0,d}}{\eps_0} \int V_-^{d/2}$.
	\end{lem}

	\begin{remark}
		Essentially, in the proof we compare the trace with the Riemann sum corresponding to the integral $\int_{-\infty}^0 (\lambda-u)^{-2} \d u = \lambda^{-1}$.
	\end{remark}

	\begin{proof}[Proof of Lemma~\ref{lem:trace}]
		We will only give the detailed proof of the first item. The proof of the second one is completely analogous. Observe first that when $\lambda > \eps_0$, we can just use the fact that $\(\lambda - H\)^{-2} \leq \lambda^{-2} \leq \(\eps_0\,\lambda\)^{-1}$ and the Cwikel--Lieb--Rozenblum bound~\eqref{eq:CLR} to get
		\begin{equation*} 
			\hd\Tr{\indic_{H\leq 0} \(\lambda - H\)^{-2}} \leq \frac{\sfL_{0,d}}{\eps_0\,\lambda} \intd V_-^{d/2}.
		\end{equation*}

		Now we assume $\lambda\leq \eps_0$ and we decompose $(-\infty, 0] = (-\infty , -\eps_0] \cup (-\eps_0, 0]$. On the first region, we can again use the Cwikel--Lieb--Rozenblum bound to get
		\begin{equation*}
			\hd\Tr{\indic_{H\leq -\eps_0} \(\lambda - H\)^{-2}} \leq \frac{\hd}{\(\lambda+\eps_0\)^2} \Tr{\indic_{H\leq - \eps_0}} \leq \frac{\sfL_{0,d}}{3\,\lambda\,\eps_0} \intd V_-^{d/2}.
		\end{equation*} 
		We now focus on the interval $(-\eps_0, 0]$. To this end, consider an integer $n \geq 1 $, soon to be chosen (the reader should think $n \sim \lambda^{-1}$), and let us partition the interval $[-\eps_0, 0]$ into $n$ disjoint intervals of equal length $\frac{\eps_0}{n}$ to obtain 
		\begin{equation*}
			\indic_{(-\eps_0, 0]}(H) = \sum_{j = 0}^{n-1} \indic_{(-\frac{j+1}{n}\,\eps_0,-\frac{j}{n}\,\eps_0]}(H) \, .
		\end{equation*}
		Note that the map $r \mapsto (\lambda - r)^{-2}$ is increasing on $\R_-$. Hence, we can estimate $H \leq -\frac{j}{n}\,\eps_0$ over each interval
		\begin{equation*} 
			\Tr{\indic_{(-\eps_0, 0]}(H) \(\lambda - H\)^{-2}} \leq \sum_{j = 0}^{n-1} \frac{1}{\(\lambda + \frac{j}{n}\, \eps_0\)^2} \Tr{\indic_{(-\frac{j+1}{n}\,\eps_0,-\frac{j}{n}\,\eps_0]}(H)}.
		\end{equation*}
		Next, we use the local eigenvalue estimate~\eqref{eq:local_estimate} to 
		estimate the traces. We obtain 
		\begin{equation*} 
			\hd\Tr{\indic_{(-\eps_0, 0]}(H) \(\lambda - H\)^{-2}} \leq \cC_1 \(\eps_0 + n \hbar\) \frac{1}{n} \sum_{j = 0}^{n-1} \frac{1}{\(\lambda + \frac{j}{n}\, \eps_0\)^2}
		\end{equation*}
		and we estimate the Riemann sum via an integral comparison using the fact that $r\mapsto(\lambda+r)^{-2}$ is decreasing. This gives, thanks to $\lambda \geq 0$
		\begin{equation*}
			\frac{1}{n}\sum_{j = 0}^{n-1} \frac{1}{\(\lambda + \frac{j}{n}\, \eps_0\)^2} \leq \frac{1}{n\,\lambda^2} + \int_0^1 \frac{\d r}{\(\lambda+ \eps_0\, r\)^2} = \frac{1}{n\,\lambda^2} + \frac{1}{\lambda\(\lambda+\eps_0\)}
			\leq \frac{1}{\lambda} \(\frac{1}{n\lambda} + \frac{1}{\eps_0}\) .
		\end{equation*}
		We now choose an integer $n \in \N$ such that $\eps_0 \leq n \lambda \leq 2 \, \eps_0$ and use $\hbar \leq \lambda$ to find 
		\begin{equation*} 
			\hd\Tr{\indic_{(-\eps_0, 0]}(H) \(\lambda - H\)^{-2}} \leq \frac{\cC_1}{\lambda} \(\eps_0 + n \lambda\) \(\frac{1}{n\lambda} + \frac{1}{\eps_0}\) \leq 6 \,\cC_1 \, .
		\end{equation*} 
		This finishes the proof. 
	\end{proof}

\subsection{General Hilbert--Schmidt bounds}

	For the proof of Theorem~\ref{thm:commutator:linear} we will develop general commutator estimates for characteristic functions of self-adjoint operators. In particular, they provide bounds for $\com{A,\opgam}$ in terms of more tractable objects. See Proposition~\ref{prop:comm_HS} and~\ref{prop:comm2} for Hilbert--Schmidt and trace-class bounds. Additionally, we will need to prove more technical counterparts in Proposition~\ref{prop:comm_frac} and Proposition~\ref{prop:comm_HS_weight}.
	
	\begin{prop}[Hilbert--Schmidt bounds]\label{prop:comm_HS} 
		Let $\opgam = \indic_{H\leq 0}$ with $H = -\hbar^2 \Delta + V$ in $d \geq 3$, and let $A$ be a normal operator on $L^2(\Rd)$. Then, for all $\lambda \geq \hbar$
		\begin{equation}\label{eq:comm_HS}
			\hd\Tr{\n{\com{A,\opgam}}^2} \leq \CHS \(\lambda
			\Nrm{\opgam\, A}{\L^\infty}^2 +
			\frac{1}{\lambda}
			\Nrm{\opgam \com{A,H} \(\id-\opgam\)}{\L^\infty}^2\)
		\end{equation}
		where $\CHS = 2\(\cC_1 + \cC_2\)$.
	\end{prop}

	\begin{remark}
		Taking $\lambda = \hbar$ and either $A = x$ or $A = \opp = - i \hbar \nabla$, then $\lambda \Nrm{\opgam\, A}{\L^\infty}^2$ becomes $O(\hbar)$ while for the second terms one can use $0 \leq \opgam \leq \id$ and the fact that 
		\begin{equation*}
			\Nrm{\com{H,x} \opgam}{\L^\infty}^2 \leq 2\,\hbar^2 \Nrm{\opp\, \opgam}{\L^\infty}^2 \quad \text{ and } \quad
			\Nrm{\com{H, \opp}\opgam}{\L^\infty}^2 \leq \hbar^2 \Nrm{\nabla V \opgam}{\L^\infty}^2 \, .
		\end{equation*}
		Thus, the right-hand side of Inequality~\eqref{eq:comm_HS} becomes $O(\hbar)$ once we can control $\Nrm{\opp\, \opgam}{\L^\infty}$, $\Nrm{x\, \opgam}{\L^\infty}$ and $\Nrm{\nabla V\,\opgam}{\L^\infty}$ uniformly in $\hbar$.
	\end{remark}

	\begin{remark}
		For the proof, we establish the more general inequality: for all $\lambda > 0$ and $d \geq 1$
		\begin{equation}\label{eq:HS}
			\Tr{\n{\com{A,\opgam}}^2} \leq 2 \Tr{\indic_{0< H< \lambda} \n{\opgam\, A}^2} + 2 \Nrm{\opgam \com{A, H} \(\id - \opgam\)}{\L^\infty}^2 \Tr{\opgam \(\lambda - H\)^{-2}} .
		\end{equation} 
		This bound can be formulated in an abstract setting. Namely, the bound is also valid for a sequence of self-adjoint operators $(H_\hbar)_{\hbar \in (0,1)}$ on a Hilbert space $\mathscr{H}$. See Appendix~\ref{appendix}. 
%		 The second bound would follow from the following local eigenvalue estimate
%		\begin{equation}
%			\hd \tr {\indic_{[a,b]}(H_\hbar)} \leq C \(\n{b- a} + h\)
%		\end{equation}
%		for all $ a \leq b \leq \epsilon_0$, where $C>0$ is independent of $a$, $b$ and $h$. 
	\end{remark}
	
	In the following proof, we rely on the fact that $H$ has discrete spectrum and prove the intermediate inequality~\eqref{eq:HS} using an eigenbasis expansion. The Hilbert--Schmidt estimate~\eqref{eq:HS} can be established without requiring that $H$ has discrete spectrum; see Appendix~\ref{appendix} for an alternative proof. Currently, we cannot adapt that alternative approach to prove the corresponding trace-class bounds in Subsection~\ref{subsection:trace}. Thus, we include here the proof with an eigenbasis expansion as it will also set up the stage for the upcoming trace-class bounds.

	\begin{proof}[Proof of Proposition~\ref{prop:comm_HS}]
		First, observe that using $\opgam \(\id - \opgam\)=0$, the fact that $A$ is normal and the cyclicity of the trace gives
		\begin{equation*}
			\Tr{\n{\com{A,\opgam}}^2} = 2 \Tr{A^* \,\opgam\, A\(\id - \opgam\)} .
		\end{equation*}
		Decomposing $\id - \opgam = \indic_{0< H < \lambda} + \indic_{H \geq \lambda}$ and using the fact that $A^*\opgam\, A = \n{\opgam\, A}^2$, this can be written
		\begin{equation}\label{eq:HS1}
			\Tr{\n{\com{A,\opgam}}^2} = 2 \Tr{\n{\opgam\, A}^2 \indic_{0< H < \lambda}} + 2 \Tr{A^* \, \opgam\, A \, \indic_{H \geq \lambda}} .
		\end{equation}
		On the right-hand side, the first term is already the first term of Equation~\eqref{eq:HS}. We now proceed to estimate the second one. To this end, we recall that the spectrum of $H$ is discrete and given by an increasing sequence of eigenvalues $(\lambda_j)_{j=0}^\infty \subset \R$. Thus, 
		\begin{equation*} 
			\opgam = \sum_{\lambda_k \leq 0} \opProj_k
			\quad \text{ where } \quad \opProj_k = \indic_{H = \lambda_k} \,,
		\end{equation*}
		which, using also the cyclicity of the trace, yields
		\begin{equation}\label{eq:}
			\Tr{A^*\, \opgam \, A \, \indic_{H \geq \lambda}} = \sum_{\lambda_k \geq \lambda}\sum_{\lambda_j < 0} \Tr{\opProj_k\, A^* \,\opProj_j\, A\, \opProj_k}.
		\end{equation} 
		Next, observe that by the definition of the projectors $\opProj$,
		\begin{equation}\label{eq:proj}
			\opProj_j \,A\, \opProj_k = \frac{\opProj_j \com{A, H} \opProj_k}{\lambda_k - \lambda_j}
		\end{equation}
		from which it follows that, bounding $(\lambda_k - \lambda_j)^{-2} \leq (\lambda - \lambda_j)^{-2}$,
		\begin{align*}
			\Tr{A^* \, \opgam \, A \, \indic_{H \geq \lambda}} &= \sum_{\lambda_j < 0} \sum_{\lambda_k \geq \lambda} \frac{1}{(\lambda_k - \lambda_j)^2} \Tr{\opProj_k \com{A,H}^* \opProj_j \com{A ,H} \opProj_k}
			\\
			&\leq \sum_{\lambda_j < 0} \sum_{\lambda_k \geq \lambda} \frac{1}{(\lambda - \lambda_j)^2} \Tr{\opProj_k \com{A,H}^* \opProj_j \com{A ,H} \opProj_k} .
		\end{align*}
		By cyclicity of the trace and by writing $\opgam \(\lambda - H\)^{-2}$ and $\indic_{H \geq \lambda}$ in terms of the projectors $\opProj$, this can be written
		\begin{equation*}
			\Tr{A^* \, \opgam \, A \, \indic_{H \geq \lambda}} \leq \Tr{\indic_{H \geq \lambda} \com{A,H}^* \opgam \(\lambda - H\)^{-2} \com{A ,H} \indic_{H \geq \lambda}} ,
		\end{equation*}
		Using $\opgam = \opgam^2$ twice, the operator inequality $\indic_{H \geq \lambda} \leq \id - \opgam $, and the H\"older inequality for the trace yields
		\begin{align*}
			\Tr{A^* \, \opgam \, A \, \indic_{H \geq \lambda}} &\leq \Nrm{\opgam \com{A, H} \(\id - \opgam\)}{\L^\infty}^2 \Tr{\opgam \(\lambda - H\)^{-2}} .
		\end{align*}
		Plugging this estimate back in Inequality~\eqref{eq:HS1} proves Formula~\eqref{eq:HS}. Proposition~\ref{prop:comm_HS} then follows for $\lambda \geq \hbar$ thanks to Lemma~\ref{lem:local_estimate} and Lemma~\ref{lem:trace}.
	\end{proof}

\subsection{General trace-class bounds}
\label{subsection:trace}

	For the case of the trace-class commutator estimates, we state two propositions. The first one will be useful for the position commutator estimates. The second one will be useful for the momentum estimates.

	\begin{prop}[First trace-class bounds]\label{prop:comm2}
		Let $\opgam = \indic_{H\leq 0}$ with $H = -\hbar^2 \Delta + V$ in $d \geq 3$ and let $A$ be a normal operator on $L^2(\Rd)$. Then, for all $\lambda \geq \hbar$
		\begin{align*} %\label{eq:prop:comm2:1}
			\hd \Tr{\n{\com{A, \opgam}}} &\leq \CTr \(\lambda \Nrm{A\,\opgam}{\L^\infty} + 
			\frac{1}{\lambda} \Nrm{\(\id - \opgam\) \com{\com{A,H},H} \opgam}{\L^\infty}\)
			\\ %\label{eq:prop:comm2:2}
			\hd \Tr{\n{\com{A, \opgam}}} &\leq \CTr \(\lambda
			\Nrm{A\,\opgam}{\L^\infty} + 
			\(1+\ln\!\(1+ \frac{\eps_0}{\lambda}\)\) \Nrm{(\id - \opgam) \com{A, H} \opgam}{\L^\infty}\) ,
		\end{align*}
		where $\cC_{\rm Tr} = 2\(\cC_1 + \cC_2\)$.
	\end{prop}

	\begin{remark}
		Again, in the case $\lambda = \hbar$, the first term in the inequalities is $O(\hbar)$ and, using $0 \leq \opgam \leq \id$, we obtain for the second terms
		\begin{align*}
			\Nrm{\com{\com{x, H} , H} \opgam}{\L^\infty} &\leq 2\, \hbar^2 \Nrm{\nabla V\,\opgam}{\L^\infty}
			\\
			\Nrm{\com{\com{\opp, H}, H}\opgam}{\L^\infty} &\leq \hbar^2 \Nrm{\(\opp \cdot \nabla^2 V + \nabla^2 V \cdot \opp\) \opgam}{\L^\infty}
			\\
		\Nrm{\com{\opp, H}\opgam}{\L^\infty} &\leq \hbar \Nrm{\nabla V \, \opgam}{\L^\infty}\, . 
		\end{align*}
		The first bound gives optimal commutators estimates for $A = x$. The second bound with $A = \opp$ can be bounded using $\opp\cdot \nabla^2 V = \nabla^2 V \cdot\opp - i \hbar \, \nabla\Delta V$, but this involves $3$ gradients on $V$. See Proposition~\ref{prop:comm_frac} for an alternative. On the other hand, if $V$ has only one derivative, the third bound with $A = \opp$ will still give a bound but with an additional $\n{\ln\hbar}$.
	\end{remark}
 
	\begin{remark} 
		We will prove the general estimates: for all $\lambda >0$ and $d \geq 1$
		\begin{align}\label{eq:tr}
			\tr{\n{\com{A, \opgam}}} &\leq 2 \Tr{\n{\indic_{0< H< \lambda} \,A\,\opgam}} + 2 \Nrm{\(\id-\opgam\) \com{\com{A,H},H} \opgam}{\L^\infty} \Tr{\(\lambda - H\)^{-2} \opgam}
			\\\label{eq:tr2}
			\tr{\n{\com{A, \opgam}}} &\leq 2 \Tr{\n{\indic_{0< H< \lambda} \,A\,\opgam}} + 2 \Nrm{\(\id-\opgam\) \com{A,H} \opgam}{\L^\infty} \Tr{\(\lambda - H\)^{-1} \opgam} .
		\end{align}
		As in the Hilbert--Schmidt case, these bounds can also be formulated in an abstract setting. 
	\end{remark}
 
	\begin{proof}
		For the proof let us write $\Nrm{A}{p}^p = \Tr{\n{A}^p}$. We only prove \eqref{eq:tr} in detail as the second one is analogous. Similarly as in the proof for the Hilbert--Schmidt estimates, we write the commutator using the operator $\id-\opgam$ as $\com{A,\opgam} = \(\id-\opgam\)A\,\opgam - \(\id-\opgam\)\opgam\, A$ and we write $\id-\opgam = \indic_{0<H<\lambda} + \indic_{H\geq\lambda}$ to get
		\begin{equation*}
			\Nrm{\com{A, \opgam}}{1} \leq 2 \Nrm{\(\id - \opgam\) A\, \opgam}{1} \leq 2 \Nrm{\indic_{0<H<\lambda}\, A \, \opgam}{1} + 2 \Nrm{\indic_{H\geq\lambda}\, A\,\opgam}{1} . 
		\end{equation*}
		It remains to bound the second term on the right-hand side. To this end, we proceed similarly as in the proof of the Hilbert--Schmidt estimates (with the same notations) but take here two commutators. Let $\opb$ be a compact operator on $L^2(\Rd)$. Then, expanding the projection operators and using Identity~\eqref{eq:proj} yields
		\begin{equation}\label{eq:comm}
		\begin{aligned}
			\Tr{\indic_{H\geq \lambda} \, A\, \opgam\, \opb}
			%&= \sum_{\lambda_k \geq \lambda} \sum_{\lambda_j \leq 0} \Tr{\opProj_k \, A\, \opProj_j \, \opb}
			%\\
			&= \sum_{\lambda_k \geq \lambda} \sum_{\lambda_j \leq 0} \frac{1}{(\lambda_k - \lambda_j)^2} \Tr{\opProj_k \com{\com{A,H},H} \opProj_j \,\opb}
			\\
			&= \sum_{\lambda_k \geq \lambda} \sum_{\lambda_j \leq 0} \frac{1}{(\lambda_k - \lambda_j)^2} \Tr{\opProj_k \, \opc^* \, \opProj_j \,\opb \, \opProj_k}
			\\
		\end{aligned}
		\end{equation}
%	\margin{Why $\opc^*$ instead of $\opc$ ??\\
%	-- Makes the tracial Cauchy Schwarz of the form $c ^* b $. Yes but perhaps less clear for the reader?}
		where we now denote $\opc := \com{\com{A, H},H}^*$. First, we observe that thanks to the Cauchy--Schwarz inequality over the trace 
		\begin{equation}\label{eq:comm_CS}
			\Tr{\n{\opProj_k \,\opc^*\, \opProj_j \,\opb}} \leq \(\Tr{\opProj_k\, \opc^* \,\opProj_j \,\opc\, \opProj_k}\)^{\frac{1}{2}} \(\Tr{\opProj_k \, \opb^* \, \opProj_j \,\opb\, \opProj_k}\)^{\frac{1}{2}} .
		\end{equation}
		Next, we bound the denominator using $\lambda_k \geq \lambda$ in Identity~\eqref{eq:comm}, exchange the sums, use the bound~\eqref{eq:comm_CS}, and the Cauchy--Schwarz inequality over $\lambda_k \in \sigma(H)$ to obtain 
		\begin{align*}
			\n{\Tr{\indic_{H\geq \lambda} \, A\, \opgam\, \opb}}
			% &\leq \sum_{\lambda_j\leq0} \frac{1}{(\lambda - \lambda_j)^2}\sum_{\lambda_k\geq \lambda} \n{\tr \opProj_k \opc^* \opProj_j \opb\opProj_k}
			%\\ 
			&\leq \sum_{\lambda_j\leq 0} \frac{1}{\(\lambda - \lambda_j\)^2} \sum_{\lambda_k\geq \lambda} \(\Tr{\opProj_k \, \opc^* \, \opProj_j \, \opc \, \opProj_k}\)^{\frac{1}{2}} \(\Tr{\opProj_k \, \opb^* \, \opProj_j \, \opb \, \opProj_k}\)^{\frac{1}{2}}
			\\
			&\leq \sum_{\lambda_j\leq0} \frac{1}{\(\lambda - \lambda_j\)^2} \(\Tr{\opc^*\, \opProj_j \opc \, \indic_{H\geq \lambda}}\)^{\frac{1}{2}} \(\Tr{\opb^* \, \opProj_j \, \opb \, \indic_{H\geq \lambda}}\)^{\frac{1}{2}} . 
		\end{align*}
		Finally, we distribute one power $(\lambda - \lambda_j)^{-2}$ in each factor and do the Cauchy--Schwarz over $\lambda_j \in \sigma(H)$ to get 
		\begin{align*}
			\n{\Tr{\indic_{H\geq \lambda} \, A\, \opgam\, \opb}} &\leq \(\Tr{\opc^* \opgam \(\lambda - H\)^{-2} \opc \, \indic_{H\geq \lambda}}\)^{\frac{1}{2}} \(\Tr{\opb^* \opgam \(\lambda - H\)^{-2} \opb \, \indic_{H\geq \lambda}}\)^{\frac{1}{2}}
			\\ 
			&\leq \Nrm{\indic_{H\geq \lambda} \com{\com{A,H},H} \opgam \(\lambda - H\)^{-1}}{2} \Nrm{\indic_{H\geq \lambda} \,\opb\, \opgam \(\lambda - H\)^{-1}}{2} . 
		\end{align*}
		The two factors can now be bounded in the following way
		\begin{align*}
			\Nrm{\indic_{H\geq \lambda} \,\opb\, \opgam \(\lambda - H\)^{-1}}{2} &\leq \Nrm{\opb}{\L^\infty} \sNrm{\opgam \(\lambda - H\)^{-1}}{2}
			\\
			\Nrm{\indic_{H\geq \lambda} \com{\com{A,H},H} \opgam \(\lambda - H\)^{-1}}{2} &\leq \Nrm{\opgam \com{\com{A,H},H} \(\id - \opgam\)}{\L^\infty} \sNrm{\opgam \(\lambda - H\)^{-1}}{2}\, .
		\end{align*}
		Observe that $\sNrm{\opgam \(\lambda - H\)^{-1}}{2}^2 = \tr{\(\lambda - H\)^{-2} \opgam}$. Thus, we now take the supremum over all compact operator with $\Nrm{\opb}{\L^\infty}\leq 1$ which leads to Inequality~\eqref{eq:tr}.

		The proof of the estimate~\eqref{eq:tr2} is identical, one simply uses one commutator instead of two. This has the effect of only introducing a singular factor $(\lambda_k - \lambda_m)^{-1}$ in Equation~\eqref{eq:comm} which then leads to $ \tr(\(\lambda - H\)^{-1}\opgam)$.
		
		The claim of the proposition then follows by combining inequalities~\eqref{eq:tr} and~\eqref{eq:tr2} with Lemma~\ref{lem:local_estimate} and~\ref{lem:trace} with $\lambda \geq \hbar$.
	\end{proof}

	For the following proposition we introduce the shifted Hamiltonian
	\begin{equation}
		H_\mu := H + \mu\geq 0
		\quad \text{ for } \quad
		\mu > \mu_0 := \Nrm{V_-}{L^\infty} \geq \n{\inf \sigma (H)} . 
	\end{equation}
	
	\begin{prop}[Second trace-class bounds]\label{prop:comm_frac}
		Let $\opgam = \indic_{H\leq 0}$ with $H = -\hbar^2 \Delta + V$ in $d \geq 3$, and let $A$ be a normal operator on $L^2(\Rd)$. Then, for all $\lambda \geq \hbar $ and $\mu> \mu_0$
		\begin{equation*}
			\hd \Tr{\n{\com{A, \opgam}}} \leq
			\CTr' \(\lambda \Nrm{A\,\opgam}{\L^\infty} + \frac{1}{\lambda} \Nrm{\(\id-\opgam\)\com{\com{A,H},H_\mu^{1/2}}\opgam}{\L^\infty}\). 
		\end{equation*}
		Here $\CTr' = 2 \,\cC_1 + 4 \sqrt{1 +\mu}\, \cC_2$.
	\end{prop}
% 
%	\begin{remark}
%		Here, the general version corresponds to
%		\begin{multline*}
%			\Tr{\n{\com{A ,\opgam}}} \leq 2 \tr\n{\indic_{0< H< \lambda} \,A\,\opgam}
%			\\
%			+ 4 \sqrt{\lambda + \mu} \Nrm{\(\id-\opgam\) \com{\com{A,H},H_0^{1/2}} \opgam}{\L^\infty} \Tr{\(\lambda - H\)^{-2} \opgam} 
%		\end{multline*}
%		for $d \geq 1 $, $\lambda > 0$ and $\mu > \mu_0 $. The proposition follows by taking $\lambda \leq \hbar $ in Lemma~\ref{lem:local_estimate} and~\ref{lem:trace}. 
%	\end{remark}

	\begin{proof}[Proof]
		Similarly as in the proof of Proposition~\ref{prop:comm2}, we first write
		\begin{equation*}
			\Nrm{\com{A, \opgam}}{1} \leq 2 \Nrm{\(\id - \opgam\) A\, \opgam}{1} \leq 2 \Nrm{\indic_{0<H<\lambda}\, A \, \opgam}{1} + 2 \Nrm{\indic_{H\geq\lambda}\, A\,\opgam}{1} . 
		\end{equation*}
		To bound the second term on the right-hand side, we take $\opb$ a compact operator on $L^2(\Rd)$. We now use a generalized version of \eqref{eq:proj} suitable for the square root functions. That is, 
		\begin{equation*}
			\opProj_k \,A\, \opProj_j = \frac{\opProj_k \com{\com{A,H},\(H+\mu\)^s} \opProj_j}{\(\lambda_k - \lambda_j\) \(\(\lambda_k+\mu\)^s - \(\lambda_j+\mu\)^s\)} \, . 
		\end{equation*}
		Together with the identity $\frac{1}{\sqrt b - \sqrt a} \leq \frac{\sqrt b + \sqrt a}{b-a}$ we obtain 
		\begin{equation*}
			\Tr{\indic_{H\geq \lambda} \, A\, \opgam\, \opb} \leq \sum_{\lambda_k \geq \lambda} \sum_{\lambda_j \leq 0} \frac{2 \sqrt{\lambda+\mu}}{(\lambda_k - \lambda_j)^2} \Tr{\opProj_k\com{\com{A,H},H_\mu^s} \opProj_j\,\opb}.
		\end{equation*}
		The proof then follows by the same steps as in the proof of Proposition~\ref{prop:comm2}.
	\end{proof}

\subsection{Proof of Theorems~\ref{thm:commutator:linear} and~\ref{thm:comm:Besov} -- Application to semiclassical gradient estimates}

	Let us now apply the results of our general commutator estimates. More precisely, we will evaluate the operator $A$ with the position operator $x$, the momentum operator $\opp = - i \hbar \nabla$ the shift operator $\tau_z$ and obtain bounds for the quantum gradients of $\opgam$.

	With our notations for the quantum gradients~\eqref{eq:quantum_gradients} and the scaled Schatten norms~\eqref{eq:Schatten}, we recall that 
	\begin{equation*}
		\Nrm{\Dhv \opgam}{\L^2}^2 = \(2\pi\)^2 h^{d-2} \Tr{\n{\com{x , \opgam}}^2},
		\qquad
		\Nrm{\Dhx \opgam}{\L^2}^2 = \(2\pi\)^2 h^{d-2} \Tr{\n{\com{\opp, \opgam}}^2}.
	\end{equation*}
	and
	\begin{equation*}
		\Nrm{\Dhv \opgam}{\L^2}^2 = \Nrm{\Dv f_\opgam}{L^2}^2 ,
		\qquad
		\Nrm{\Dhx \opgam}{\L^2}^2 = \Nrm{\Dx f_\opgam}{L^2}^2.
	\end{equation*}
	We then have the following.
	\begin{prop}\label{prop:linear:comm_HS}
		Let $\hbar\in (0,1)$ and $\opgam = \indic_{H\leq 0}$ with $H = -\hbar^2 \Delta + V$ in $d \geq 3$. Then, it holds that 
		\begin{align*}
			\Nrm{\Dhv \opgam}{\L^2}^2 &\leq \frac{\CHS}{\hbar} \(\Nrm{x\,\opgam}{\L^\infty}^2 + 4 \Nrm{\opp \, \opgam}{\L^\infty}^2\)
			\\
			\Nrm{\Dhx \opgam}{\L^2}^2 &\leq \frac{\CHS}{\hbar} \(\Nrm{\opp\,\opgam}{\L^\infty}^2 + \Nrm{\nabla V \opgam}{\L^\infty}^2\) \, . 
		\end{align*}
		Additionally, for any $z \in \Rdd$
		\begin{equation*}
			\Nrm{\sfT_z\opgam-\opgam}{\L^2} \leq \cD_\opgam
			\min\!\(\tfrac{\n{z}}{\sqrt{\hbar}}, \sqrt{\n{z}}, 1\)
		\end{equation*}
		where $\mathcal D_\opgam = \CHS \(1 + \Nrm{\nabla V \opgam}{\L^\infty}^2 + \Nrm{x\,\opgam}{\L^\infty}^2 + 5 \Nrm{\opp \, \opgam}{\L^\infty}^2\) + \hd \Tr{\opgam}$. If $\mathcal D_\opgam$ is bounded uniformly in $\hbar$, it implies in particular that there exists $C$ independent of $\hbar$ such that
		\begin{equation*}
			\Nrm{\opgam}{\cW^{1,2}} \leq \frac{C}{\sqrt{\hbar}} \quad \text{ and } \quad \Nrm{\opgam}{\cB^{1/2}_{2,\infty}} \leq C.
		\end{equation*}
	\end{prop}

	Let us now state our main results for the semiclassical gradients~\eqref{eq:quantum_gradients} in the trace-class norm. Again we recall that
	\begin{align*}
		\Nrm{\Dhx \opgam}{\L^1} = 2\pi \,h^{d-1} \tr{\n{\com{\opp, \opgam}}}
		\quad
		\text{ and }
		\quad
		\Nrm{\Dhv \opgam}{\L^1} = 2\pi \,h^{d-1} \tr{\n{\com{x, \opgam}}} \, .
	\end{align*}

	\begin{prop}\label{prop:linear:commTR}
		Let $\hbar \in (0,1)$ and $\opgam = \indic_{H\leq 0}$ with $H = -\hbar^2 \Delta + V$ in $d \geq 3$. Then,
		\begin{align*} %\label{eq:comm1_x}
			\Nrm{\Dhv\opgam}{\L^1} &\leq \CTr \(\Nrm{x\,\opgam}{\L^\infty} + 2 \Nrm{\nabla V \opgam}{\L^\infty}\)
			\\ %\label{eq:comm1_x2}
			\Nrm{\Dhx\opgam}{\L^1}
			&\leq \CTr \(\Nrm{\opp \,\opgam}{\L^\infty} + \(1 + \ln(1 + \eps_0/\hbar)\) \Nrm{\nabla V \opgam}{\L^\infty}\).
		\end{align*} 
		Additionally, for all $\mu > \mu_0 = \Nrm{V_-}{L^\infty}$ 
		\begin{equation*} %\label{eq:comm1_p}
			\Nrm{\Dhx \opgam}{\L^1} \leq \CTr' \(\Nrm{\opp\,\opgam}{\L^\infty} + 2\, \sNrm{\nabla^2V \opgam}{\L^\infty} + \tfrac{2}{\sqrt{\mu}} \sNrm{\nabla^2V \cdot \opp\, \opgam}{\L^\infty}\) .
		\end{equation*}
	\end{prop}

	\begin{remark}
		It follows immediately from Identity~\eqref{eq:fund_thm} that
		\begin{equation*}
			\Nrm{\sfT_z\opgam-\opgam}{\L^1} \leq \(\Nrm{\Dhx\opgam}{\L^1} + \Nrm{\Dhv\opgam}{\L^1}\) \n{z},
		\end{equation*}
		and it follows from the fact that $\intd \nabla \vr_\opgam\,\varphi = \hd \Tr{\com{\nabla,\opgam}\varphi}$ that
		\begin{equation}\label{eq:rho_W11}
			\Nrm{\nabla\vr_\opgam}{L^1} \leq \Nrm{\Dhx{\opgam}}{\L^1}.
		\end{equation}
	\end{remark}

	\begin{proof}[Proof of Proposition~\ref{prop:linear:comm_HS}]
		Taking $A = \opp $ or $A = x $ in Proposition~\ref{prop:comm_HS} gives directly the two first inequalities. Next, consider $z = (y,\xi)\in \Rdd$ so that, by~\cite[Formula~(28)]{lafleche_quantum_2024}, we have 
		\begin{equation}\label{eq:fund_thm}
			\sfT_z\opgam - \opgam = z\cdot\int_0^1 \sfT_{\theta z} \Dh\opgam \d\theta 
		\end{equation}
		where $\Dh\opgam = (\Dhx\opgam,\Dhv\opgam)$. Thus, we find in Schatten norm 
		\begin{equation*}
			\Nrm{\sfT_z\opgam - \opgam}{\L^2}^2 \leq \CHS\, \frac{\n{z}^2}{\hbar} \(5 \Nrm{\opp\,\opgam}{\L^\infty}^2 + \Nrm{x\,\opgam}{\L^\infty}^2 + \Nrm{\nabla V \opgam}{\L^\infty}^2\) .
		\end{equation*}
		This is sufficient to analyze $\n{z} \leq \hbar$. Consider now $\hbar \leq \n{z} \leq 1$. Since $\sfT_z\opgam = \tau_z\, \opgam\, \tau_z^*$ with $\tau_z$ unitary, it follows that
		\begin{equation*}
			\Nrm{\sfT_z\opgam - \opgam}{\L^2} = \Nrm{\tau_z\, \opgam \, \tau_z^* - \opgam \,\tau_z \,\tau_z^*}{\L^2} = \Nrm{\com{\tau_z, \opgam}}{\L^2} . 
		\end{equation*}
		We now take $\lambda = \n{z} \leq 1$ in Inequality~\eqref{eq:comm_HS} with $A = \tau_z$. Thanks to Lemma~\ref{lem:local_estimate} and~\ref{lem:trace} we then get 
		\begin{equation*}
			\Nrm{\sfT_z\opgam - \opgam}{\L^2}^2 \leq \CHS \(\n{z} \Nrm{\opgam}{\L^\infty}^2 + \frac{1}{\n{z}} \Nrm{\(\sfT_zH-H\)\opgam}{\L^\infty}^2\) .
		\end{equation*}
		Using Formula~\eqref{eq:fund_thm} for $H$ gives
		\begin{equation*}
			\Nrm{\(\sfT_zH-H\)\opgam}{\L^\infty}^2 \leq \n{z}^2 \(2\Nrm{\opp\,\opgam}{\L^\infty}^2 + \Nrm{\nabla V \opgam}{\L^\infty}^2\) .
		\end{equation*}
		This yields
		\begin{equation*}
			\Nrm{\sfT_z\opgam - \opgam}{\L^2}^2 \leq \CHS \n{z}\(\Nrm{\opgam}{\L^\infty}^2 + 2 \Nrm{\opp\,\opgam}{\L^\infty}^2 + \Nrm{\nabla V \opgam}{\L^\infty}^2\) .
		\end{equation*}
		Finally, for $\n{z} \geq 1$, it is sufficient to use the triangle inequality for Schatten norms and the fact that $\sfT_z$ is an isometry to get
		\begin{equation*}
			\Nrm{\sfT_z\opgam - \opgam}{\L^2}^2 \leq 2 \Nrm{\opgam}{\L^2}^2 = 2 \,\hd\Tr{\opgam} . 
		\end{equation*}
		This finishes the proof after we collect all the constants in $\mathcal D _\opgam$.
	\end{proof}

	We need the following lemma to evaluate the commutator with the fractional operator in the trace norm. In practice, we take $A = \nabla$, which yields $\com{\nabla,H} = \nabla V$. Thus it suffices to compute the commutator with the fractional operator in the case of multiplication operators.

	\begin{lem}[Fractional commutator]\label{lem:frac}
		Let $d \geq1$, $\mu > \mu_0 = \Nrm{V_-}{L^\infty}$ and $\phi \in C^1(\Rd)$ seen as an operator of multiplication. Then
		\begin{equation*} 
			\Nrm{\(\id - \opgam\) \com{\(H+\mu\)^{1/2}, \phi} \opgam}{\L^\infty} \leq 2\,\hbar \(\Nrm{\nabla \phi\, \opgam}{\L^\infty} + \tfrac{1}{\sqrt{\mu}} \Nrm{\nabla\phi \cdot \opp\, \opgam}{\L^\infty}\) .
		\end{equation*}
	\end{lem}

	\begin{proof}
		Recall $H_\mu = H + \mu \geq - \Nrm{V_-}{L^\infty} + \mu > 0$. To compute the fractional commutator we use the formula
		\begin{equation*}
			H_\mu^{1/2} = \frac{1}{\pi} \int_0^\infty \frac{H_\mu}{H_\mu+t}\, \frac{\d t}{t^{1/2}}
		\end{equation*}
		which gives
		\begin{align*}
			\pi \com{H_\mu^{1/2}, \phi} &= \int_0^\infty \frac{1}{H_\mu+t} \com{\phi, H} \frac{H_\mu}{H_\mu+t}\, \frac{\d t}{t^{1/2}} + \int_0^\infty \frac{1}{H_\mu+t} \com{\phi, H} \frac{\d t}{t^{1/2}}
			\\
			&= \int_0^\infty \frac{1}{H_\mu+t} \com{\phi, H} \frac{H_\mu}{H_\mu+t}\, \frac{\d t}{t^{1/2}} + \pi \,H_\mu^{1/2} \com{\phi, H} . 
		\end{align*}
		Since $\com{\phi,H} = i\hbar\(\opp \cdot \nabla \phi + \nabla\phi \cdot \opp\)$ where $\nabla \phi$ is the operator of multiplication by $\nabla \phi$, and $\opgam$ commutes with $H_\mu$, it yields
		\begin{multline*} 
			I := \frac{\pi}{\hbar} \Nrm{\(\id - \opgam\) \com{H_\mu^{1/2}, \phi} \opgam}{\L^\infty} \leq \pi \Nrm{\(\id-\opgam\) H_\mu^{1/2} \(\opp \cdot \nabla \phi + \nabla\phi \cdot \opp\) \opgam}{\L^\infty}
			\\
			+ \int_0^\infty \(\Nrm{\frac{\indic_{H_\mu \geq \mu}}{H_\mu+t}\, \nabla \phi\cdot \opp\,\opgam \, \frac{H_\mu}{H_\mu+t}}{\L^\infty} + \Nrm{\frac{\indic_{H_\mu \geq \mu}}{H_\mu+t}\, \opp \cdot \nabla \phi\, \opgam \, \frac{H_\mu}{H_\mu+t}}{\L^\infty}\) \frac{\d t}{t^{1/2}} \, .
		\end{multline*}
		Observing that $\opgam = \indic_{0 \leq H_\mu\leq \mu}$, $\id - \opgam = \indic_{H_\mu \geq \mu}$ and $u \mapsto u/(u+t)$ is increasing on $\R_+$, it follows that
		\begin{multline*} 
			I \leq \pi \(\Nrm{H_\mu^{-1/2} \opp \cdot \nabla \phi\, \opgam}{\L^\infty} + \mu^{-1/2} \Nrm{\nabla\phi \cdot \opp\, \opgam}{\L^\infty}\)
			\\
			+ \int_0^\infty \(\frac{1}{\mu+t}\Nrm{\nabla \phi\cdot \opp\,\opgam}{\L^\infty} + \frac{1}{(\mu+t)^{1/2}} \Nrm{H_\mu^{-1/2}\, \opp \cdot \nabla \phi\, \opgam}{\L^\infty}\) \frac{\mu}{\mu+t}\,\frac{\d t}{t^{1/2}} \, .
		\end{multline*}
		Notice that since $\mu \geq \Nrm{V_-}{L^\infty}$, it follows that
		\begin{equation*}
			\n{\opp\,H_\mu^{-1/2}}^2 = H_\mu^{-1/2} \n{\opp}^2 H_\mu^{-1/2} = H_\mu^{-1/2} \(H_\mu-V-\mu\) H_\mu^{-1/2} \leq 1 \, .
		\end{equation*}
		Therefore, it gives
		\begin{multline*} 
			I \leq \pi \(\Nrm{\nabla \phi\, \opgam}{\L^\infty} + \mu^{-1/2} \Nrm{\nabla\phi \cdot \opp\, \opgam}{\L^\infty}\)
			\\
			+ \int_0^\infty \(\frac{1}{\mu+t}\Nrm{\nabla \phi\cdot \opp\,\opgam}{\L^\infty} + \frac{1}{(\mu+t)^{1/2}} \Nrm{\nabla \phi\, \opgam}{\L^\infty}\) \frac{\mu}{\mu+t}\,\frac{\d t}{t^{1/2}} \, .
		\end{multline*}
		The result follows by computing the integrals in $t$ and using $1+\frac{2}{\pi}, \frac{3}{2} \leq 2$.
	\end{proof}

	\begin{proof}[Proof of Proposition~\ref{prop:linear:commTR}]
		Taking $A = \opp $ or $A = x $ in the first estimate in Proposition~\ref{prop:comm2} gives directly the two first inequalities. Now for $\mu > \mu_0 $ we use the fractional commutator inequalities for $A = \opp$. In this case $\com{\com{\opp,H},H_\mu^{1/2}} =- i \hbar \com{\nabla V,H_\mu^{1/2}}$. It follows from Lemma~\ref{lem:frac} that
		\begin{equation*}
			\Nrm{\(\id - \opgam\) \com{\com{\opp ,H},H_\mu^{1/2}} \opgam}{\L^\infty} \leq 2\,\hbar^2 \(\Nrm{\nabla^2V \opgam}{\L^\infty} + \tfrac{1}{\sqrt{\mu}} \Nrm{\nabla^2V \cdot \opp\, \opgam}{\L^\infty}\).
		\end{equation*}
		It then suffices to apply Proposition~\ref{prop:comm_frac}, which finishes the proof.
	\end{proof}

	We are now ready to give a proof of our first main result, Theorem~\ref{thm:commutator:linear}. 
 
	\begin{proof}[Proof of Theorem~\ref{thm:commutator:linear}]
		In the proof, we make use of the estimates coming from propositions~\ref{prop:linear:comm_HS} and~\ref{prop:linear:commTR}. Observe that they depend only upon the constants $\CHS$, $\CTr$, $\CTr'$, $\cC_1$ and $\cC_2$, which are manifestly independent of $\hbar$. For the proof, let us set $\cC_*$ as $5$ times the maximum of these constants. Before we turn to the proof of the commutator estimates, let us give some preliminary estimates on the various moments that appear in the bounds. Namely, thanks to Inequality~\eqref{eq:weight_1}, Inequality~\eqref{eq:smallness_out_of_bulk_3}, Inequality~\eqref{eq:Agmon_rho_2_Linfty_global} and Lemma~\ref{lem:LT} we have for all $0 < \beta \leq 1/ 4 \hbar$
		\begin{align}\label{eq:m1}
			\Nrm{\opp\, \opgam}{\L^\infty} &\leq \sqrt 2 \Nrm{V_-}{L^\infty}^{1/2}
			\\ 
			\Nrm{e^{\beta\,d_1}\opgam}{\L^\infty} &\leq e^{\beta} + C_d \Nrm{V_-}{L^{p'}}^{p'/2}
			\\
			\Nrm{e^{\beta d_1} \opp \,\opgam}{\L^\infty} &\leq \sqrt 2\, e^\beta \Nrm{V_-}{L^\infty} + \sqrt e\, C_d \Nrm{V_-}{L^{p'}}^{p'/2}
		\end{align}
		where $p' = 1+\frac{d}{2}$, $C_d = 2/\sqrt{3} \(2d/(\pi e)\)^{d/2} \CLT^{1/2}$. Let us consider now $R > 0$ so that $\Omega_1 = \{V \leq 1\} \subset B_R$. In particular, $d_1 \geq \n{x}-R$ on $\Rd$. We then obtain 
		\begin{align}\label{eq:m2}
			\Nrm{x\, \opgam}{\L^\infty}
			&\leq e^R \Nrm{\n{x} e^{-\beta \n{x}}}{L^\infty} \(e^\beta + C_d \Nrm{V_-}{L^{p'}}^{p'/2}\)
			\\\label{eq:m3}
			\Nrm{\nabla V \opgam}{\L^\infty}
			&\leq e^R \Nrm{\nabla V e^{- \beta \n{x}}}{L^\infty} \(e^\beta + C_d \Nrm{V_-}{L^{p'}}^{p'/2}\)
			\\\label{eq:m4}
			\Nrm{\nabla^2V \opgam}{\L^\infty}
			&\leq e^R \Nrm{\nabla^2 V e^{-\beta \n{x}}}{L^\infty} \(e^\beta + C_d \Nrm{V_-}{L^{p'}}^{p'/2}\)
			\\\label{eq:m5}
			\Nrm{\nabla^2V \opp \, \opgam}{\L^\infty}
			&\leq e^R \Nrm{\nabla^2 V e^{-\beta \n{x}}}{L^\infty} \(\sqrt 2\, e^\beta \Nrm{V_-}{L^\infty} + \sqrt e\, C_d \Nrm{V_-}{L^{p'}}^{p'/2}\) . 
		\end{align} 

		Let us turn to the proof of the commutator estimates for $V$ satisfying $\eqref{hyp:V}$ with parameter $\beta \geq 0$. Note that if $ \beta\,\hbar > 1/4$, then we have
		\begin{equation*}
			\Nrm{\com{x, \opgam}}{\L^2}^2 \leq 4 \Nrm{x\,\opgam}{\L^\infty} \Nrm{\opgam}{\L^1} \leq 16 \,\beta\, \hbar \, \sfL_{0,d} \Nrm{x\, \opgam}{\L^\infty}
		\end{equation*}
		where we used $\opgam^2 = \opgam$ and the CLR inequality~\eqref{eq:CLR}. Note now that $\Nrm{x\, \opgam}{\L^\infty}$ is uniformly bounded in $\hbar \leq 1$ thanks to Inequality~\eqref{eq:m2} with parameter $0 < \beta_1 \leq 1/4$. An analogous argument is also true for the commutator with $\opp$.

		Thus, let us assume $0 \leq \beta\, \hbar \leq 1/ 4 $. First, from Proposition~\ref{prop:linear:comm_HS} we get
		\begin{align*}
			\Nrm{\com{x,\opgam}}{\L^2} &\leq \cC_*^{1/2} \hbar^{1/2} \(\Nrm{x\,\opgam}{\L^\infty} + \Nrm{\opp\, \opgam}{\L^\infty}\)
			\\ 
			\Nrm{\com{\opp,\opgam}}{\L^2} &\leq \cC_*^{1/2} \hbar^{1/2} \(\Nrm{\opp\,\opgam}{\L^\infty} + \Nrm{\nabla V \opgam}{\L^\infty}\) . 
		\end{align*}
		For $V$ satisfying~\eqref{hyp:V}, the right-hand side becomes bounded uniformly in $\hbar$ thanks to inequalities~\eqref{eq:m1}, \eqref{eq:m2} and~\eqref{eq:m3}. This proves the case $p=2$. In the case $p=\infty$, it suffices to write the commutator as a difference and use the triangle inequality for the operator norm to get
		\begin{equation*}
			\Nrm{\com{x,\opgam}}{\L^\infty} \leq 2 \Nrm{x\,\opgam}{\L^\infty} \quad \text{ and } \quad \Nrm{\com{\opp,\opgam}}{\L^\infty} \leq 2 \Nrm{\opp\,\opgam}{\L^\infty} .
		\end{equation*}
		Again, in each case, the right-hand side becomes bounded uniformly in $\hbar$ thanks to inequalities~\eqref{eq:m1} and~\eqref{eq:m2}. The case $p\in(2,\infty)$ now follows by interpolation of the case $p=2$ and $p=\infty$, that is, by H\"older's inequality for Schatten norms,
		\begin{equation*}
			\Nrm{\com{x,\opgam}}{\L^p} \leq \Nrm{\com{x,\opgam}}{\L^2}^{2/p} \Nrm{\com{x,\opgam}}{\L^\infty}^{1-2/p} \leq C\, \hbar^{1/p} . 
		\end{equation*}
		and similarly for $\com{\opp,\opgam}$.
		
		As for the trace-class estimates, we apply Proposition~\ref{prop:linear:commTR} to find
		\begin{align*}
			\hd \Tr{\n{\com{x,\opgam}}} &\leq \cC_*\, \hbar \(\Nrm{x\,\opgam}{\L^\infty} + \Nrm{\nabla V \opgam}{\L^\infty}\)
			\\
			\hd \Tr{\n{\com{\opp,\opgam}}} &\leq \cC_*\, \hbar \(\Nrm{\opp \,\opgam}{\L^\infty} + \ln (1 + \eps_0\, \hbar^{-1}) \Nrm{\nabla V \opgam}{\L^\infty}\). 
		\end{align*}
		Thanks to inequalities~\eqref{eq:m1}, \eqref{eq:m2} and~\eqref{eq:m3}, this proves the claim for $p=1$. The case $p\in(1,2)$ now follows by interpolation of the cases $p=1$ and $p=2$.
		
		Finally, assume $e^{-\beta \n{x}}\, \nabla^2 V \in L^\infty$ uniformly in $\hbar$. Then, Proposition~\ref{prop:linear:commTR} for $\mu = \Nrm{V_-}{L^\infty} + 1$ implies
		\begin{equation*}
			\hd \Tr{\n{\com{\opp,\opgam}}} \leq \cC_* \,\hbar \(\Nrm{\opp\,\opgam}{\L^\infty} + \sNrm{\nabla^2V \opgam}{\L^\infty} + \sNrm{\nabla^2V \cdot \opp\, \opgam}{\L^\infty}\) .
		\end{equation*}
		Thanks to inequalities~\eqref{eq:m1}, \eqref{eq:m4} and~\eqref{eq:m5}, this proves the claim for $p=1$ and the case $p\in(1,2)$ follows again by interpolation.
	\end{proof}

	\begin{proof}[Proof of Theorem~\ref{thm:comm:Besov}]
		For $p=2$ the claim follows from Proposition~\ref{prop:linear:comm_HS}. The constant $\mathcal D_\opgam$ is estimated similarly as in the proof of Theorem~\ref{thm:commutator:linear}. For $p=1$, the claim follows from the trace-class commutator estimates in Theorem~\ref{thm:commutator:linear} and the fact that for any $z \in \R^6$, by Formula~\eqref{eq:fund_thm},
		\begin{equation*}
			\Nrm{\sfT_z\opgam-\opgam}{\L^1} \leq \(\Nrm{\Dhx\opgam}{\L^1} + \Nrm{\Dhv\opgam}{\L^1}\) \n{z} .
		\end{equation*}
		For $p=\infty$, notice that on the one hand, since $0 \leq \opgam \leq 1$ and $0\leq \sfT_z\opgam\leq 1$, it follows that $\Nrm{\sfT_z\opgam-\opgam}{\L^\infty} \leq 1$. On the other hand, Formula~\eqref{eq:fund_thm} gives
		\begin{equation*}
			\Nrm{\sfT_z\opgam-\opgam}{\L^\infty} \leq \(\Nrm{\Dhx\opgam}{\L^\infty} + \Nrm{\Dhv\opgam}{\L^\infty}\) \n{z}
		\end{equation*}
		and $\Nrm{\Dhx\opgam}{\L^\infty}$ and $\Nrm{\Dhv\opgam}{\L^\infty}$ are bounded uniformly in $\hbar$ similarly as in the proof of Theorem~\ref{thm:commutator:linear}. For every other $p \in [1,\infty]$ we can use an interpolation argument, i.e. H\"older's inequality for Schatten norms.
	\end{proof}

\section{Quantitative local Weyl's law}

	The main goal of this section is to prove Theorem~\ref{thm:weyl:linear}. To this end, we shall obtain quantitative versions of the local Weyl law and the Weyl law in the phase space. The proof of the theorem then follows from propositions~\ref{prop:linear_local_Weyl_law} and~\ref{prop:linear_local_Weyl_law2} and will be given at the end of this section. In what follows, $\beta \geq 0$ will be a fixed non-negative parameter. 
	
	Let us introduce some notations in order to keep track of the constants. We recall $\Omega_a$ and $\Omega_{a,b}$
	where introduced in \eqref{eq:omega}. We denote by $\cM$ the space of measures on $\Rdd $ with finite total variation, with norm $\Nrm{F}{\cM}$. We let $g_\eps = \eps^{-d/2} \,e^{-\pi\n{x}^2/\eps}$ be a Gaussian at scale $\eps \in (0,1]$ on $\Rd$. Next, we introduce the following quantities for any classical function $f$, and operator $\op$. First, 
	\begin{equation}\label{eq:M} 
		\mathsf M_{f, \op} = \frac{d}{8\pi}\(M_\op + M_f\) , 
		\qquad
		M_f = \intdd f(x,\xi)\d x\d\xi\,, 
		\qquad
		M_\op = \hd \Tr{\op}.
	\end{equation}
	Secondly, for $q \in \{1,2\}$ 
	\begin{align}\label{eq:L1}
		&\sfL_{f, \op, 1} = \frac{2\sqrt{p'}}{p} \omega_d^{1/d} C_{d,\beta} \(\Nrm{e^{\beta \n{x}} \vr_f}{L^1}^{\frac{d-2}{d}} + \Nrm{e^{\beta \n{x}} \vr_\op}{L^1}^{\frac{d-2}{d}}\)
		\\\label{eq:L2}
		&\mathsf L _{f, \op, 2} = \frac{2\sqrt{p'}}{p} \omega_d^{1/d} \(\Nrm{\vr_f}{L^\infty}^{\frac{d-2}{d}}+
		\Nrm{\vr_\op}{L^\infty}^{\frac{d-2}{d}}\)
	\end{align}
	with $p = 1 +\frac{2}{d}$ and $C_{d , \beta} = \Nrm{e^{- \frac{(d-2) \beta}{2d} \n{x}}}{L^{d}(\Rd)}\Nrm{e^{\beta \n{x}} g_1}{L^1 (\Rd)}$. Finally, 
	\begin{align*}
		\mathsf D _{f ,\op} &= \tfrac{d}{\pi} \min\!\(
		\Nrm{\nabla V_-}{L^\infty} \Nrm{\nabla(\vr_f + \vr_\op)}{L^1},
		\Nrm{V}{W^{2,1} (\Omega_1)} \Nrm{\vr_f	+ \vr_\op}{L^{\infty}}\) 
		\\ %\label{C:f}
		&\qquad 
		+ \frac{1}{2}\intd\n{x}^2 e^{- \pi \n{x}^2 + \beta \n{x}} \d x \Nrm{e^{- \beta \n{x}} \nabla V_+}{L^\infty} \Nrm{e^{+ \beta \n{x}} \nabla \vr_f}{L^1} . 
	\end{align*}
	
	We now state our result regarding convergence of the densities.	
	
	\begin{prop}[Convergence of densities]\label{prop:linear_local_Weyl_law}
		Let $\hbar\in(0,1)$, $f = \indic_{\n{\xi}^2+V\leq 0}$ and $\opgam = \indic_{H\leq 0}$ where $H = -\hbar^2 \Delta + V$, in $d \geq 3$. Denote $\sM:= \sM_{f , \opgam}$, $ \sL_q := \sL_{f ,\opgam, q}$ and $\D:= \D_{f , \opgam}$. Then, there exists a constant $C_d>0$ such that 
		\begin{align*}%	\label{dens:21} %\label{eq:linear_Weyl_rho_L2}
			\Nrm{\vr_f-\vr_\opgam}{L^2(\Rdd)}
			&\leq C_d\, h^{1/3} (1 + \Nrm{\vr_\opgam}{\dot{B}^{1/2}_{2,\infty}})^{2/3}
			\(
			1 + \sL_2\, \sM^{1/2} + \sL_2\, \D^{1/2} h^{1/3}
			\)
			\\ 
			%\\\label{eq:linear_Weyl_rho_L1}%\label{dens:3}
			\Nrm{\vr_f-\vr_\opgam}{L^1(\Rdd)}
			&\leq C_d\, h^{1/2}
			\(
			1 +\Nrm{\nabla \vr_\opgam}{L^1}^{1/2} + \D^{1/2}
			\) \(1 + \sL_1 + \sM^{1/2}\).
		\end{align*}
	\end{prop}
	
	For the next result, we introduce the Husimi transform of the density operator $\opgam$ as follows 
	\begin{equation*}
		m_\opgam : = g_{h/2} \otimes g_{h/2} * f_\opgam \,.
	\end{equation*}
	For notational convenience, we introduce the following two constants, depending on $V$ and $\opgam$ only through quantities that are uniformly bounded in $\hbar$:
	\begin{align}\label{C1}
		\sC_1 &= d\, \om^\frac{2}{d} \n{\Omega_{1,1}}^\frac{2}{d}\, M_\opgam^{1-\frac{2}{d}} + \n{\Omega_{1,1}} 2^\frac{d}{2}
		\\\label{C2}
		\sC_2 &= \(2\(\frac{d+2}{2e\pi}\)^{1+\frac{d}{2}} + \frac{8}{3 \pi}\) M_\opgam \(1 +\Nrm{V_-}{L^\infty}\).
	\end{align}

	\begin{prop}[Convergence of states]\label{prop:linear_local_Weyl_law2}
		With the same notations as in Proposition~\ref{prop:linear_local_Weyl_law} 
		\begin{equation*}
			\Nrm{f - m_\opgam}{L^1 (\Rdd)} \leq h^{1/2} \(\(\mathsf L_1 + \sC_1\) \(\mathsf M + \D\) + \sC_2\) .
		\end{equation*}
		Additionally, there exists $C_d>0$ such that 
		\begin{align*} 
			\Nrm{f - f_\opgam}{L^2 (\Rdd)}
			&\leq C_d\, h^{1/4} \Nrm{f}{\dot B_{2,\infty}^{1/2}} + \Nrm{f - m_\opgam}{L^1 (\Rdd)}^{1/2}
			\\
			\Nrm{\op_f-\opgam}{\L^1}
			&\leq C_d\, h^{1/2} \(\Nrm{\Dh\opgam}{\L^1} + \Nrm{\nabla f}{\cM}\) + \Nrm{f - m_\opgam}{L^1 (\Rdd)}.
		\end{align*}
	\end{prop}
	
	We now follow the classic strategy of using coherent states together with the above variational principle to first get the following quantitative version of the Weyl law. To optimize the rate of convergence, we use anisotropic coherent states. Let $g_\eps = \eps^{-d/2} \,e^{-\pi\n{x}^2/\eps}$. Then we can consider the function of the phase space 
	\begin{equation*}
		G_\eps(x,\xi) := g_\eps(x)\,g_\frac{h^2}{4\eps}(\xi)\,.
	\end{equation*}
	Computing the kernel of its Weyl quantization gives
	\begin{equation*}
		\op_{G_\eps}(x,y) = h^{-d}\,g_\eps(\tfrac{x+y}{2})\, \widehat{g}_\frac{h^2}{4\eps}(\tfrac{x-y}{h}) = h^{-d}\eps^{-d/2} \,e^{-\pi \,\frac{\n{x}^2+\n{y}^2}{2\eps}}
	\end{equation*}
	that is, $\op_{G_\eps} = h^{-d} \ket{\psi_\eps}\bra{\psi_\eps}$ with $\psi_\eps = \eps^{-d/4}\, e^{-\pi\n{x}^2/(2\eps)}$. In particular, $\op_{G_\eps}\geq 0$.
	
	To keep our notations close to the classical case, we will write the coherent states transform (or Toeplitz operator, or Wick quantization) induced by these coherent states using in the form of a semiclassical convolution, as introduced already for instance in \cite{werner_quantum_1984, lafleche_quantum_2024}. This semiclassical convolution is defined for any operator $\op\in\L^p$ and function $f \in L^{p'}(\Rdd)$ by 
	\begin{equation*}
		f \star\op = \intdd g(z)\,\sfT_z\op\d z
	\end{equation*}
	where $\sfT_z$ is the operator of translation in the phase space $\sfT_z\op := \tau_{z} \, \op \, \tau_{-z}$ with $\tau_z = e^{-i\,z_0^\perp\cdot\opz/\hbar}$. It is a positive operator whenever $g\geq 0$ and $\op\geq 0$ and it verifies the analogue of Young's inequalities (see~\cite{werner_quantum_1984}): $\Nrm{g\star \op}{\L^p} \leq \Nrm{g}{L^q(\Rdd)} \Nrm{\op}{\L^r}$ if $1+\frac{1}{p} = \frac{1}{q}+\frac{1}{r}$. In particular, we denote by
	\begin{equation*}
		\tildop := G_\eps \star \op
	\end{equation*}
	and the Toeplitz operator of a function $f$ induced by these coherent states is nothing but $\tildop_f = G_\eps \star \op_f = f \star \op_{G_\eps}$. In particular, if $0\leq f\leq 1$, then 
	$0\leq \tildop_f\leq 1$ in the sense of operators. 
	We also denote by $\tilde{f}$ the function of the phase space $\tilde{f} := G_\eps * f$, so that $\tildop_f = \op_{\tilde{f}}$. Notice also that $\tilde{f}_\op$ is then the Husimi function associated to $\op$ induced by this anisotropic family of coherent states. 
	
	Let us first recall the following result sometimes referred to as the variational principle or the generalized min-max principle (see e.g. \cite{lieb_analysis_2001, frank_weyls_2023}) and additionally give a quantitative version of it.
 
	\begin{lem}[Quantitative Variational Principle]\label{prop:minmax}
		Let $H = -\hbar^2\Delta + V$ and $\opgam = \indic_{H \leq 0}$. Then, for any operator $0 \leq \op \leq 1$ there holds 
		\begin{equation}\label{eq:variational_principle}
			\Tr{H\(\op - \opgam\)} = \Tr{\n{H} \n{\op- \opgam}^2} + \Tr{\n{H} \op \(\id - \op\)}.
		\end{equation}
		Similarly, if $\cH(x,\xi) = \n{\xi}^2 + V(x)$ and $f = \indic_{\cH \leq 0}$. Then, for any function $0 \leq g \leq 1$
		\begin{equation}\label{eq:variational_principle_classical}
			\intdd \cH \(g - f\) = \intdd \n{\cH} \n{f-g}.
		\end{equation}
	\end{lem}
	
	\begin{remark}\label{rmk:var_principle_classical}
		The quantum version is reminiscent of the fact that in the classical version, we may write the right-hand side as follows
		\begin{equation*}
			\intdd \cH \(g - f\) = \intdd \n{\cH} \n{f-g}^2 + \intdd \n{\cH} g \(1-g\).
		\end{equation*}
	\end{remark}
	
	\begin{proof}
		Since $\opgam^2 = \opgam$, it follows that
		\begin{align*}
			\Tr{H\(\op-\opgam\)} &= \Tr{H\(\id-\opgam\)\(\op-\opgam\)} + \Tr{H\,\opgam\(\op-\opgam\)}
			\\
			&= \Tr{\n{H} \(\id-\opgam\) \op} + \Tr{\n{H} \opgam\(\id-\op\)}
		\end{align*}
		which can also be written as
		\begin{equation*}
			\Tr{H\(\op - \opgam\)} = \Tr{\n{H} \(\op^2 - \opgam\,\op - \op\,\opgam + \opgam^2\)} + \Tr{\n{H} \(\op - \op^2\)}
		\end{equation*}
		which gives Identity~\eqref{eq:variational_principle}. The classical case is treated analogously by splitting the integral into the regions where $f=0$ and $f=1$. 
	\end{proof}
	
	Let us introduce here the linear version of the functional $\scE$. That is, 
	\begin{equation*}
		\Eps_\op := \Tr{\(-\hbar^2 \Delta + V\)\op} , 
		\quad \text{ and } \quad
		\Eps_f := \intdd \(\n{\xi}^2 + V(x)\) f(x,\xi) \d x \d \xi\,.
	\end{equation*}
	Thanks to the properties of the Wigner transform, we obtain that for any $\op$ and $f$
	\begin{equation}\label{energy:relations}
		\Eps_{\op_f} = \Eps_f
		\quad
		\text{ and }
		\quad
		\Eps_{f_\op} = \Eps_\op \, .
	\end{equation}
	We use the following notation for a classical function $g : \Rd \times \Rd 
	\rightarrow \R$
	\begin{equation*}
		\cH_g(x, \xi) := \n{\xi}^2 - c_d\, \vr_g(x)^\frac{2}{d} \quad \text{ with } \quad c_d = \omega_d^{-\,\frac{2}{d}} \, .
	\end{equation*}
	We also introduce the following notations in order to keep track of the classical and quantum errors
	\begin{align}\label{eq:Q1}
		Q(g) &:= \intdd \n{\cH_g} \n{\indic_{\cH_g \leq 0} - g}
		\\\label{eq:Q2}
		\widetilde Q _\opgam(\op) &:= \hd \Tr{\n{H} \n{\op -\opgam}^2 + \n{H} (\op - \op^2)}
	\end{align}
	for any classical function $0 \leq g \leq 1$, and any operator $0 \leq \op \leq \id$.

	\begin{lem}\label{lem:energy_convexity}
		Let $\cH := \n{\xi}^2 + V(x)$ and $f := \indic_{\cH \leq 0}$. Then, for any function $0 \leq g \leq 1$
		\begin{equation}\label{eq:energy_convexity}
			\Eps_g - \Eps_f = Q(g) + \intd V_+\, \vr_g + \frac{c_d}{p} \intd \vr_g^p - \vr_f^p - p\,\vr_f^{p-1} \(\vr_g - \vr_f\) 
		\end{equation}
		with $p = 1 + \frac{2}{d}$ and $c_d = \omega_d^{-\,2/d}$. In particular, the content of the third integral verifies 
		\begin{equation}\label{eq:Dp_rho}
			\vr_g^p -\vr_f^p - p\,\vr_f^{p-1} \(\vr_g -\vr_f\) \geq \(p-1\) \n{\vr_g^{p/2} - \vr_f^{p/2}}^2.
		\end{equation}
	\end{lem}
	
	\begin{proof}
		Notice that, since $\vr_f = \om \,V_-^{d/2}$
		\begin{align*}
			\intdd \(\n{\xi}^2+V\) \(g - f\) &= \intdd \n{\xi}^2 \(g -f\) +
			\intd V \(\vr_g - \vr_f\)
			\\
			&= \intdd \n{\xi}^2 \(g -f\) + \intd V_+\, \vr_g - c_d \intd \vr_f^\frac{2}{d} \(\vr_g - \vr_f\).
		\end{align*}
		Therefore, using the fact that 
		$\intdd f\n{\xi}^2\d x\d\xi = \frac{c_d}{p} \intd \vr_f^p$, it gives
		\begin{multline*}
			\intdd \(\n{\xi}^2+V\) \(g - f\) \d z = \intd V_+\, \vr_g + \intdd \cH_g\, g \d x\d\xi
			\\
			+ \frac{c_d}{p} \intd p\, \vr_g^p -\vr_f^p - p\,\vr_f^{p-1} \(\vr_g - \vr_f\).
		\end{multline*}
		Now using the classical quantitative bathtub principle formula~\eqref{eq:variational_principle_classical} gives
		\begin{equation*}
			\intdd \cH_g \(g - \indic_{\cH_g \leq 0}\)\d x\d\xi = \intdd 
			\n{\cH_{g}} \n{\indic_{\cH_g \leq 0} - g}\d x\d\xi
		\end{equation*}
		and on the other side,
		\begin{equation*}
			\intdd \cH_g\, \indic_{\cH_g \leq 0} \d x\d\xi = -\intdd \(\n{\xi}^2-c_d\, \vr_g^{2/d}\)_- = -\frac{c_d}{p'}\intd \vr_g^p\d y
		\end{equation*}
		where $p'$ is the dual exponent of $p$. This proves Identity~\eqref{eq:energy_convexity}. To get Inequality~\eqref{eq:Dp_rho}, one can use Young's inequality for the product to get that for any $c>0$, $c \leq \(\frac{2}{p}-1\) c^p + \frac{2}{p'}\, c^{p/2}$ which can be written $\(p-1\) \n{c^{p/2}-1}^2 \leq c^p + 1 - p\(c-1\)$, and implies that for any $a,b\geq 0$
		\begin{equation*}
			\(p-1\) \n{b^{p/2}-a^{p/2}}^2 \leq b^p - a^p - p \,a^{p-1} \(b-a\),
		\end{equation*} 
		which implies the result by taking $a = \vr_f$ and $b = \vr_g$.
	\end{proof}

	\begin{cor}[Control via energies]\label{cor:energy}
		Let $H = -\hbar^2\Delta + V$, $\cH = \n{\xi}^2 + V(x)$, $\opgam = \indic_{H\leq 0}$, $f = \indic_{\cH \leq 0}$. 
		Then 
		\begin{equation*}
			Q(f_{\tilde{\opgam}}) + \widetilde Q_\opgam (\op_{\tilde f}) + \frac{c_d}{p'} \intd \n{\vr_{\tilde{\opgam}}^{p/2} - \vr_f^{p/2}}^2 + \intd V_+ \,\vr_{\tilde{\opgam}} \leq \Eps_{\tilde f} -	\Eps_f + \Eps_{\tilde{\opgam}} - \Eps_\opgam \, .
		\end{equation*}
		Here $p = 1 + \frac{2}{d}$, $c_d = \omega_d^{-2/d}$ and $Q$ and $\widetilde Q$ are defined in equations~\eqref{eq:Q1} and~\eqref{eq:Q2}.
	\end{cor}

	\begin{proof}
		We employ Equation~\eqref{eq:variational_principle} with $\op = \op_{\tilde f}$, Lemma~\ref{lem:energy_convexity} with $g = \tilde f_\opgam$, the relations~\eqref{energy:relations}, and the second inequality in Lemma~\ref{lem:energy_convexity}.
	\end{proof}

	The previous corollary establishes that the energy differences control various quantities of interest. In our next lemma, we establish estimates that will be helpful in the control of the energies. Recall $\sM$ was introduced in Proposition~\ref{prop:linear_local_Weyl_law}.

	\begin{lem}[Energy estimates]\label{lem:energies}
		Take the same notations as in Corollary~\ref{cor:energy}. Then if $V_-\in W^{1,\infty}(\Rd)$,
		\begin{align*} 
			\Eps_{\tilde f} - \Eps_{f} + \Eps_{\tilde \opgam} - \Eps_{\opgam} \leq 
			\frac{h^2}{\eps} \sM + \frac{d \eps}{4\pi} \Nrm{\nabla V_-}{L^\infty} \Nrm{\nabla (\vr_f + \vr_\opgam)}{L^1}
			+ \intd \(g_\eps * V_+\) \(\vr_f + \vr_\opgam\).
		\end{align*}	
%		\begin{align*}
%			\Eps_{\tilde f} - \Eps_f &\leq \frac{d\,h^2}{8\pi\,\eps}\, M_f + \frac{d\,\eps}{4\pi} \Nrm{\nabla V_-}{L^\infty} \Nrm{\nabla \vr_f}{L^1} + \intd \(g_\eps * V_+\) \vr_f
%			\\
%			\Eps_{\tilde \opgam} - \Eps_\opgam &\leq \frac{d\,h^2}{8\pi\,\eps}\, M_\opgam + \frac{d\,\eps}{4\pi} \Nrm{\nabla V_-}{L^\infty} \Nrm{\nabla \vr_\opgam}{L^1} + \intd \(g_\eps * V_+ \) \vr_\opgam \, , 
%		\end{align*}
			On the other hand, if $V \in W^{2,1}(\Omega)$ where $\Omega $ is an open set containing $\{V \leq 0\}$, then 
		\begin{align*}
			\Eps_{\tilde f} - \Eps_{f} + \Eps_{\tilde \opgam} - \Eps_{\opgam} \leq \frac{h^2}{\eps} \sM + \frac{d\,\eps}{\pi} \Nrm{V}{W^{2,1}(\Omega)} \Nrm{\vr_f + \vr_\opgam}{L^{\infty}} + \intd \(g_\eps * V_+\) \(\vr_f + \vr_\opgam\).
		\end{align*} 
		Additionally, we also have the estimate for the positive part 
		\begin{align*}
			\intd \(g_\eps * V_+\) \vr_f &\leq C_{d, \beta}\,
			\eps\Nrm{e^{- \beta \n{x}} \nabla V_+}{L^\infty} \Nrm{e^{+ \beta \n{x}} \nabla \vr_f}{L^1}
%			\\
%			\Nrm{\nabla \vr_f}{L^\infty} &\leq \frac{d \om}{2} \Nrm{V_-}{L^\infty}^{d/2-1} \Nrm{\nabla V_-}{L^\infty}.
		\end{align*}
		where $C_{d, \beta} = \frac{1}{2} \intd e^{\beta \n{x}} \n{x}^2 g_1 (x) \d x$.
	\end{lem}

	\begin{proof}
		In general,	since $\intdd \n{\xi}^2 G_\eps(z)\d z = \frac{d\,h^2}{8\pi\,\eps}$, $G_\eps$ is even and $\intdd f = M_f$, it follows that
		\begin{equation*}
			\intdd \n{\xi}^2 \tilde{f}(z)\d z = \intdd \n{\xi}^2 f(z)\d z + \frac{d\,h^2}{8\pi\,\eps}\, M_f
		\end{equation*}
		which yields
		\begin{equation*}
			\Eps_{\tilde f} = \Eps_f + \frac{d\,h^2}{8\pi\,\eps}\, M_f + \intd \(g_\eps * V - V\) \vr_f \, .
		\end{equation*}
		An analogous proof in the case of quantum densities yields
		\begin{equation*}
			\Eps_{\tilde{\opgam}} = \Eps_\opgam + \frac{d\,h^2}{8\pi\,\eps} M_\opgam + \intd \(g_\eps * V - V\) \vr_\opgam \, .
		\end{equation*}
		In what follows, we let $\rho = \vr_r + \vr_\opgam$. First, we decompose $V = V_+ - V_-$ and for the part containing $V_-$, we use a symmetrization argument to write 
		\begin{equation*}
			\intd \(g_\eps * V_- - V_-\) \rho = \frac{1}{2} \intdd \(V_-(x) - V_- (y)\) g_\eps(x-y) \(\rho(x) - \rho(y)\) \d x \d y \, .
		\end{equation*}
		First, assume $V_- \in W^{1, \infty}(\Rd)$. Then, a standard Taylor argument shows that 
		\begin{equation*}
			\intd \(g_\eps * V_- - V_-\) \rho \leq \frac{d\,\eps}{4\pi} \Nrm{\nabla V_-}{L^\infty} \Nrm{\nabla \rho}{L^1}.
		\end{equation*}
		
		Secondly, assume $V_- \in W^{2,1}(\Omega)$. We let $\chi \in C_c^\infty(\Rd , [0,1]) $ be so that $\chi = 1 $ on $\Omega$ and $\chi = 0$ on the complement of $\Omega_1 = \{x : \dist(x, \Omega) < 1 \}$. Here, $\chi$ can be chosen so that $\Nrm{\nabla \chi}{L^\infty} \leq 2 $ and $\Nrm{\nabla^2 \chi}{L^\infty} \leq 2$. Write $V = \chi\,V + \(1 - \chi\) V$ so that
		\begin{equation*}
			g* V - V = g* \(\chi\, V\) - \chi\, V + g*\(\(1 - \chi\)V\) - \(1 - \chi\)V \, . 
		\end{equation*}
		Observe that $1 - \chi$ is supported on $\{V \geq 0\}$. Thus for any $\rho \geq 0$
		\begin{equation*}
			\intd \(g* V - V\) \rho \leq \intd \(g* \(\chi\, V\) - \chi\, V\) \rho + \intd \(g* V_+\) \rho \, .
		\end{equation*}
		A second order Taylor expansion for $V_\chi := \chi\, V$ and the fact that $g_\eps$ is even imply 
		\begin{align*}
			V_\chi - 
			g_\eps*V_\chi &= \intd \(V_\chi (x)-V_\chi (x-y)\) g_\eps(y)\d y
			\\
			&= \intd g_\eps(y) \,y\otimes y : \int_0^1 \(1-\theta\) \nabla^2 V_\chi (x-\theta \,y)\d y
		\end{align*}
		where $y\otimes y : \nabla^2 V$ is the double contraction of tensors, that is $y\otimes y : \nabla^2 V = y\cdot \nabla^2 V\cdot y$. The desired Inequality then follows from H\"older's inequality and the product rule
		\begin{equation*}
			\Nrm{\nabla^2 V_\chi}{L^1(\Rd)} \leq\Nrm{\nabla^2 V}{L^1 (\Omega_1)} + 4 \Nrm{\nabla V}{L^1 (\Omega_1)} + 2\Nrm{V}{L^1 (\Omega_1)} \leq 4\Nrm{V}{W^{2,1}(\Omega_1)} 
		\end{equation*}
		which finishes the proof in this case.
		
		Finally, we look at the $V_+$ part for $\vr_f = \omega_d\, V_-^{d/2}$. Let us symmetrize the integrals, change variables $x \mapsto x+y$, and do a Taylor expansion to see 
		\begin{align*}
			\intd g_\eps * V_+ \vr_f &= \frac{1}{2} \intdd \(V_+(y + x) - V_+ (y)\) g_\eps(x) \(\vr_f (y + x) - \vr_f (y)\) \d x \d y
			\\
			&= 
			\frac{1}{2} \int_0^{1} \int_0^1 \intdd x \cdot \nabla V_+ \(y + t_1 x\) g_\eps (x) \, x \cdot \nabla \vr_f \(y + t_2 x\) \d x \d y \d t_1 \d t_2 \, . 
		\end{align*}
		Next, we change $y \mapsto y - t_2\, x $ and use $\n{t_1 - t_2} \leq 1$ to find 
		\begin{equation*}
			\intd g_\eps * V_+ \vr_f \leq \frac{1}{2} \intd \n{x}^2 e^{\beta \n{x}}
			g_\eps(x) \d x \Nrm{e^{-\beta \n{x}}\, \nabla V_+}{L^\infty} \Nrm{e^{\beta \n{y}} \vr_f}{L^1}\,. 
		\end{equation*}
		Thanks to $\eps \leq 1$ we have the bound $\frac{1}{2} \intd e^{\beta \n{x}} \n{x}^2 g_\eps(x) \leq C_{d, \beta} \, \eps$. This finishes the proof. 
	\end{proof}

	A combination of Corollary~\ref{cor:energy} and Lemma~\ref{lem:energies} yields the following result. 

	\begin{prop}\label{prop:energy2:linear}
		Take the same notations as in Corollary~\ref{cor:energy}. Then, for all $h ,\eps >0$ 
		\begin{equation}\label{eq:big_estimate}
			Q(\tilde{f}_\opgam) + \widetilde Q_\opgam(\op_{\tilde f}) + \frac{c_d}{p'} \intd \n{\vr_{\tilde{\opgam}}^{p/2} - \vr_f^{p/2}}^2 \leq
			\frac{h^2}{\eps} \, \sM + \eps \, \D \, .
		\end{equation}
		Here, $\sM$ and $\D$ are defined in Proposition~\ref{prop:linear_local_Weyl_law}, and $p = 1 + \frac{2}{d}$, $c_d = \omega_{d}^{-2 /d}$.
	\end{prop}
	
	We now use the previous result to get Lebesgue norms estimates. We recall $\Omega_{a,b} = \{x : \dist(x , \{V \leq a \}) < b \}$.

	\begin{lem}\label{lemma:weyl}
		With the same notations as in Lemma~\ref{lem:energies}, Proposition~\ref{prop:energy2:linear}, Proposition~\ref{prop:linear_local_Weyl_law}. Let $d \geq 3$. Then for all $0 < \eps \leq 1$ and $q\in [1,2]$
		\begin{equation}\label{eq:CV_Lq_smoothed_density}
			\Nrm{\vr_f-\vr_{\tilde{\opgam}}}{L^q(\Rd)} \leq \sL_q \(\frac{h^2}{\eps} \sM + \eps\, \D \)^{1/2} . 
		\end{equation} 
		Additionally,
		\begin{align}\label{eq:CV_L1_Husimi1}
			\Nrm{f - \tilde{f}_\opgam}{L^1 (\Rdd)} &\leq \(\sL_1 + \sC_1\) \Big(\frac{h^2}{\eps} \sM + \eps\, \D \Big)^{1/2} + \sC_2 \(h + \eps\)
			\\\nonumber
			\Nrm{\tilde f - f _\opgam}{L^2 (\Rdd)} &\leq \Nrm{f - \tilde{f}_\opgam}{L^1 (\Rdd)}^{1/2}
		\end{align}
		where $\sC_1$ and $\sC_2$ are defined by equations~\eqref{C1} and~\eqref{C2}.
	\end{lem}

	\begin{proof}
		Notice that all the terms on the left-hand side of inequality~\eqref{eq:big_estimate} are positive, hence they are controlled individually by the sum of the terms on the right-hand side.

		First we prove Inequality~\eqref{eq:CV_Lq_smoothed_density}. Recall for any $a,b\geq 0$ and $\theta\leq 1$, we have the following inequality $\n{a-b} \leq \frac{1}{\theta} \n{a^\theta-b^\theta} \(a^{1-\theta}+b^{1-\theta}\)$. Set $\theta = 1 - p$ with $ p = 1 + \frac{2}{d}$ so that $1 -\theta = \frac{d-2}{2d} > 0$.
		% 1/q = 1/r + 1/2  <=>  r = 2q/(2-q)
		Then, we find thanks to H\"older's inequality 
		\begin{align*}
			\Nrm{\vr- \vr_{\tilde{\opgam}}}{L^1}
			&\leq \frac{2}{p} \Nrm{\vr_f^{p/2}-\vr_{\tilde{\opgam}}^{p/2}}{L^2} \(\Nrm{\vr_f^{\frac{d-2}{2d}}}{L^2} + \Nrm{\vr_{\tilde{\opgam}}^{\frac{d-2}{2d}}}{L^2}\)
			\\
			\Nrm{\vr- \vr_{\tilde{\opgam}}}{L^2} &\leq \frac{2}{p} \Nrm{\vr_f^{p/2}-\vr_{\tilde{\opgam}}^{p/2}}{L^2} \(\Nrm{\vr_f}{L^\infty}^{\frac{d-2}{2d}} + \Nrm{\vr_{\tilde{\opgam}}}{L^\infty}^{\frac{d-2}{2d}}\) .
		\end{align*}
		Observe that Proposition~\ref{prop:energy2:linear} readily gives the desired inequality for the $L^2$ estimate. As for the $L^1$ estimate, we note that for any $\vr \geq 0$ and $\delta > 0$, if $r = \frac{2d}{d-2}$, so that $r ' = d$ and $\frac{1}{r} + \frac{1}{r'} = \frac{1}{2}$, then by H\"older's inequality
		\begin{equation*}
			\Nrm{\vr^\frac{d-2}{2d}}{L^2} \leq \Nrm{e^{- \delta \n{x}}}{L^{r'}} \Nrm{e^{\delta \n{x}} \vr^\frac{d-2}{2d}}{L^r} = \Nrm{e^{- \delta \n{x}}}{L^{d}} \Nrm{e^{\frac{2 d \delta}{d-2}\n{x}}\vr}{L^1}^{\frac{d-2}{2d}} .
		\end{equation*}
		We now choose $\delta = \frac{d-2}{2d}\beta $, and use $\Nrm{e^{\beta \n{x}} \vr_{\tilde \opgam}}{L^p} \leq \Nrm{e^{\beta \n{x}} g_1}{L^1} \Nrm{e^{\beta\n{x}} \vr_{\opgam}}{L^p}$ for $\eps \leq 1 $ to finish the proof of the first estimates.
		
		Secondly, we prove~\eqref{eq:CV_L1_Husimi1}. To this end, we set $v := c_d \,\vr_{\tilde{\opgam}}^{2/d}$ and $\cH_{\tilde f_\opgam} = \n{\xi}^2-v$ and use the triangle inequality
		\begin{equation}\label{eq1}
			\Nrm{f - \widetilde f_\opgam}{L^1} \leq \Nrm{\indic_{\cH_f \leq 0} - \indic_{\cH_{\tilde{f}_\opgam} \leq 0}}{L^1} + \Nrm{\indic_{\cH_{\tilde{f}_\opgam} \leq 0} - \widetilde{f}_\opgam}{L^1}
		\end{equation}
		where we used $f = \indic_{\cH_f} $ with $ \cH_f = \n{\xi}^2 - c_d\,\vr_f^{2/d}$. The first term on the right-hand side of Inequality~\eqref{eq1} can be readily estimated with \eqref{eq:CV_Lq_smoothed_density} with $q=1$ via 
		\begin{equation*}
			\intdd \n{\indic_{\cH_f\leq 0}-\indic_{\cH_{\tilde{f}_\opgam}\leq 0}} = \intd \n{\vr_f^{1/d} - \vr_{\tilde{\opgam}}^{1/d}}^d \leq \intd \n{\vr_f - \vr_{\tilde{\opgam}}}.
		\end{equation*}
		For the second term on the right-hand side of Inequality~\eqref{eq1}, we have 
		\begin{multline*}
			\frac{1}{2}\int_{\Omega_{1,1}^c\times\Rd} \n{\indic_{\cH_{\tilde{f}_\opgam} \leq 0} - \tilde{f}_\opgam} \leq \int_{\Omega_{1,1}^c} \vr_{\tilde{\opgam}}
			\\
			\leq \int_{\Omega_{1,1}^c\times B_{1/2}} \vr_\opgam(x-y)\,g_\eps(y) \d x\d y + \int_{\Omega_{1,1}^c\times B_{1/2}^c} \vr_\opgam(x-y)\,g_\eps(y) \d x\d y
			\\
			\leq \int_{\Omega_{1,1/2}^c\times B_{1/2}} \vr_\opgam(x)\,g_\eps(y) \d x\d y + \frac{M_\opgam}{\eps^{d/2}}\int_{B_{1/2}^c}\,e^{-\frac{\pi\n{y}^2}{\eps}} \d y
			\\
			\leq \int_{\Omega_{1,1/2}^c} \vr_\opgam + C_d ' \,M_\opgam \, \eps
		\end{multline*}
		where $C_d' = 2\(\frac{d+2}{2e\pi}\)^{1+d/2}$. By Agmon's estimates, and more precisely Inequality~\eqref{eq:smallness_out_of_bulk_1}
		\begin{equation*}
			\int_{\Omega_{1,1/2}^c} \vr_\opgam \leq \frac{4}{3}\,e^{-\frac{1}{2\hbar}} \intd \vr_\opgam\,V_-
			\leq \frac{4 h}{3 \pi} M_\opgam\Nrm{V_-}{L^\infty}.
		\end{equation*}
		On the other hand, for any $R>0$, splitting the integral over the region where $\lvert \cH_{\tilde{f}_\opgam}\rvert< R$ and $\lvert \cH_{\tilde{f}_\opgam}\rvert \geq R$ yields
		\begin{equation*}
			\int_{\Omega_{1,1}\times\Rd} \n{\indic_{\cH_{\tilde{f}_\opgam} \leq 0} - \tilde{f}_\opgam} \leq \int_{\Omega_{1,1}\times\Rd} \indic_{\n{\n{\xi}^2 - v}< R} \d x\d\xi + \frac{1}{R}\, Q (\tilde{f}_\opgam) \,.
		\end{equation*}%
		Computing the first integral of this last expression gives
		\begin{multline*}
			\int_{\Omega_{1,1}\times\Rd} \indic_{\n{\n{\xi}^2 - v}< R} \d x\d\xi = \om \int_{\Omega_{1,1}} (v+R)_+^{d/2}-(v-R)_+^{d/2}
			\\
			\leq d \,\om \,R \int_{\Omega_{1,1}} v^{\frac{d}{2}-1} + \n{\Omega_{1,1}} \(2R\)^{d/2}.
		\end{multline*}%
		Using H\"older's inequality to control the integral of $v^{d/2-1}$ by $M_\opgam$, we finally get
		\begin{equation}\label{eq:CV_L1_Husimi_0}
			\intdd \n{\indic_{\cH_{\tilde{f}_\opgam} \leq 0} - \tilde{f}_\opgam}
			\leq M_\opgam \(C'_d \eps + \frac{8 h}{3 \pi}\Nrm{V_-}{L^\infty} h\) + \sC_{1}\, R + \frac{1}{R}\, Q (\tilde{f}_\opgam)
		\end{equation}%
		with $\sC_{1} = d\, \om^\frac{2}{d} \n{\Omega_{1,1}}^\frac{2}{d}\, M_\opgam^{1-\frac{2}{d}} + \n{\Omega_{1,1}} 2^\frac{d}{2}$. Finally choose $R= (\frac{h^2}{\eps} \sM + \eps\, \D)^{1/2} $. The desired estimate then follows from the bound of $Q(\tilde f_{\opgam})$ in Proposition~\eqref{prop:energy2:linear}, and the definition of $\sC_2$.
		
		For the last inequality we observe that 
		\begin{align*}\nonumber
			\intd \n{f-\tilde{f}_\opgam} 
			&= \intdd f\(1-\tilde{f}_\opgam\) + \(1-f\) \tilde{f}_\opgam = \intdd f + \tilde{f}_\opgam - 2 \,f \,\tilde{f}_\opgam
			\\
			&
			= \intdd \tilde{f} + f_\opgam - 2 \,\tilde{f} \,f_\opgam 
			= \hd \Tr{\tildop_f + \opgam - 2\,\tildop_f\,\opgam}
			\\
			&= \hd \Tr{\n{\tildop_f - \opgam}^2 + \tildop_f \(\id-\tildop_f\)}
			\geq \Nrm{\tilde f - f_\opgam}{L^2}^2
		\end{align*}
		where in the third equality we used $\scalar{f}{G* f_\opgam} = \scalar{G* f}{f_\opgam}$ and $\intdd f = \intdd G* f$. 
	\end{proof}
	
	It remains to remove the convolutions by $G_\eps$ using the regularity of $\opgam$ to obtain the proof of the main result of the section.
	
	\begin{proof}[Proof of Proposition~\ref{prop:linear_local_Weyl_law}]
		First, since $\vr_\opgam \in \dot{B}^{1/2}_{2,\infty}$ we get from Inequality~\eqref{eq:CV_Lq_smoothed_density} for $q = 2$, and the elementary bound $\sqrt{a + b} \leq \sqrt a + \sqrt b$, 
		that for some constant $C_d>0$
		\begin{align*}
			\Nrm{\vr_\opgam - \vr_f}{L^2} &\leq \Nrm{\vr_{\tilde \opgam} - \vr_\opgam}{L^2} + \Nrm{\vr_f - \vr_{\tilde \opgam}}{L^2} 
			\\
			&\leq C_d \(\eps^{1/4} \Nrm{\vr_\opgam}{\dot{B}^{1/2}_{2,\infty}} + \sL_2 \sM^{1/2} \frac{h}{\eps^{1/2}} + \sL_2 \D^{1/2} \eps^{1/2}\) .
		\end{align*}
		We minimize the first two terms with respect to $\eps$. That is, we put
		\begin{equation*}
			\eps = \frac{h^{4/3}}{1 + \Nrm{\vr_\opgam}{\dot{B}^{1/2}_{2,\infty}}^{4/3}} \leq 1 \, .
		\end{equation*}
		This gives the bound, for an updated constant $C_d>0$,
		\begin{equation*}
			\Nrm{\vr_\opgam - \vr_f}{L^2} \leq C_d \, h^{1/3} \,\Big(1 + \Nrm{\vr_\opgam}{\dot{B}^{1/2}_{2,\infty}}^{2/3}\Big) \(1 + \sL_2 \sM^{1/2} + \sL_2 \D^{1/2} h^{1/3}\). 
		\end{equation*}

		Secondly, we take $q = 1$ in Inequality~\eqref{eq:CV_Lq_smoothed_density} which gives, for a constant $C_d > 0$
		\begin{align*}
			\Nrm{\vr_\opgam - \vr_f}{L^1} 
			&\leq \Nrm{\vr_{\tilde \opgam} - \vr_\opgam}{L^1} + \Nrm{\vr_f - \vr_{\tilde \opgam}}{L^1} 
			\\
			&
			\leq C_d \(\(\Nrm{\nabla \vr_\opgam}{L^1} + \sL_1 \D^{1/2}\) \eps^{1/2} + \sL_1 \sM^{1/2} \frac{h}{\eps^{1/2}}\). 
		\end{align*}
		We choose $\eps$ as
		\begin{equation*}
			\eps = \frac{h}{1 +\Nrm{\nabla \vr_\opgam}{L^1}} \leq 1 
		\end{equation*}
		which gives, for a new constant $C_d>0$
		\begin{align*}
			\Nrm{\vr_\opgam - \vr_f}{L^1} 
			&\leq C_d \, h^{1/2} \(1 +\Nrm{\nabla \vr_\opgam}{L^1}^{1/2} + \D^{1/2}\) \(1 + \sL_1 + \sM^{1/2}\)
		\end{align*}
		which finishes the proof. 
	\end{proof}

	\begin{proof}[Proof of Proposition~\ref{prop:linear_local_Weyl_law2}]
		The first inequality follows directly from \eqref{eq:CV_L1_Husimi1} by putting $\eps = h$. As a consequence, we obtain also thanks to the triangle inequality and $\tilde{f} = G_\eps * f$ with $f\in \dot{B}_{2,\infty}^{1/2}$ 
		\begin{equation*}
			\Nrm{f-f_\opgam}{L^2} \leq \Nrm{f-\tilde{f}}{L^2} + \Nrm{\tilde{f}-f_\opgam}{L^2} \leq C_d \, h^{1/4} \Nrm{f}{\dot{B}_{2,\infty}^{1/2}} + \Nrm{f - \tilde f_\opgam}{L^1}^{1/2} . 
		\end{equation*}
		Let us now prove the third inequality. It follows from~\cite[Equation~(50)]{lafleche_quantum_2024} that with $\eps = h$ 
		\begin{equation*}
			\Nrm{\tilde{\tilde{\opgam}}-\opgam}{\L^1} \leq C_d \, \sqrt{h} \Nrm{\Dh\opgam}{\L^1} \quad \text{ and } \quad \Nrm{\tildop_g-\op_g}{\L^1} \leq C_d \, \sqrt{h} \Nrm{\Dh\op_g}{\L^1}
		\end{equation*}
		with $C_d = \frac{2d\,\omega_{2d}}{\(2d+1\) \omega_{2d+1}}$ and $\n{\Dh\opgam}^2 = \n{\Dhx\opgam}^2 + \n{\Dhv\opgam}^2$. Therefore, taking $g = \tilde{f}$, it follows by the triangle inequality for the trace norm that
		\begin{equation*}
			\Nrm{\op_f-\opgam}{\L^1} \leq C_d \(\Nrm{\Dh\opgam}{\L^1} + \Nrm{\Dh\op_{\tilde{f}}}{\L^1}\) h^{1/2} + \Nrm{\tildop_f-\tilde{\tilde{\opgam}}}{\L^1}
		\end{equation*}
		which by the properties of the semiclassical convolution implies that
		\begin{equation*}
			\Nrm{\op_f-\opgam}{\L^1} \leq C_d \(\Nrm{\Dh\opgam}{\L^1} + \Nrm{\Dh\op_{\tilde{f}}}{\L^1}\) h^{1/2} + \Nrm{f-\tilde{f}_\opgam}{L^1(\Rdd)}\,. 
		\end{equation*}
		The desired estimate follows by $\Nrm{\Dh\op_{\tilde{f}}}{\L^1} = \Nrm{\tildop_{\nabla f}}{\L^1} \leq \Nrm{\nabla f}{\cM(\Rdd)}$.
	\end{proof}

	Finally, we turn to the proof of Theorem~\ref{thm:weyl:linear}. Let us first argue that the limiting state is in both $BV(\Rdd)$ and $B_{2, \infty}^{1/2}(\Rdd)$. To see this, consider $F (x,\xi) = \indic_{\n{\xi}\leq \rho(x)^{1/d}}$ for some $\rho \geq 0$ with $\rho \in W^{1,1}(\Rd) \cap L^{\frac{d-1}{d}} (\Rd)$. Then, for any $\varphi \in C_c^1(\Rdd)$, we have
	\begin{align*}
		\intdd F(x,\xi) \, \Dx \varphi(x,\xi) \d x \d \xi &= \frac{-1}{d} \intd \nabla \rho(x) \( \int_{\SS^{d-1}} \varphi(x, \rho(x)^{1/d}\sigma) \d \sigma\) \d x
		\\
		\intdd F(x,\xi) \Dv \varphi(x,\xi) \d x \d \xi &= - \intd \rho(x)^{\frac{d-1}{d}} \( \int_{\SS^{d-1}} \varphi(x,\rho(x)^{1/d}\sigma) \d \sigma\) \d x \ . 
 	\end{align*}
	Thus, $ \nabla F \in \mathcal M$ and so $F$ is $BV$. In addition, since $F$ is a characteristic function, it follows that $F \in B_{p, \infty}^{1/p}$ for all $1 \leq p < \infty$. See e.g.~\cite[Theorem 2]{sickel_regularity_2021}.

	\begin{proof}[Proof of Theorem~\ref{thm:weyl:linear}]
		Let $V$ satisfy \eqref{hyp:V} with parameter $\beta \geq 0$ in $d=3$. In view of Proposition~\ref{prop:linear_local_Weyl_law} and~\ref{prop:linear_local_Weyl_law2}, it suffices only to estimate the quantities $\sM$, $\sL_q$ and $\D$. 
		To this end, recall $\vr_f = \omega_d \, V_-^{3/2} \in W^{1,1}(\R^3)$ is  continuous and compactly supported.
		Thus,   $f \in B_{2, \infty}^{1/2} $ and $ \nabla f \in \mathcal M$. Without loss of generality, we assume $\beta \leq 1 /8 \hbar$. Otherwise, we may proceed as in the proof of Theorem~\ref{thm:commutator:linear}.
		
		First, thanks to the CLR bound~\eqref{eq:CLR} and the compact support of $\vr_f$, it is clear that $\sM$ is bounded uniformly in $\hbar$. Additionally, thanks to Agmon's estimates~\eqref{exp:decay}, we have $\Nrm{e^{\beta \n{x}} \vr_\opgam}{L^1}\leq C$ uniformly in $\hbar$. From Lemma~\ref{lem:Linfty} we also get that \begin{equation*}
			\Nrm{\vr_\opgam}{L^\infty} \leq C \(1 +\Nrm{V_-}{L^\infty} +\Nrm{V e^{-\beta \n{x}}}{L^\infty} \Nrm{e^{\beta \n{x}} \opgam}{\L^\infty}\).
		\end{equation*}
		Thanks to Proposition~\ref{prop:Agmon}, we can control $\Nrm{e^{\beta \n{x}} \opgam}{\L^\infty}$. Thus, $\vr_\opgam \in L^\infty$ uniformly in $\hbar$ and we conclude that $\sL_q$ is controlled for $q=1$ and $q=2$. Finally, for $\D$ we note that if $V$ satisfies~\eqref{hyp:V}, then, thanks to $\vr_f \in W^{1,1}(\R^3)$ and $\n{\nabla V_\pm} \leq \n{\nabla V}$, we get
		\begin{align*}
			\D &\leq C \(1 +\Nrm{\nabla V_-}{L^\infty} + \Nrm{e^{- \beta \n{x}} \nabla V_+}{L^1}\) \(1 +\Nrm{\nabla \vr_f}{L^1}\) \(1 +\Nrm{\nabla \vr_\opgam}{L^1}\)
			\\
			&\leq C \(1 +\Nrm{\nabla \vr_\opgam}{L^1}\)
		\end{align*}
		for some constant $C >0$ independent of $\hbar$. Finally, we recall that thanks to Theorem~\ref{thm:commutator:linear}, we get $\Nrm{\nabla \vr_\opgam}{L^1} \leq C \n{\ln \hbar}$. Thus, $\D^{1/2} h^{1/3}\leq C$ for all $\hbar \in (0,1)$. This finishes the proof.
	\end{proof}

\section{Regularity of the position density}
	
	In this section we establish further regularity properties of the position density $\vr_\opgam$. These estimates will be particularly helpful when we study interacting particle systems with Coulomb potentials. First, we note that it follows from the decay of $\opgam$ in the $\n{\opp}$ direction that $\hbar\,\nabla\vr_\opgam$ is in $L^\infty$ uniformly in $\hbar$, as the following lemma shows.
 
	\begin{lem}\label{lem:rho_inf}
		In dimension $d\leq 3$,
		\begin{align*}
			\Nrm{\nabla\vr_\opgam}{L^\infty} &\leq \frac{C_d}{\hbar} \(1+\Nrm{\n{\opp}^4\opgam}{\L^\infty}^{5/4}\)
			\\
	%		\Nrm{\nabla^2\vr_\opgam}{L^\infty} &\leq \frac{C_d}{\hbar^2} \(1+\Nrm{\n{\opp}^4\opgam}{\L^\infty}^{3/2}\)
	%		\\
			\Nrm{\nabla^3\vr_\opgam}{L^\infty} &\leq \frac{4\,C_d}{\hbar^3} \(1+\Nrm{\n{\opp}^6\opgam}{\L^\infty}^{7/6}\) ,
		\end{align*}
		with $C_d = 4 \intd \weight{x}^{-4}\d x = \frac{8\,\pi^2}{\(4-d\) \omega_{4-d}}$. Additionally, in dimension $d \geq 3$
		\begin{align*}
			\Nrm{\nabla\vr_\opgam}{L^1} &\leq
			\frac{C_d}{\hbar} \Nrm{\opp \opgam}{\L^\infty}
			\\
	%		\Nrm{\Delta\vr_\opgam}{L^\infty} &\leq \frac{C_d}{\hbar^2} \(1+\Nrm{\n{\opp}^4\opgam}{\L^\infty}^{3/2}\)
	%		\\
			\Nrm{\nabla^3\vr_\opgam}{L^1} &\leq \frac{C_d}{\hbar^3} \(\Nrm{\opp^3 \opgam}{\L^\infty} + \Nrm{\opp \opgam}{\L^\infty} \Nrm{\opp^2 \opgam}{\L^\infty}\)
		\end{align*}
		with $C_d = 16 \, \sL_{0,d} $. If $V$ has four locally bounded derivatives, the last factor of the right-hand side of both inequalities can be controlled thanks to Lemma~\ref{lem:weight_4} and Lemma~\ref{lem:weight_6}.
	\end{lem}
 
	\begin{proof}
		We proceed by duality. Let $\varphi\in C^\infty_c(\Rd)$ be such that $\Nrm{\varphi}{L^1} = 1$. Then
		\begin{equation*}
			\n{\intd \nabla \vr_\opgam\, \varphi} = \n{\frac{i}{\hbar}\,\hd \Tr{\opp \,\opgam\,\varphi - \opgam\,\opp\,\varphi}} \leq \frac{2}{\hbar} \,\hd \n{\Tr{\opp \,\opgam\,\varphi}} .
		\end{equation*}
		Moreover, by the cyclicity of the trace, H\"older's inequality and the fact that $\frac{\varphi}{\n{\varphi}} \,\indic_{\varphi\neq 0}$ is a bounded multiplication operator with norm $1$, it follows that
		\begin{align*}
			\hd\n{\Tr{\opp \,\opgam\,\varphi}} &= \n{\hd\Tr{\weight{\opp}^2\opp \,\opgam \weight{\opp}^2 \weight{\opp}^{-2}\,\varphi \weight{\opp}^{-2}}}
			\\
			&\leq \Nrm{\weight{\opp}^2\opp \,\opgam}{\L^\infty} \Nrm{\opgam \weight{\opp}^2}{\L^\infty} \Nrm{\weight{\opp}^{-2} \sqrt{\n{\varphi}}}{\L^2}^2 .
		\end{align*}
		Moreover, it follows from H\"older's inequality for Schatten norms together with the Araki--Lieb--Thirring inequality (or more precisely the Heinz inequality~\cite{heinz_beitrage_1951}) that for any $k\in[0,4]$,
		\begin{equation*}
			\Nrm{\n{\opp}^k\opgam}{\L^\infty} \leq \Nrm{\n{\opp}^4\opgam}{\L^\infty}^{k/4}.
		\end{equation*}
		Putting all these inequalities together yields 
		\begin{equation*}
			\n{\intd \nabla \vr_\opgam\, \varphi} \leq \frac{2}{\hbar}\(\Nrm{\n{\opp}^4\opgam}{\L^\infty}^{1/4}+\Nrm{\n{\opp}^4\opgam}{\L^\infty}^{3/4}\)\(1+\Nrm{\n{\opp}^4\opgam}{\L^\infty}^{1/2}\) \Nrm{\weight{x}^{-2}}{L^2}^2
		\end{equation*}
		where we used the fact that $\Nrm{\varphi}{L^1} = 1$. This gives the first inequality. The second inequality follows analogously.
		
		For the second part, we proceed analogously with $\varphi \in C_c^\infty$ so that $\Nrm{\varphi}{L^\infty}=1 $. We get thanks to the CLR bound~\eqref{eq:CLR}
		\begin{equation*}
			\n{\intd \nabla \vr_\opgam \varphi} \leq \frac{2}{\hbar} \Nrm{\opp \, \opgam}{\L^1} \leq \frac{2\, \sL_{0,d}}{\hbar} \Nrm{\opp\, \opgam}{\L^\infty} . 
		\end{equation*}
		For the three derivatives, we expand the commutator (here understood as tensors)
		\begin{equation*}
			\com{\opp, \com{\opp, \com{\opp, \opgam}}} = \opp^3\, \opgam - 3\, \opp^2\, \opgam\, \opp + 
			3\, \opp\, \opgam\, \opp^2 - \opgam\, \opp^3
		\end{equation*}
		and proceed analogously. 
	\end{proof}

	In particular it follows from the above lemma and Inequality~\eqref{eq:rho_W11} that for any $p\in[1,\infty]$,
	\begin{equation*}
		\hbar^{1/p'} \Nrm{\nabla\vr_\opgam}{L^p} \leq C_d^{1/p'} \(1+\sNrm{\n{\opp}^4\opgam}{\L^\infty}^{5/4}\)^{1/p'} \Nrm{\Dhx\opgam}{\L^1}^{1/p}
	\end{equation*}
	which is bounded uniformly in $\hbar$ whenever $\sNrm{\n{\opp}^4\opgam}{\L^\infty}$ and $\Nrm{\Dhx\opgam}{\L^1}$ are. However, in the case of nonlinear interactions with a Coulomb potential, we can only estimate $\Nrm{\Dhx\opgam}{\L^1}$ by a term diverging like $\n{\ln \hbar}^{1/2}$. We are still able to prove however that $\hbar^{1/p'} \Nrm{\nabla\vr_\opgam}{L^p}$ is bounded uniformly in $\hbar$ for $p\in[2,\infty]$ by considering weighted Hilbert--Schmidt commutator estimates. This is the content of the next proposition.

	\begin{prop}[$L^2$ regularity of $\vr_\opgam$]\label{prop:regu_rho_L2}
		Let $\hbar \in (0,1)$ and $\opgam = \indic_{H \leq 0}$ with $H=-\hbar^2\Delta+V$ in $d = 3$. Assume that $V$ verifies~\eqref{hyp:V} and, additionally, that there exists $C>0$ and $\beta \geq 0$ independent of $\hbar$ such that for any $0\leq n \leq 5$
		\begin{equation*}\tag{H4}\label{hyp:V4}
			e^{- \beta \n{x}} \n{\nabla^n V(x)} \leq C \, \hbar^{-(n - 1)_+} \, .
		\end{equation*}
		Then there exists a constant $C>0$ such that for any $x\in\R^3$,
		\begin{equation*}
			\Nrm{T_x\vr_\opgam-\vr_\opgam}{L^2} \leq C \min\!\(\tfrac{\n{x}}{\hbar}, \sqrt{\n{x}}, 1\) . 
		\end{equation*}
		In particular, for any $s<1/2$,
		\begin{equation*}
			\sqrt{\hbar}\Nrm{\nabla\vr_\opgam}{L^2}, \qquad \Nrm{\vr_\opgam}{B^{1/2}_{2,\infty}} \quad \text{ and } \quad \Nrm{\vr_\opgam}{H^s}
		\end{equation*}
		are bounded uniformly in $\hbar$.
	\end{prop}
	
	One of the main ingredients of the proof is the following weighted variant of Proposition~\ref{prop:comm_HS} for Hilbert--Schmidt estimates. 
 
	\begin{prop}[Weighted Hilbert--Schmidt bounds]\label{prop:comm_HS_weight}
		Let $\opgam = \indic_{H\leq 0}$ and $H = -\hbar^2 \Delta + V$ with $d \geq 1$. In addition, let $A$ be a normal operator on $L^2(\Rd)$, and let $m$ be an operator in $L^2(\Rd)$ with polynomial growth in $x$ or $\opp$. Then for any $\lambda >0$,
		\begin{multline*}
			\Tr{\n{\com{A,\opgam}}^2 m} \leq \sum_{\ii=1}^3 \Nrm{\indic_{H \geq \lambda} \com{\opb_\ii ,H} \(\lambda - H\)^{-1}\opgam}{2} \Nrm{\indic_{H \geq \lambda} \com{\opc_\ii ,H} \(\lambda - H\)^{-1}\opgam}{2}
			\\
			+ \Tr{\indic_{0< H< \lambda} \(A^*\opgam\, A\, m + m\,A^* \opgam\, A - m\,\opgam \n{A}^2\)} + \Tr{\com{A^*,m} \opgam\, A}
		\end{multline*}
		where $(\opb_\ii,\opc_\ii)_{\ii=1\dots 3} = \((A^*,A\,m),(m\,A^*,A),(-m,\n{A}^2)\)$.
	\end{prop}
 
	\begin{proof}[Proof of Proposition~\ref{prop:comm_HS_weight}]
		Using the cyclicity of the trace
		\begin{align}\nonumber 
			\Tr{\n{\com{A,\opgam}}^2 m} &= \Tr{\(A^*\opgam\, A + \opgam \n{A}^2 \opgam - A^* \opgam \, A \, \opgam - A\, \opgam \, A^* \, \opgam\) m}
			\\\label{eq:1}
			&= \Tr{\(\id-\opgam\) \(A^*\opgam\, A\, m + m\,A^* \opgam\, A - m\,\opgam \n{A}^2\) + \com{A^*,m} \opgam\, A} .
		\end{align}
		With the notation $B = A^*\opgam\, A\, m + m\,A^* \opgam\, A - m\,\opgam \n{A}^2$, one can then decompose $\id - \opgam = \indic_{0< H< \lambda} + \indic_{H\geq \lambda}$ to get 
		\begin{equation}\label{eq:2}
			\Tr{\n{\com{A,\opgam}}^2 m} = \Tr{\indic_{H\geq \lambda}\, B} + \Tr{\indic_{0< H< \lambda}\, B} + \Tr{\com{A^*,m} \opgam\, A}
		\end{equation}
		and it remains to treat the first term of the right-hand side, $\Tr{\indic_{H\geq \lambda}\, B}$, which can be written as a sum of terms of the form $\Tr{\indic_{H\geq \lambda}\, \opb \, \opgam\, \opc}$. Keeping the same notation as in the proof of Proposition~\ref{prop:comm2} we have 
		\begin{equation*} %\label{eq:proj0}
			\Tr{\indic_{H\geq \lambda}\, \opb \, \opgam\, \opc} = \sum_{\lambda_k \geq \lambda} \sum_{\lambda_j \leq 0} \frac{1}{\(\lambda_k - \lambda_j\)^2} \Tr{\opProj_k \com{H,\opb} \opProj_j \com{\opc ,H} \opProj_k} .
		\end{equation*}
		Following the proof of Proposition~\ref{prop:comm2} we get the estimate
		\begin{equation*}
			\n{\Tr{\indic_{H\geq \lambda}\, \opb \, \opgam\, \opc}} \leq \Nrm{\indic_{H\geq \lambda} \com{H,\opb} \opgam \(\lambda - H\)^{-1}}{2} \Nrm{\indic_{H\geq \lambda} \com{H,\opc} \opgam \(\lambda - H\)^{-1}}{2} . 
		\end{equation*}
		The last inequality combined with \eqref{eq:1} and \eqref{eq:2} finishes the proof. 
	\end{proof}
 
	Our first step towards the proof of the $L^2$ regularity is the following lemma, which is based on an application of the weighted Hilbert--Schmidt commutator bound in Proposition~\ref{prop:comm_HS_weight}. 
	\begin{lem}\label{lem:comm_HS_weight_Dx}
		Let $\opgam = \indic_{H\leq 0}$ and $H = -\hbar^2 \Delta + V$ with $d \geq 3$. Then
		\begin{equation*}\label{eq:comm_HS_weight_Dx}
			\Nrm{\Dhx\opgam \n{\opp}^2}{\L^2}^2 \leq \frac{\cC_3}{\hbar} \(\sfA_\opgam + \sfB_\opgam\) 
		\end{equation*}
		with $\cC_3 =3 \, \cC_1 + \cC_2$, and 
		\begin{align*}
			\sfA_\opgam &= \frac{1}{\hbar} \Nrm{\nabla V\,\opgam}{\L^\infty} \Nrm{\com{\n{\opp}^4 \opp,V}\opgam}{\L^\infty} + \frac{1}{\hbar^2} \Nrm{\com{\n{\opp}^2,V}\opgam}{\L^\infty} \Nrm{\com{\n{\opp}^4,V}\opgam}{\L^\infty}
			\\
			\sfB_\opgam &= \sNrm{\weight{\opp}^5 \opgam}{\L^\infty} .
		\end{align*}
		Additionally,
		\begin{equation*}
			\Nrm{\nabla \vr_\opgam}{L^2} \leq \frac{\Nrm{\weight{x}^{-2}}{L^2(\Rd)}}{\sqrt{\hbar}} \(\Nrm{\Dhx\opgam \n{\opp}^2}{\L^2} + \Nrm{\Dhx\opgam}{\L^2}\).
		\end{equation*}
	\end{lem}
 
	\begin{proof}
		Taking $A = \nabla$, $m = \n{\opp}^4$ and $\lambda = \hbar$ in Proposition~\ref{prop:comm_HS_weight} gives
		\begin{equation*}
			\Nrm{\Dhx\opgam \n{\opp}^2}{\L^2}^2 =
			\hd \Tr{\n{\com{\nabla,\opgam}}^2 \n{\opp}^4} \leq \sfA_\opgam \Tr{\(\hbar - H\)^{-2} \opgam} + 3 \, \sfB_\opgam \, \hbar^{-2} \, \Tr{\indic_{0< H< \hbar}}
		\end{equation*}
		where $\sfA_\opgam$ and $\sfB_\opgam$ are defined in the statement. Therefore, by Lemma~\ref{lem:local_estimate} and Lemma~\ref{lem:trace}
		\begin{equation*}
			\Nrm{\Dhx\opgam \n{\opp}^2}{\L^2}^2 \leq \(3 \,\cC_1 + \cC_2\)
			\frac{\sfA_\opgam + \sfB_\opgam}{\hbar} \, .
		\end{equation*}
		To get the second equation we use the fact that $\hd \,\Dhx\op(x,x) = \nabla\rho$ and the Cauchy--Schwarz inequality for the trace to get for any $\varphi\in C^\infty_c(\Rd)$ 
		with $\Nrm{\varphi}{L^2} =1 $
		\begin{equation*}
			\n{\intd \nabla\rho\, \varphi} = \hd \Tr{\Dhx\opgam\, \varphi} 
			%= \hd \Tr{\Dhx\opgam \weight{\opp}^2 \weight{\opp}^{-2} \varphi} \\ 
			\leq \Nrm{\Dhx\opgam \weight{\opp}^2}{\L^2} \Nrm{\weight{\opp}^{-2} \varphi}{\L^2} . 
		\end{equation*}
		Finally, we compute $\Nrm{\weight{\opp}^{-2} \varphi}{\L^2} \leq \Nrm{\weight{x}^{-2}}{L^2(\Rd)} $ and use the triangle inequality.
	\end{proof}
 
	The previous lemma is sufficient to analyze the region $\n{x}\leq \hbar$. To complete the proof of the $L^2$ regularity of $\vr_\opgam$, it remains to look at the cases where $\n{x}\geq \hbar$. Recall our notation $\sfT_x\opgam = \tau_x\,\opgam\,\tau_{-x}$.
 
	\begin{lem}\label{lem:comm5}
		Let $\opgam = \indic_{H\leq 0}$ and $H = -\hbar^2 \Delta + V$ with $d \geq 3$. Let $\beta \geq 0$. Then, for any $x_0 \in\Rd$ such that $\n{x_0}\in(\hbar, 1)$,
		\begin{equation*}
			\hd\Tr{\n{\sfT_{x_0}\opgam-\opgam}^2 \n{\opp}^4} \leq \cC_4 
			\(\sfA_\opgam ' + \sfB_\opgam '\) \n{x_0}
		\end{equation*}
		where $\cC_4 = 2\,\cC_1 + e^{2 \beta}\; \cC_2$, and 
		\begin{align*}
			\sfA_\opgam' &= \(\frac{1}{\hbar}\Nrm{\com{\n{\opp}^4,V}\opgam}{\L^\infty} + \Nrm{e^{-\beta \n{x}}\,\nabla V}{L^\infty}\Nrm{e^{\beta \n{x}} \n{\opp}^4 \opgam}{\L^\infty}\)
			\Nrm{e^{-\beta \n{x}}\,\nabla V}{L^\infty} \Nrm{e^{\beta \n{x}}\,\opgam}{\L^\infty}
			\\
			\sfB_\opgam' &= 2 \Nrm{\n{\opp}^4 \opgam}{\L^\infty} .
		\end{align*}
	\end{lem}
 
	\begin{proof}
		Let $x_0\in\Rd$ be such that $\n{x_0}\in [\hbar,1]$. Then, taking $A = \tau_{x_0} = e^{i\,x_0\cdot \opp/\hbar}$, $m = \n{\opp}^4$ and $\lambda = \n{x_0} \leq 1 $ in Proposition~\ref{prop:comm_HS}, and then using the lemmas~\ref{lem:trace} and ~\ref{lem:local_estimate} for $\lambda = \n{x_0}\geq \hbar$ 
		\begin{multline*}
			\hd\Tr{\n{\sfT_{x_0}\opgam-\opgam}^2 \n{\opp}^4}
			\\
			\leq 2\,\cC _1 \Nrm{\n{\opp}^4 \opgam}{\L^\infty} + \frac{\cC_2}{\n{x_0}^2} \Nrm{\(\sfT_{x_0}V-V\)\opgam}{\L^\infty} \Nrm{\com{\tau_{x_0}\n{\opp}^4,V}\opgam}{\L^\infty}.
		\end{multline*}
		%	Hence it follows from Lemma~\ref{lem:trace}, the local estimate~\ref{lem:local_estimate} and the fact that $\n{x_0}\geq \hbar$ that
		%	\begin{align*}
			%		\hd\Tr{\n{\sfT_{x_0}\opgam-\opgam}^2 \n{\opp}^4} 
			%		\leq \(
			% \cC_1 C_\opgam ' + \cC_2 D_\opgam ' \) \n{x_0}.
			%	\end{align*}
		We now estimate the term $\n{x_0}^{-2}$ by extracting a factor $\n{x_0}$ from each norm. Notice that by the Leibniz rule for commutators, the fact that $\tau_{x_0}$ commutes with $\n{\opp}^4$ and the unitarity of $\tau_{x_0}$, it follows that
		\begin{equation*}
			\Nrm{\com{\tau_{x_0}\n{\opp}^4,V}\opgam}{\L^\infty} \leq \frac{\n{x_0}}{\hbar} \Nrm{\com{\n{\opp}^4,V}\opgam}{\L^\infty} + \Nrm{\(\sfT_{x_0}V-V\) \tau_{x_0} \n{\opp}^4 \opgam}{\L^\infty} 
		\end{equation*}
		where we used $1 \leq\frac{\n{x_0}}{\hbar}$. 
		%	The first term on the right-hand side is bounded using Lemma~\ref{lem:comm_p5_V} and the fact that $\hbar \leq \n{x_0}$ by
		%	\begin{equation*}
			%		\Nrm{\com{\n{\opp}^4,V}\opgam}{\L^\infty} \leq C_d\, C_{\opgam,V} \(1 + \Nrm{V_-}{L^\infty}\) \n{x_0}.
			%	\end{equation*}
		On the other hand, it follows from a first order Taylor formula that
		\begin{equation*}
			\n{\sfT_{x_0}V-V} \leq \Nrm{e^{-\beta \n{x}}\nabla V}{L^\infty} \n{x_0} \int_0^1 e^{\beta\n{\(1-\theta\)x+\theta x_0}} \d \theta \leq \Nrm{e^{-\beta \n{x}}\nabla V}{L^\infty} \n{x_0} e^{\beta \n{x_0}}.
		\end{equation*}
		From this we deduce that 
		\begin{align*}
			\Nrm{\(\sfT_{x_0}V-V\) \tau_{x_0} \n{\opp}^4 \opgam}{\L^\infty} &= \Nrm{\tau_{x_0}\(V-\sfT_{-x_0}V\) \n{\opp}^4 \opgam}{\L^\infty}
			\\
			&\leq \n{x_0} e^{\beta \n{x_0}} \Nrm{e^{-\beta \n{x}}\nabla V}{L^\infty} \Nrm{e^{\beta \n{x}} \n{\opp}^4 \opgam}{\L^\infty} ,
		\end{align*}
		and similarly
		\begin{equation*}
			\Nrm{\(\sfT_{x_0}V-V\)\opgam}{\L^\infty} \leq \n{x_0} e^{\beta \n{x_0}} \Nrm{e^{-\beta \n{x}}\nabla V}{L^\infty} \Nrm{e^{\beta \n{x}}\opgam}{\L^\infty} . 
		\end{equation*}
		This finishes the proof after we gather all estimates. 
	\end{proof}
	
	The multiple terms appearing in the previous Lemmas are estimated as follows. 

	\begin{lem}\label{lemma:alt}
		Let $\opgam = \indic_{H \leq 0}$ with $H = - \hbar^2 + V(x)$ in $d \geq 1$. Assume $V$ satisfies \eqref{hyp:V4}. Then, for any $\beta>0$ there exists $C>0$ such that
		\begin{align}\label{eq:p2}
			\Nrm{\com{\n{\opp}^2, V} \opgam}{\L^\infty} 
			&\leq C \hbar \Nrm{e^{\beta \n{x}} \weight{\opp} \opgam}{\L^\infty} 
			\\\label{eq:p4}
			\Nrm{\com{\n{\opp}^4, V}\opgam}{\L^\infty} 
			&\leq C \hbar \Nrm{e^{\beta \n{x}} \weight{\opp}^3 \opgam}{\L^\infty}
			\\\label{eq:p5} 
			\Nrm{\com{\n{\opp}^4 \opp,V}\opgam}{\L^\infty} 
			&\leq C \hbar \Nrm{e^{\beta \n{x}} \weight{\opp}^4 \opgam}{\L^\infty}
		\end{align}
		Additionally, 
		\begin{equation*}
			\Nrm{e^{\beta \n{x}} \n{\opp}^4 \opgam}{\L^\infty}
			\leq C \Nrm{e^{2 \beta \n{x}} \weight{\opp}^2 \opgam}{\L^\infty}.
		\end{equation*}
	\end{lem}
 
	\begin{proof}
		For notational simplicity we write $\nrm{\cdot} := \Nrm{\cdot}{\L^\infty}$. Additionally, for an operator $A$ we let $\nabla A := \com{\nabla, A}$ and $\Delta A := \com{\nabla, \cdot \com{\nabla, A}}$. First, we compute 
		\begin{equation}\label{eq:rec1}
			\nrm{\com{\n{\opp}^2, A} \op} \leq \hbar^2 \nrm{\Delta A \, \op} + 2\, \hbar \nrm{\(\nabla A \cdot \opp\) \op} ,
		\end{equation}
		and \eqref{eq:p2} follows with $A = V$ and $\op = \opgam$, and using Hypothesis~\eqref{hyp:V4}. Next, a similar computation shows
		\begin{equation}\label{eq:rec2}
			\nrm{\com{\n{\opp}^4, A} \op} \leq \nrm{\com{\n{\opp}^2, \com{\opp^2, A}} \op} + 2 \nrm{\com{\n{\opp}^2, A} \n{\opp}^2 \op} .
		\end{equation}
		For the first term in~\eqref{eq:rec2}, we use Inequality~\eqref{eq:rec1} with $\com{\n{\opp}^2, A}$ instead of $A$ so that 
		\begin{equation*}
			\nrm{\com{\n{\opp}^2, \com{\opp^2, A}} \op} \leq \hbar^2 \nrm{\com{\n{\opp}^2, \Delta A}\op} + 2\, \hbar \nrm{\com{\n{\opp}^2, \nabla A \cdot \opp} \op}
		\end{equation*}
		where we used $\Delta \com{\n{\opp}^2, B} = \com{\n{\opp}^2, \Delta B}$ and $\nabla \com{\n{\opp}^2, B} = \com{\n{\opp}^2, \nabla B}$ for any operator $B$. Using again~\eqref{eq:rec1} we find for the two terms 
		\begin{align*}
			\nrm{\com{\n{\opp}^2, \Delta A}\op} &\leq \hbar^2 \nrm{\Delta^2 A \, \op} + 2\, \hbar 
			\nrm{\nabla \Delta A \cdot \opp\, \op}
			\\ 
			\nrm{\com{\n{\opp}^2, \nabla A \cdot \opp}\op} &\leq \hbar^2 \nrm{\(\nabla\Delta A \cdot \opp\) \op} + 2 \hbar \nrm{\(\nabla \(\nabla A \cdot\opp \) \cdot \opp\) \op} . 
		\end{align*}
		For the second term in \eqref{eq:rec2}, we use Inequality~\eqref{eq:rec1} with $\n{\opp}^2\op$ instead of $\op$. Putting all together we obtain 
		\begin{align}\nonumber 
			\nrm{{\n{\opp}^4, A} \op} &\leq \hbar^2 
			\(\hbar^2 \nrm{\Delta^2 A \, \op} + 2\, \hbar \nrm{\(\nabla \Delta A \cdot \opp\) \op}\)
			\\\label{eq:rec3} 
			&\qquad + 2\, \hbar \(\hbar^2 \nrm{\(\nabla\Delta A \cdot \opp\) \op} + 2\, \hbar \nrm{\(\nabla \(\nabla A \cdot\opp \) \cdot \opp\) \op} \)
			\\\nonumber
			&\qquad + 2 \(\hbar^2 \nrm{\Delta A \n{\opp}^2 \op} + 2\, \hbar \nrm{\(\nabla A \cdot \opp\) \n{\opp}^2 \op}\) .
		\end{align}
		The bound~\eqref{eq:p4} follows by taking $A = V$, $\op = \opgam$ and using~\eqref{hyp:V4}. As for the other inequality, 
		\begin{equation*}
			\nrm{\com{\opp \n{\opp}^4, A} \op} \leq \hbar \nrm{\com{\n{\opp}^4, \nabla A}\op} + 2		\nrm{\com{\n{\opp}^4, A} \opp \,\op} . 
		\end{equation*}
		so that~\eqref{eq:p5} follows by an appropriate application of~\eqref{eq:rec3}. Finally, for the last bound, we write
		\begin{equation*}
			\n{\opp}^4 \opgam =\n{\opp}^2 \(H - V\) \opgam = \n{\opp}^2 \opgam H - V \n{\opp}^2\opgam - 2\, \hbar \,\nabla V\cdot \opp\, \opgam - \Delta V \opgam\, , 
		\end{equation*}
		from which we conclude
		\begin{equation*}
			\nrm{e^{\beta \n{x}} \n{\opp}^4 \opgam} \leq \nrm{e^{\beta \n{x}} \n{\opp}^2 \opgam} \nrm{\opgam H} + \nrm{e^{\beta \n{x}} V \n{\opp}^2 \opgam} + 2 \, \hbar \nrm{e^{\beta\n{x}} \,\nabla V \cdot \opp\, \opgam}
			+ \nrm{e^{\beta \n{x}}\, \Delta V \opgam}
		\end{equation*}
		and the desired bound then follows by using \eqref{hyp:V4}. 
	\end{proof}

	Finally, we turn to the proof of the main result of this subsection.

	\begin{proof}[Proof of Proposition~\ref{prop:regu_rho_L2}]
		We consider $d =3$, and separate the proof into three cases. First, when $\n{x}\leq \hbar$, we use Lemma~\ref{lem:comm_HS_weight_Dx}. In order to control the constants $\sfA _\opgam$ and $\sfB_\opgam$ it suffices to show that
		\begin{equation*}
			\sNrm{\weight{\opp}^5 \opgam}{\L^\infty} \, ,
			\quad 
			\frac{1}{\hbar} \Nrm{\com{\opp \n{\opp}^4,V}\opgam}{\L^\infty} , 
			\quad 
			\frac{1}{\hbar} \Nrm{\com{\n{\opp}^2,V}\opgam}{\L^\infty} ,
			\quad 
			\frac{1}{\hbar} \Nrm{\com{\n{\opp}^4,V}\opgam}{\L^\infty}
		\end{equation*}
		are uniformly bounded in $\hbar$. The first quantity is controlled thanks to
		Lemma~\ref{lem:weight_6}. The other three are controlled thanks to Lemma~\ref{lemma:alt} in terms of $\Nrm{e^{2 \beta \n{x}} \weight{\opp}^2 \op}{\L^\infty}$. The latter is bounded thanks to Agmon's estimates in Lemma~\ref{lem:agmon2}. Secondly, for $\hbar \leq \n{x} \leq 1$, we use Lemma~\ref{lem:comm5}. In this case, the control of the constants $\sfA'_\opgam$ and $\sfB'_\opgam$ is dealt with analogously. Finally, when $\n{x}\geq 1$, then we just notice that
		\begin{equation*}
			\Nrm{T_x\vr_\opgam-\vr_\opgam}{L^2} \leq 2 \Nrm{\vr_\opgam}{L^2} \leq 2 \Nrm{\vr_\opgam}{L^\infty}^{1/2} \Nrm{\vr_\opgam}{L^1}^{1/2}
		\end{equation*}
		and we know by Inequality~\eqref{eq:CLR} and Lemma~\ref{lem:Linfty} that these quantities are bounded independently of $\hbar$. The bound in $B^{1/2}_{2,\infty}$ then follows by definition of Besov spaces and the $H^s$ bound by the continuous embedding $B^{1/2}_{2,\infty} \subset H^s$ for $s<1/2$.
		%		To get the regularity of $\sqrt{\vr_\opgam}$, we cut the integral in two pieces. On the one hand, by our estimate on $\nabla\rho$,
		%		\begin{equation*}
		%			\int_{\vr_\opgam\geq \eps} \n{\nabla\sqrt{\vr_\opgam}}^2 = \int_{\vr_\opgam\geq \eps} \frac{1}{4\,\vr_\opgam} \n{\nabla\vr_\opgam}^2 \leq \frac{1}{4\,\eps} \Nrm{\nabla\vr_\opgam}{L^2}^2.
		%		\end{equation*}
		%		On the other hand, in the spirit of the Hoffmann--Ostenhof inequality, when $\vr_\opgam$ is small, we diagonalize $\opgam$ as in Formula~\eqref{eq:rho_2_diagonalisation} and use the Cauchy--Schwarz inequality to get
		%		\begin{equation*}
		%			\n{\nabla\vr_\opgam}^2 = \n{\hd\sum_{j\geq 0}\lambda_j 2\Re{\conj{\psi_j}\nabla\psi_j}}^2 \leq \frac{2}{\hbar^2} \,\vr_\opgam \,\vr_{\opp\cdot\opgam\,\opp}.
		%		\end{equation*}
		%		and so by Inequality~\eqref{eq:Agmon_rho_2},
		%		\begin{equation*}
		%			\int_{\vr_\opgam\leq \eps} \n{\nabla\sqrt{\vr_\opgam}}^2 \leq \frac{1}{2\,\hbar^2} \int_{\vr_\opgam\leq \eps} \(\frac{\hbar^2}{2}\,\Delta - V\)\vr_\opgam
		%		\end{equation*}
	\end{proof}

\section{Interacting case: Hartree theory}

	In this section, we consider the case of interacting particles, that is, as described in introduction, the potential is given by
	\begin{equation*}
		V_\op := K * \vr_\op + U
	\end{equation*}
	with $K(x) = \kappa \n{x}^{-a}$ with $\kappa>0$ and $a\in[0,1]$ and $U : \Rd\to\R$ some external potential and which we assume satisfies Hypothesis~\eqref{hyp:U}. We study the minimizers $\opgam := \opHF$ of the Hartree energy functional $\scE_\op$ given in Equation~\eqref{eq:Hartree_enregy}. They satisfy the equation (see e.g.~\cite{nguyen_weyl_2024})
	\begin{equation}\label{mini}
		\opgam = \indic_{H_\opgam < 0} + \opq 
		\quad \text{ with } \quad
		H_\opgam := -\hbar^2\Delta + V_\opgam
	\end{equation}
	where $ 0 \leq \opq \leq 1 $ is a self-adjoint operator such that $\ran(\opq) \subset \ker(H_{\opgam})$. In particular, 
	\begin{equation}\label{min:1}
		\opgam \leq \indic_{H_{\opgam \leq 0}}\,. 
	\end{equation}
	One of the main ingredients in the proof of the commutator estimates will be to prove that the potential $V_\opgam$ satisfies the assumptions explained of the linear part, i.e. some estimates on its regularity and growth uniform in $\hbar$. This is the goal of the next subsection.

\subsection{A priori estimates}

	Let us start with recalling some well-known estimates on the convolution $K * \rho$ by the singular kernel $K$.
	\begin{lem}\label{lem:conv}
		Let $d=3$, and denote $K(x) = \kappa \n{x}^{-a}$ for $\kappa \in \R$ and $a \in (0,1]$. Let $\rho\in L^1\cap L^\infty$. Then for all $\delta>0 $ and $\alpha \in \(0,1\)$
		\begin{align*}
			\Nrm{K * \rho}{L^\infty}
			&\leq \n{\kappa}\(\delta^{-a} \Nrm{\rho}{L^1} + C_a\, \delta^{d-a} \Nrm{\rho}{L^\infty}\)
			\\
			\Nrm{\nabla K * \rho}{L^\infty} 
			&\leq C_{\theta} \n{\kappa} \Nrm{\rho}{L^\infty}^\theta \Nrm{\rho}{L^1}^{1-\theta}
			\\ 
			\Nrm{\nabla K * \rho}{C^{0,\alpha}}
			&\leq C \n{\kappa} \(\Nrm{\rho}{L^1} + \Nrm{\rho}{L^\infty}\)
		\end{align*}
		with $C_a = \tfrac{4\,\pi}{3-a}$, $\theta = \frac{a+1}{3}$ and $C_{\theta} = \frac{a}{1-\theta}\, (\frac{4\,\pi}{3\,\theta})^\theta$. Additionally, in the case when $a<1$ there is a constant $C = C(a)>0$ such that 
		\begin{equation*}
			\Nrm{\nabla^2 K * \rho}{L^\infty} \leq C \n{\kappa} \(\Nrm{\rho}{L^1} + \Nrm{\rho}{L^\infty}\).
		\end{equation*}
		Moreover, in the case when $a=1$, there is $C>0$ such that 
		\begin{equation*}
			\Nrm{\nabla^2K * \rho}{L^\infty} \leq C \(1+ \Nrm{\rho}{L^1} + \Nrm{\rho}{L^\infty} \ln(1+\Nrm{\nabla\rho}{L^\infty})\).
		\end{equation*}
	\end{lem}

	\begin{proof}
		Observe that for all $a \in (0,1]$ we have $\n{x}^{-a} \in L^1_{\loc} (\R^3)$ and we split
		\begin{equation*}
			K* \rho(x) = 
			\int_{\n{x - y} \leq \delta} \frac{\rho(y)}{\n{x -y}^a} \d y 
			+ 
			\int_{\n{x - y} \geq \delta} \frac{\rho(y)}{\n{x -y}^a} \d y \, . 
		\end{equation*}
		The first integral can be estimated using $\Nrm{\rho}{L^\infty}$ and the $\d y$ integral is of order $\delta^{d-a}$. For the second integral we bound $\n{x -y}^{-a} \leq \delta^{-a}$ and use the $\Nrm{\rho}{L^1}$ norm. The estimates for the derivatives are identical.

		In the case of the Coulomb potential, it is well-known that $\nabla K* \rho$ is log-Lipschitz (see e.g. \cite{lions_propagation_1991, iacobelli_enhanced_2024}) and so in particular H\"older continuous. The last inequality is a variant of the inequality used in \cite{beale_remarks_1984} (a proof can be found in \cite{kato_well-posedness_1986}) and can be proved again by splitting the integrals and using $\Nrm{\nabla \rho}{L^\infty}$ when $\n{x-y}$ is small.
	\end{proof}

	As a consequence, we obtain the $L^p$ estimates which depend only on the external trap.
	\begin{prop}\label{prop:Lp}
		Let $d=3$ and let $\opgam$ be as in \eqref{mini}. Then, for all $p \in [1,\infty]$ there exists $C = C_{\beta, R, a ,p}>0$ independent of $\hbar$ such that
		\begin{equation*}
			\Nrm{\vr_\opgam}{L^p} \leq C \(1 + \Nrm{U_-}{L^\infty} + \int_{\R^3} U_-^{3/2}\).
		\end{equation*}
	\end{prop}

	\begin{proof}
		Since $K\geq 0$, it follows that $V_\opgam \geq U$, and so $\(V_\opgam\)_- \leq U_-$. In particular, Inequality~\eqref{min:1}, Agmon's estimates and Lemma~\ref{lem:Linfty} give the exponential decay
		\begin{equation*}
			\Nrm{e^{\beta \n{x}}\, \opgam}{\L^\infty} \leq C_{d,\beta} \(1+ \Nrm{U_-}{L^\infty} + \int_{\R^3} U_-^{3/2}\)
		\end{equation*}
		where we also used the $L^1$ estimate~\eqref{eq:CLR}. Observe now that the right-hand side is independent of $\hbar$ and $\vr_\opgam$. As a consequence, we can now bootstrap the $L^\infty$ estimates for $\vr_\opgam$. More precisely, from Lemma~\ref{lem:Linfty} and Lemma~\ref{lem:conv} 
		\begin{align*} 
			\Nrm{\vr_\opgam}{L^\infty}
			&\leq C \(1 + \Nrm{U_-}{L^\infty} + \Nrm{U\, \opgam}{\L^\infty} + \Nrm{K* \vr_\opgam}{L^\infty}\)
			\\
			&\leq C \(1 +\Nrm{U_-}{L^\infty} + \Nrm{U e^{- \beta \n{x}}}{L^\infty} \Nrm{e^{\beta \n{x}}\, \opgam}{\L^\infty} + \delta^{3-a}\Nrm{\vr_\opgam}{L^\infty} + \delta^{-a} \Nrm{\vr_\opgam}{L^1}\) . 
		\end{align*}
		Therefore we choose $\delta = \delta(a)>0$ small enough and use again Inequality~\eqref{eq:CLR} to conclude 
		\begin{equation*}
			\Nrm{\vr_\opgam}{L^\infty} \leq C_{d, \beta, R, a} \(1 + \Nrm{U_-}{L^\infty} + \intd U_-^{3/2} \) .
		\end{equation*}
		The estimates for the $L^1$ norm come from~\eqref{eq:CLR} with $\kappa\geq 0$. It then suffices to interpolate. 
	\end{proof}

\subsubsection{Commutator estimates}
	
	Thanks to the uniform $C^{1,\alpha}$ estimates of Lemma~\ref{lem:conv} and~\cite[Theorem~1.5]{mikkelsen_sharp_2023} we can find $\eps_0 >0$ small enough so that 
	\begin{equation}\label{eq:soren3}
		\n{\hd \Tr{\indic_{H_\opgam \leq E}} - \iint_{\n{\xi}^2 + V_\opgam(x)\leq E} \d x \d \xi} \leq \cC _0 \, \hbar
	\end{equation}
	for all $E \in [- \eps_0, \eps_0]$, and all $\hbar \in (0,1)$. In particular, the local eigenvalue estimate contained in Lemma~\ref{lem:local_estimate} implies 
	\begin{equation}\label{eq:soren4}
		\hd \tr \indic_{[a,b]}(H_\opgam) \leq \cC_1 \(\n{b-a} + \hbar \) , \qquad - \eps_0 \leq a \leq b \leq \eps_0 \, .
	\end{equation}
	Thanks to $K \geq 0$, the constant $\cC_1>0$ depends only on $\cC_0>0$, the external trap $U(x)$, and the parameter $\eps_0$. We are now ready for the proof of Theorem~\ref{thm:commutator_nonlinear}.

	\begin{proof}[Proof of Theorem~\ref{thm:commutator_nonlinear}]
		Arguing as in the proof of Theorem~\ref{thm:commutator:linear}, it suffices to look at the case $\hbar \leq h_\beta := \min(1, 1 /8\beta)$. For the proof, let us take $\opgam$ as in \eqref{mini} and write
		\begin{equation*}
			\opgam = \widetilde \opgam + \opq \, , 
			\quad \text{ with } \quad
			\widetilde \opgam = \indic_{H_\opgam < 0} \, .
		\end{equation*}
		First, let us remove the error term $\opq$. Consider $A = A^*$ a self-adjoint operator. Then, since $0 \leq \opq \leq \id$ and $\opq = \indic_{H_{\opgam}=0}\, \opq \leq \indic_{H_{\opgam} \leq0}$ and Inequality~\eqref{eq:soren4}, we get 
		\begin{equation*}
			\Nrm{\com{A , \opq}}{\L^1} \leq 2 \Nrm{\opq\, A}{\L^1} \leq \hd\tr{\indic_{H_{\opgam}=0}} \Nrm{\opq\, A}{\L^\infty} \leq \cC_1 \Nrm{A \,\indic_{H_{\opgam} \leq 0}}{\L^\infty} \hbar \,. 
		\end{equation*}
		Arguing as in the proof of Theorem~\ref{thm:commutator:linear}, $\Nrm{x \, \indic_{H_{\opgam} \leq 0}}{\L^\infty}$ and $\Nrm{\opp \, \indic_{H_{\opgam} \leq 0}}{\L^\infty}$ are bounded uniformly in $\hbar$. Thus, there exists a constant $C$ independent of $\hbar$ such that
		\begin{equation*}
			\hd \Tr{\n{\com{x, \opq}}} \leq C\,\hbar \quad \text{ and } \quad \hd \Tr{\n{\com{\opp, \opq}}} \leq C\,\hbar \, .
		\end{equation*}
		Thanks to H\"older's inequality for Schatten norms, $\Nrm{\com{A , \opq}}{\L^p} \leq 2^{1 / p '}\Nrm{A \, \opq}{\L^\infty}^{1/p'} \Nrm{\com{A, \opq}}{\L^1}^{1/p}$ and we can automatically control the $p>1$ norms. Thus, it suffices to analyze $\widetilde \opgam$ and, for notational simplicity, we drop the tilde $\opgam := \widetilde \opgam$ throughout the proof. Let us also note that the same argument implies that all estimates contained in Section~\ref{sec:com} hold for either $\indic_{H \leq 0}$ and $\indic_{H <0}$.
		
		Secondly, we look at the Hilbert--Schmidt estimates. A straightforward application of Proposition~\ref{prop:linear:comm_HS} implies that there is $C > 0 $ such that for all 
		$\hbar \in (0, h_\beta]$
		\begin{align*}
			\hd \Tr{\n{\com{x ,\opgam}}^2} &\leq C \,h \(\Nrm{x\,\opgam}{\L^\infty}^2 + \Nrm{\opp\, \opgam}{\L^\infty}^2\)
			\\ 
			\hd \Tr{\n{\com{\opp ,\opgam}}^2} &\leq C \, h \(\Nrm{\opp\,\opgam}{\L^\infty}^2 + \Nrm{\nabla V_\opgam \opgam}{\L^\infty}^2\) .
		\end{align*}
		Similarly as in the proof of Theorem~\ref{thm:commutator:linear}, all terms on the right-hand side can be bounded using either energy estimates, or Agmon's estimates. Note that here we have thanks to Lemma~\ref{lem:conv}
		\begin{equation}\label{eq:derivative}
			\Nrm{\nabla V_\opgam\, \opgam}{\L^\infty} \leq C_{\beta} \Nrm{\opgam}{\L^\infty} + C_\beta\Nrm{e^{\beta \n{x}}\, \opgam}{\L^\infty}
		\end{equation}
		which is controlled with Agmon's estimates, for all $0 \leq a \leq 1$.

		Thirdly, for the trace-class estimates we apply Proposition~\ref{prop:linear:commTR} for $\mu = \Nrm{U_-}{L^\infty} + 1$ to find that there is $C > 0$ such that for all $ \hbar \in (0, h_\beta]$ 
		\begin{align}\label{eq:commx}
			\hd \tr{\n{\com{x,\opgam}}} &\leq C\, \hbar \(\Nrm{x\,\opgam}{\L^\infty} + \Nrm{\nabla V_\opgam\, \opgam}{\L^\infty}\)
			\\\nonumber
			\hd \tr{\n{\com{\opp ,\opgam}}} &\leq C \, \hbar \(\Nrm{\opp \,\opgam}{\L^\infty} + \Nrm{\nabla^2 V_\opgam(x) \opgam}{\L^\infty} + \Nrm{\nabla^2 V_\opgam (x) \cdot \opp \, \opgam}{\L^\infty}\) . 
		\end{align} 
		Let $0 \leq a < 1$. Then, the estimate for the first derivative~\eqref{eq:derivative} also holds for two derivatives. We repeat the same argument as before and conclude that most terms are uniformly bounded in $\hbar$. The only exception is the following mixed term which, thanks to hypotheses~\eqref{hyp:K} and \eqref{hyp:U}, and Lemma~\ref{lem:conv}, satisfies the bound 
		\begin{equation*}
			\Nrm{\nabla^2 V_\opgam(x) \cdot \opp \, \opgam}{\L^\infty}\leq C_{\kappa, a}\Nrm{\opp\, \opgam}{\L^\infty} + C_\beta\Nrm{e^{\beta \n{x}} \opp\, \opgam}{\L^\infty} .
		\end{equation*}
		It suffices to estimate $\sNrm{e^{\beta \n{x}}\, \opp\, \opgam}{\L^\infty}$, to
		which we use the generalized Agmon's estimates from Lemma~\ref{lem:agmon_gradients}.
		
		Consider now $a =1$. First, we observe that the argument for the position estimates in~\eqref{eq:commx} is unchanged, since it only requires one bounded derivative. On the other hand, for the momentum we use the second trace-class bounds in Proposition~\ref{prop:comm_frac} and Lemma~\ref{lem:frac} to compute for $\mu = 1 + \Nrm{U_-}{L^\infty}$ and any $\lambda \geq \hbar $
		\begin{equation*}
			\hd \tr{\n{\com{\opp ,\opgam}}} \leq C \(\lambda \Nrm{\opp \opgam}{\L^\infty} + \frac{\hbar}{\lambda} \(\Nrm{\nabla^2 V_\opgam (x) \cdot \opp \, \opgam}{\L^\infty} + \Nrm{\nabla^2 V_\opgam(x) \opgam}{\L^\infty}
			\)\) .
		\end{equation*}
		We split $V_\opgam = U + K * \vr_\opgam$ and obtain
		\begin{align*}
			\Nrm{\nabla^2 V_\opgam(x) \opgam}{\L^\infty} &\leq \Nrm{e^{-\beta \n{x}} \nabla^2 U}{L^\infty} \Nrm{e^{\beta \n{x}} \opgam}{\L^\infty} + \Nrm{\nabla^2 K * \vr_{\opgam}}{L^\infty} \Nrm{\opgam}{\L^\infty}
			\\
			\Nrm{\nabla^2 V_\opgam(x) \cdot \opp \opgam}{\L^\infty} &\leq \Nrm{e^{- \beta \n{x}} \nabla^2 U}{L^\infty} \Nrm{e^{\beta \n{x}} \opp \opgam}{\L^\infty} + \Nrm{\nabla^2 K * \vr_{\opgam}}{L^\infty} \Nrm{\opp\, \opgam}{\L^\infty} . 
		\end{align*}
		By assumption and thanks to Agmon's estimates, the $U$-dependent terms are uniformly bounded in $\hbar$. On the other hand, thanks to Lemma~\ref{lem:conv} and the $L^p$ estimates for the density $\vr_\opgam$ we get
		\begin{equation*}
			\Nrm{\nabla^2K * \vr_\opgam}{L^\infty} \leq C \ln(1 + \Nrm{\nabla \vr_\opgam}{L^\infty}) \, . 
		\end{equation*}
		Thanks to Lemma~\ref{lem:rho_inf} we also have the $W^{1,\infty}$ bound
		\begin{equation*}
			\Nrm{\nabla\vr_\opgam}{L^\infty} \leq \frac{C_d}{\hbar} \, \Big(1+ \Nrm{\n{\opp}^4\opgam}{\L^\infty}^{5/4}\Big) \, .
		\end{equation*}
		Moreover, arguing similarly as above, Lemma~\ref{lem:weight_4} and Lemma~\ref{lem:conv} yields that $\Nrm{\n{\opp}^4\opgam}{\L^\infty}$ is bounded uniformly in $\hbar$. Putting the last five bounds together, we get, for an appropriate constant $C>0$ independent of $\hbar$
		\begin{equation*}
			\hd \tr{\n{\com{\opp ,\opgam}}} \leq C \, \Big(\lambda + \frac{\hbar}{\lambda} \n{\ln \hbar}\Big) \, . 
		\end{equation*}
		This finishes the proof by taking $\lambda = \hbar \n{\ln \hbar}^{1/2} \geq \hbar$.
	\end{proof}
	
	\subsection{Weyl's law}
	
	Let us now turn to the proof of the local Weyl law in the non-linear setting. Recall the notation
	\begin{equation*}
		\scE_f = \intdd (\n{\xi}^2+ U (x)) + \frac{1}{2} \intd \vr_f\(K*\vr_f\)
	\end{equation*}
	for any function $f \in L^1$ for which $\Eps_f < \infty$. Recall as well $\vr_f(x) = \intd f(x,\xi) \d \xi$, $\widetilde f = G_\eps * f $ and $\widetilde \op = G_\eps \star \op$. As in the linear case, we also introduce the following classical and quantum errors
	\begin{align*} %\label{eq:nonlinear_Q1}
		Q(g) &:= \intdd \n{\cH_g} \n{\indic_{\cH_g \leq 0} - g}
		\\ %\label{eq:nonlinear_Q2}
		Q_\opgam(\op) &:= \hd \Tr{\n{H_\opgam} \n{\op -\opgam}^2 + \n{H_\opgam} (\op - \op^2)}
	\end{align*}
	for any classical function $ 0 \leq g \leq 1 $, and any operator $0 \leq \op \leq \id$. First, we establish the following general bound in analogy to the linear case. 

	\begin{prop}\label{prop:energy:nonlinear}
		Let $ 0 \leq f \leq 1 $ and $ 0 \leq \opgam \leq \id $ be minimizers of $\scE_f$ and $\scE_\opgam$, respectively, in $d \geq3$. Let $p = 1 + \frac{2}{d}$, $c_d = 1 / \omega^{2/d}$ and for any function $0 \leq g \leq 1$ denote 
		\begin{equation*}
			\cH_g(x,\xi) = \n{\xi}^2- c_d \,\vr_{g}^{p-1} \, .
		\end{equation*}
		Then, we have that
		\begin{equation*}
			Q_\opgam(\tildop_f) + Q(\tilde{f}_\opgam) + \frac{c_d}{p'} \intd \n{\vr_{\tildop}^{p/2} -\vr_f^{p/2}}^2 + \intd U _+ \vr_{\tildop} \leq \scE_{\tilde \opgam} - \scE_\opgam + \scE_{\tilde f} - \scE_f \, . 
		\end{equation*}
	\end{prop}

	\begin{proof}
		We split the proof in two parts. Let us recall first that for any $\op$ the quantitative variational principle in Lemma~\ref{prop:minmax} gives
		\begin{equation}\label{eq:energy0}
			Q_\opgam(\op) \leq \scE_\op - \scE_\opgam \, .
		\end{equation}
		Secondly, using $\vr_f = \omega_d \(U + K*\vr_f\)_-^{d/2}$ and considering now an arbitrary function $0 \leq g \leq 1$, one obtains
		\begin{align*}
			\scE_g - \scE_f &- \frac{1}{2} \intd \(\vr_g - \vr_f\) K* \(\vr_g-\vr_f\)
			\\
			&= \intdd \n{\xi}^2 \(g -f\)\d x\d\xi + \intd \(U+K*\vr_f\) \(\vr_g - \vr_f\)
			\\
			&= \intdd \n{\xi}^2 \(g -f\)\d x\d\xi + \intd \(U+K*\vr_f\)_+ \vr_g - c_d \intd \vr_f^\frac{2}{d} \(\vr_g - \vr_f\).
		\end{align*}
		Therefore, using the fact that $\intdd f\n{\xi}^2\d x\d\xi = \frac{c_d}{p} \intd \vr_f^p$ and $\widehat{K}\geq 0$, it gives
		\begin{equation}\label{eq:energy}
			\scE_g - \scE_f \geq
			\intd U _+ \vr_g + 
			\intdd \(\n{\xi}^2- c_d\,\vr_g^{p-1}\) g + \frac{c_d}{p} \intd p \,\vr_g^p - \vr_f^p - p\,\vr_f^{p-1} \(\vr_g - \vr_f\)
		\end{equation}
		where we also used $\(U + K*\vr_f\)_+ \vr_g \geq U _+ \vr_g$. Arguing as in the proof of Lemma~\ref{lem:energies} we use the "quantitative bathtub principle" as in Lemma~\ref{prop:minmax} to get
		\begin{equation*}
			\intdd \cH_g\, g = 
%			\intdd \cH_g (g - \indic_{\cH_g \leq 0}) + \intdd \cH_g \indic_{\cH_g \leq 0} = 
				\intdd \n{\cH_g} \n{g - \indic_{\cH_g \leq 0}} - \frac{c_d}{p'} \intd \vr_g^p \, . 
		\end{equation*}
		Combined with \eqref{eq:energy} and using $1-1/p' = 1/p$ we find 
		\begin{equation}\label{eq:energy2}
			\scE_g - \scE_f \geq \intd U _+ \vr_g + \intdd \n{\cH_g} \n{g - \indic_{\cH_g \leq 0}} + \frac{c_d}{p} \intd \vr_g^p -\vr_f^p - p\,\vr_f^{p-1} \(\vr_g - \vr_f\) .
		\end{equation}
		We finish the proof by taking $\op = \op_{\tilde f}$, $g = f_{\tilde \opgam}$, putting together \eqref{eq:energy0} and \eqref{eq:energy2}, and employing Young's inequality. 
	\end{proof}
	
	In order to connect our analysis back to the linear case we make the following observation. Since $\vr_{\tilde{f}} = g_\eps * \vr_f$, $\widehat{K} \geq 0$ and $0\leq \widehat{g}_\eps^2 \leq 1$, it follows from Plancherel formula that
	\begin{equation*}
		\intd \vr_{\tilde f} \,K*\vr_{\tilde f} = \intd \widehat{K}\,\widehat{g}_{\eps}^2 \n{\widehat{\rho}_f}^2 \leq \intd \vr_f \,K*\vr_f
	\end{equation*}
	and similarly for $\vr_\opgam$. Therefore, arguing as in the linear case, one finds 
	\begin{align*}
		\scE_{\tilde f} - \scE_f &\leq \frac{d \,h^2}{8\pi\, \eps} M_f + \intd \(g_\eps * U - U\) \vr_f
		\\
		\scE_{\tildop} - \scE_\opgam &\leq \frac{d \,h^2}{8\pi\, \eps} M_\opgam + \intd \(g_\eps * U - U\) \vr_\opgam \, ,
	\end{align*}
	with $M_f = \intdd f$ and $M_\opgam = \hd\Tr{\opgam}$. The expression on the right-hand side was analyzed in detail in Lemma~\ref{lem:energies} for $V = U \in W^{2,q}(\Omega_1)$ where $\Omega_1 = \{x : \Dist{x , \{U \leq 0 \}} \leq 1\}$, which led to Proposition~\ref{prop:energy2:linear}. We record here the analogous result.

	\begin{lem}\label{lemma:weyl:NL}
		Under the same conditions of Proposition~\ref{prop:energy:nonlinear}
		\begin{equation*}
			Q_\opgam(\tildop_f) + Q(\tilde{f}_\opgam) + \frac{c_d}{p'} \intd \n{\vr_{\tildop}^{p/2} -\vr_f^{p/2}}^2 
			\leq \sM \,\frac{h^2}{\eps} + \D' \eps \, . 
		\end{equation*}
		Here, we employ the following notations for any $\beta \geq 0$ 
		\begin{align*} 
			\sM &= \frac{d}{8\pi} \(M_\opgam + M_f\)
			\\
			\D' &= \tfrac{d}{\pi} \Nrm{U}{W^{2,1}(\Omega_1)}
			\(\Nrm{\vr_f}{L^\infty} + \Nrm{\vr_\opgam}{L^{\infty}}\) + C_{d , \beta} \Nrm{e^{- \beta \n{x}} \nabla U_+}{L^\infty} \Nrm{e^{\beta \n{x}} \nabla \vr_f}{L^1}
		\end{align*}
		where $C_{d, \beta} = \frac{1}{2} \intd \n{x}^2 e^{-\pi \n{x}^2 + \beta \n{x}} \d x$.
	\end{lem}

	With these bounds, we can replicate the proof in the linear case and obtain the proof of Theorem~\ref{thm:weyl:nonlinear}.

	\begin{proof}[Proof of Theorem~\ref{thm:weyl:nonlinear}]
		For the proof, denote for simplicity $f := \fTF$ and $\opgam := \opHF$ the minimizers of the classical and quantum energies, respectively. We let $K$ and $U$ satisfy hypotheses \eqref{hyp:K} and \eqref{hyp:U}, respectively, with parameter $\beta>0$. Without loss of generality, we may assume $\hbar \leq 1 / 8 \beta $. Otherwise we can proceed as in the proof of Theorem~\ref{thm:commutator:linear} and modify the value of the overall constant.
		
		First, we prove the estimates for the densities. Starting from Lemma~\ref{lemma:weyl:NL} we follow the proof of Proposition~\ref{prop:linear_local_Weyl_law}. We find that there exists $C >0$ such that 
		\begin{align}\label{eq:dens:2}%\label{eq:linear_Weyl_rho_L2}
			\Nrm{\vr_f-\vr_\opgam}{L^2} &\leq C \, h^{1/3} \(1 + \Nrm{\vr_\opgam}{\dot{B}^{1/2}_{2,\infty}}\)^{2/3} \(1 + \sL_2 \, \sM^{1/2} + \sL_2\, \D'^{1/2} h^{1/3}\)
			\\ \nonumber %\label{eq:linear_Weyl_rho_L1} %\label{dens:3}
			\Nrm{\vr_f-\vr_\opgam}{L^1} &\leq C \, h^{1/2}\(1 +\Nrm{\nabla \vr_\opgam}{L^1}^{1/2}+ \D'^{1/2}\) \(1 + \sL_1 + \sM^{1/2}\)
		\end{align}
		where $\sL_1$ and $\sL_2$ are defined as in \eqref{eq:L1} and \eqref{eq:L2}. All quantities $\sM$, $\D'$, $\sL_1$ and $\sL_2$ can be bounded uniformly in $\hbar$ thanks to either $L^p$ or Agmon's estimates. The desired bound then follows by the different estimates on $\Nrm{\nabla \vr_\opgam}{L^1}$ that arise in the cases $0 \leq a < 1$ and $a=1$. See e.g. Remark~\ref{rem:rho}.
		
		For the convergence of states, start from Lemma~\ref{lemma:weyl:NL} and follow the proof Proposition~\ref{prop:linear_local_Weyl_law2}. We obtain, in terms of the Husimi measure $m_\opgam$, that for some constant $C > 0$
		\begin{align*} 
			\Nrm{f - m_\opgam}{L^1} &\leq h^{1/2} \(\(\mathsf L_1 + \sC_1\) \(\sM + \D'\) + \sC_2\)
			\\ 
			\Nrm{f - f_\opgam}{L^2} &\leq C \, h^{1/4} \Nrm{f}{\dot B_{2,\infty}^{1/2}} + \Nrm{f - m_\opgam}{L^1}^{\frac{1}{2}}
			\\
			\Nrm{\op_f-\opgam}{\L^1} &\leq C \, h^{1/2} \(\Nrm{\Dh\opgam}{\L^1} + \Nrm{\nabla f}{\cM} \) + \Nrm{f - m_\opgam}{L^1}
		\end{align*}
		where $\sC_1$ and $\sC_2$ are defined in \eqref{C1} and \eqref{C2} in terms of $\opgam$ and $U(x)$. 
		They are uniformly bounded in $\hbar$ thanks to $L^p$ estimates. 		Arguing as in the linear case, 
		we know that 
		$f \in \dot B_{2, \infty}^{1/2}(\R^6)$ and $\nabla f \in \cM(\R^6)$.
		Thus, we now use the estimates for the quantum gradient $\Nrm{\Dh \opgam}{\L^1}$ that follow from Theorem~\ref{thm:commutator_nonlinear}.
	
		Finally, we note that the potential 
		\begin{equation*}
			V_{\opgam} = U (x) + K * \vr_{\opgam} 
		\end{equation*}
		verifies Hypothesis~\eqref{hyp:V4}. To see this, first, thanks to Lemma~\ref{lem:weight_4}, we obtain that $\sNrm{\n{\opp}^4 \opgam}{\L^\infty}$ is bounded uniformly in $\hbar$, thanks to Lemma~\ref{lem:conv}, and Hypothesis~\eqref{hyp:V4}. Therefore, thanks to Lemma~\ref{lem:rho_inf} we obtain
		\begin{equation*}
			\Nrm{\nabla \vr_{\opgam}}{L^\infty} \leq C / \hbar 
			\quad \text{ and } \quad \Nrm{\nabla^2 \vr_{\opgam}}{L^\infty} \leq C / \hbar^2 \, . 
		\end{equation*}
		Thus, similarly, we can now use Lemma~\ref{lem:weight_6} to obtain
		\begin{equation*}
			\Nrm{\n{\opp}^6 \opgam}{\L^\infty} \leq C \, . 
		\end{equation*} 
		We conclude that all the moments in Lemma~\ref{lem:rho_inf} are uniformly bounded in $\hbar$. Let us set $W_\opgam = \n{x}^{-1} * \vr_{\opgam}$. Thanks to well-known elliptic-type estimates satisfied by the Coulomb potential, recorded in Lemma~\ref{lem:conv}, we obtain
		\begin{equation*}
			\Nrm{W_\opgam}{L^\infty} \leq C \, , 
			\quad 
			\Nrm{\nabla W_\opgam}{L^\infty} \leq C \, , 
			\quad 
			\Nrm{\nabla^2 W_\opgam}{L^\infty} \leq C 
			\(1 + \n{\ln\hbar}\) , 
		\end{equation*}
		in view of $\ln \Nrm{\nabla \vr_{\opgam}}{L^\infty}
		\leq C \(1 + \n{\ln \hbar}\) $. Similarly, for higher derivatives 
		\begin{equation*}
			\nabla^n W_\opgam = \nabla^2(\n{x}^{-1} * \nabla^{n -2} \vr_{\opgam}) \, . 
		\end{equation*}
		and so for $n = 5$
		\begin{equation*}
			\n{\nabla^5 W_\opgam} \leq C \(1 + \Nrm{\nabla^3 \vr_{\opgam}}{L^1} + \Nrm{\nabla^3 \vr_{\opgam}}{L^\infty} \(1 + \ln(1 + \Nrm{\nabla^3 \vr_{\opgam}}{L^\infty})\)\) \leq C\,\frac{1 + \n{\ln\hbar}}{\hbar^3} \, . 
		\end{equation*}
		By assumption, the external trap $U(x)$ verifies the same estimate for $ n = 2$ and $n =5$. We conclude the same is true for $V_{\opgam}$. Finally, it now suffices to interpolate between $2 \leq n \leq 5$. Therefore, we use Proposition~\ref{prop:regu_rho_L2} to conclude that $\Nrm{\vr_{\opgam}}{\dot B_{2 , \infty}^{1/2}}$ is uniformly bounded in $\hbar$. We plug this estimate back in~\eqref{eq:dens:2}. This gives the desired estimate and finishes the proof.
	\end{proof}

\appendix
\section{Another proof of the Hilbert--Schmidt commutator estimate}
\label{appendix}

	We give here another proof of a Hilbert--Schmidt commutator estimate which does not use the decomposition of $\opgam$ under eigenspaces.

	\begin{prop}
		Let $\opgam := \indic_{H\leq 0}$ and $A$ be an operator such that $A\,\opgam$ and $A^*\,\opgam$ are trace class. Then for any $\lambda>0$,
		\begin{equation*}
			\Tr{\n{\com{A,\opgam}}^2}^\frac{1}{2} \leq \frac{B + \sqrt{B^2 + 4 \, D^2}}{2} \leq B + D,
		\end{equation*}
		with
		\begin{align*}
			B &= \Tr{\n{\indic_{H \geq \lambda} \com{H,A^*}}^2 \(H-\lambda\)^{-2} \opgam}^\frac{1}{2} + \Tr{\n{\indic_{H \geq \lambda} \com{H,A}}^2 \(H-\lambda\)^{-2} \opgam}^\frac{1}{2}
			\\
			D &= \Tr{\(\n{\opgam\, A}^2 + \n{\opgam\, A^*}^2\) \indic_{0< H < \lambda}}^\frac{1}{2}.
		\end{align*}
		In the particular case when $A$ is a normal operator, then the same inequality holds with
		\begin{align*}
			B &= 2\,\Tr{\n{\indic_{H \geq \lambda} \com{H,A}}^2 \(H-\lambda\)^{-2} \opgam}^\frac{1}{2}
			\\
			D &= 2\,\Tr{\n{\opgam\, A}^2\indic_{0< H < \lambda}}^\frac{1}{2}.
		\end{align*}
	\end{prop}
	
	\begin{proof}
		Using the cyclicity of the trace and the fact that $\opgam^2=\opgam$ gives
		\begin{equation*}
			\Tr{\n{\com{A,\opgam}}^2} = \Tr{\(\id - \opgam\)A^*\,\opgam\, A} + \Tr{\(\id - \opgam\)A\,\opgam\, A^*} =: I_1 + I_2 \, .
		\end{equation*}
		Let us look at $I_1$. It holds
		\begin{equation*}
			I_1 = \Tr{A^*\,\opgam\,A\, \indic_{0< H < \lambda}} + \Tr{A^* \, \opgam\, A \, \indic_{H \geq \lambda}} .
		\end{equation*}
		Since $\lambda>0$, $\opgam \(H-\lambda\)^{-1}$ is well-defined by functional calculus and so
		\begin{multline*}
			\Tr{A^* \, \opgam\, A \, \indic_{H \geq \lambda}} = \Tr{A^*\,\opgam \(H-\lambda\)^{-1} \(H-\lambda'\) A\,\indic_{H \geq \lambda}}
			\\
			= \Tr{A^*\,\opgam \(H-\lambda\)^{-1} \com{H,A}\indic_{H \geq \lambda}} + \Tr{A^*\,\opgam \(H-\lambda\)^{-1} A\, \(H-\lambda\) \indic_{H \geq \lambda}} .
		\end{multline*}
		Since $\opgam \(H-\lambda\)^{-1} \leq 0$ and $\(H-\lambda\) \indic_{H \geq \lambda} \geq 0$, the last term of the above equation is nonpositive, which leads to
		\begin{equation*}
			\Tr{A^* \, \opgam\, A \, \indic_{H \geq \lambda}} \leq \Tr{A^*\,\opgam \(H-\lambda\)^{-1} \com{H,A} \indic_{H \geq \lambda}} .
		\end{equation*}
		Moreover, since $\opgam= \opgam^2$ and $\opgam\, \indic_{H \geq \lambda} = 0$, one observes that $\indic_{H \geq \lambda} \,A^*\, \opgam = \indic_{H \geq \lambda} \com{A^*, \opgam}\opgam$, and so
		\begin{align*}
			\Tr{A^* \, \opgam\, A \, \indic_{H \geq \lambda}} &\leq \Tr{\com{A^*, \opgam}\opgam \(H-\lambda\)^{-1} \com{H,A} \indic_{H \geq \lambda}}
			\\
			&\leq \Tr{\n{\com{\opgam,A}}^2}^\frac{1}{2} \Tr{\n{\indic_{H \geq \lambda} \com{H,A^*}}^2 \(H-\lambda\)^{-2} \opgam}^\frac{1}{2},
		\end{align*}
		where the last inequality follows from the Cauchy--Schwarz inequality and the cyclicity of the trace. The same inequality holds for $I_2$ up to replacing $A$ by $A^*$, which yields finally
		\begin{equation*}
			\Tr{\n{\com{A,\opgam}}^2} \leq \Tr{\n{\com{\opgam,A}}^2}^\frac{1}{2} B + D\, .
		\end{equation*}
		Solving the above inequality finishes the proof.
	\end{proof}

\bigskip

\paragraph{\textbf{Acknowledgments}.} E.C. gratefully acknowledges support from NSF under grants No. DMS-2009549 and DMS-2052789 through Nata\v sa Pavlovi\'c.

%% ********************  Bibliographie  ********************

\renewcommand{\bibname}{\centerline{Bibliography}}
\bibliographystyle{abbrv} % apalike, ieee, plain, alpha, unsrt, abbrv
\bibliography{Vlasov}

\end{document}